\documentclass{lmcs}
\pdfoutput=1
\usepackage[utf8]{inputenc}

\usepackage{lastpage}
\lmcsdoi{21}{1}{18}
\lmcsheading{}{\pageref{LastPage}}{}{}%
{Dec.~12,~2023}{Feb.~21,~2025}{}


\usepackage{hyperref}
\usepackage{xspace}
\usepackage{cleveref}
\usepackage{amssymb}
\usepackage{yhmath} 
\usepackage{bbding} 
\usepackage{algorithm}
\usepackage{algorithmic}
\usepackage[xcolor, hyperref, cleveref, notion, quotation, electronic]{knowledge}
\usepackage{mathcommand}
\knowledgeconfigure{quotation, protect quotation={tikzcd}}
\knowledgeconfigure{diagnose line=true, diagnose bar=true}

\definecolor{Dark Ruby Red}{HTML}{7c1b1e}
\definecolor{Dark Blue Sapphire}{HTML}{004452} 
\definecolor{Dark Gamboge}{HTML}{be7c00}

\IfKnowledgePaperModeTF{
}{
    \knowledgestyle{intro notion}{color={Dark Ruby Red}, emphasize}
    \knowledgestyle{notion}{color={Dark Blue Sapphire}}
    \hypersetup{
        colorlinks=true,
        breaklinks=true,
        linkcolor={Dark Blue Sapphire}, 
        citecolor={Dark Blue Sapphire}, 
        filecolor={Dark Blue Sapphire}, 
        urlcolor={Dark Blue Sapphire},
    }
    \IfKnowledgeElectronicModeTF{
    }{
        \knowledgeconfigure{anchor point color={Dark Ruby Red}, anchor point shape=corner}
        \knowledgestyle{intro unknown}{color={Dark Gamboge}, emphasize}
        \knowledgestyle{intro unknown cont}{color={Dark Gamboge}, emphasize}
        \knowledgestyle{kl unknown}{color={Dark Gamboge}}
        \knowledgestyle{kl unknown cont}{color={Dark Gamboge}}
    }
}

\knowledge{notion}
 | notion@notice
 | definition@notice

\knowledge{text={i.e.}}
 | ie

 \knowledge{notion}
| $r_{\mid{i+1}}$

\knowledge{notion}
 | database
 | databases
 | Database

\knowledge{notion}
 | repair
 | repairs
 
\knowledge{notion}
 | atom
 | atoms

 \knowledge{notion}
 | partial repair

\knowledge{notion}
 | relational signature
 | signature

 \knowledge{notion}
 | $R$-fact
 | $R$-facts

 \knowledge{notion}
 | fact
 | facts

\knowledge{notion}
 | Boolean conjunctive query
 | conjunctive query
 | query
 | boolean conjunctive query

\knowledge{notion}
 | primary keys
 | primary key constraint
 | primary key constraints
 | Primary key constraints
 | primary key

 \knowledge{notion}
 | key constraint
 | key constraints
 
\knowledge{notion}
 | $\Gamma$-equivalent
 | $\Gamma$-equivalence
 | equivalent

\knowledge{notion}
 | $\Gamma$-blocks
 | block
 | blocks

\knowledge{notion}
 | solution
 | solutions

\knowledge{notion}
 | match 
 | matches
 | matched
 | matching

\knowledge{notion}
 | certain
 | certainty

\knowledge{notion}
 | self-join-free
 | Self-join-free
 | self-join-free query

\knowledge{notion}
 | self-join
 | self-join query

\knowledge{notion}
 | path query
 | path queries
 | Path query
 | Path queries
 | path
 | Path
 | paths
 | Paths

\knowledge{notion}
 | $\Gamma$-derivation
 | $\Gamma$-derivations
 | derivation
 | derivations

\knowledge{notion}
 | $\Gamma$-determined
 | $\Gamma$-determinacy

\knowledge{notion}
 | mutually $\Gamma$-determined
 | mutual $\Gamma$-determinacy

\knowledge{notion}
 | stable
 | stable set
 | stable sets

\knowledge{notion}
 | $\Gamma$-sequence

\knowledge{notion}
 | $i$-minimal
 | $0$-minimal
 | $(i+1)$-minimal
 | $i$-minimality

 \knowledge{notion}
 | minimal

\knowledge{notion}
 | $i$-compatible
 | $(i+1)$-compatible

\knowledge{notion}
 | strong $i$-minimal
 | strong $i$-minimality

\knowledge{notion}
 | connected
 | connectedness

\knowledge{notion}
 | valid path

 \knowledge{notion}
 | $k$-set
 | $k$-sets

 \knowledge{notion}
 | bounded
 | boundedness

 \knowledge{notion}
 | attack graph

\knowledge{notion}
 | weak attack

\knowledge{notion}
 | weak cycle
 | strong cycle
 | weak cycles

\knowledge{notion}
 | core graph

 \knowledge{notion}
 | attacks
 | attacking
 | attack
 | attacked

\knowledge{notion}
 | solution graph
 | Solution graph
 | solution graphs

\knowledge{notion}
 | $k$-obstruction set

\knowledge{notion}
 | triangle
 | triangles
 | Triangles

 \knowledge{notion}
 | Saturating Bipartite Matching problem
 | SBM problem
 
 \knowledge{notion}
 | Bipartite Matching problem
 | BPM problem

\knowledge{notion}
 | $V_1$-saturating matching

\knowledge{notion}
 | self-loops
 | self-loop

\knowledge{notion}
 | conjecture
 | Dichotomy conjecture
 | dichotomy conjecture
 | Dichotomy Conjecture

\knowledge{notion}
 | second derivation rule

\knowledge{notion}
 | first derivation rule

\usepackage{soul}
\newcommand{\eg}{\textit{e.g.}}
\newcommand{\cf}{\textit{cf.}}

\renewcommand{\epsilon}{\varepsilon}

\newcommand{\tup}[1]{\langle #1 \rangle}
\knowledgenewrobustcmd{\vars}{\cmdkl{\textit{vars}}}
\knowledgenewrobustcmd{\dcup}{\mathrel{\cmdkl{\dot{\cup}}}}
\knowledgenewrobustcmd{\key}{\cmdkl{\textit{key}}}
\newcommand{\ql}[1]{\ensuremath{q_{\leq #1}}}
\knowledgenewrobustcmd{\certain}{\cmdkl{\ensuremath{\textsf{\textup{certain}}}}}
\knowledgenewrobustcmd{\simblock}[1][]{\mathrel{\cmdkl{\ensuremath{\sim_{#1}}}}}
\newcommand\conp{\textup{\sc coNP}\xspace}
\newcommand\ptime{\textup{\sc PTime}\xspace}
\newcommand\nl{\textup{\sc NL}\xspace}
\newcommand\logspace{\textup{\sc LogSpace}\xspace}
\newcommand\fo{\textup{FO}\xspace}
\newcommand\cert{\textup{Cert}}
\knowledgenewrobustcmd{\midi}[2]{\cmdkl{#1_{|#2}}}
\knowledgenewrobustcmd{\Cqk}[1][k]{\cmdkl{\ensuremath{\cert_{#1}}}}
\knowledgenewrobustcmd{\Cqkp}[1][k]{\cmdkl{\ensuremath{\cert^{\!+}_{#1}}}}
\knowledgenewrobustcmd{\Deltak}[1][k]{\cmdkl{\ensuremath{\Delta_{#1}}}}
\knowledgenewrobustcmd{\Deltakp}[1][k]{\cmdkl{\ensuremath{\Delta^{\!+}_{#1}}}}
\knowledgenewrobustcmd{\Deltakpi}[1][k]{\cmdkl{\ensuremath{\Delta^{\!+}_{#1}}}}
\knowledgenewrobustcmd{\Deltaki}[1][k]{\cmdkl{\ensuremath{\Delta_{#1}}}}
\knowledgenewrobustcmd{\psii}[1][i,k]{\cmdkl{\ensuremath{\psi_{#1}}}}
\knowledgenewrobustcmd{\solgraph}{\cmdkl{\ensuremath{\mathsf{G}_D}}}

\newcommand\set[1]{\ensuremath{\{#1\}}}
\knowledgenewrobustcmd{\HP}{\cmdkl{\textup{PCond}}}
\knowledgenewrobustcmd{\HPtau}[1][\tau]{\cmdkl{\textup{PCond}_{#1}}}
\knowledgenewrobustcmd{\HPtaui}[2]{\cmdkl{\textup{PCond}_{#1}(#2)}}
\knowledgenewrobustcmd{\HPtauA}[2]{\cmdkl{\textup{PCond}_{#1}(#2)}}
\knowledgenewrobustcmd{\FactorCond}{\cmdkl{\textup{FactorCond}}}

\newcommand\balpha{\bar\alpha}
\newcommand\bbeta{\bar\beta}
\newcommand\bdelta{\bar\delta}
\newcommand\bu{\bar u}
\newcommand\bv{\bar v}

\newcommand\bc{\bar c}

\knowledgenewrobustcmd{\Lq}{\cmdkl{\ensuremath{\mathcal{L}^\looparrowright(q)}}}
\knowledgenewrobustcmd{\trace}{\cmdkl{\ensuremath{\textit{trace}}}}
\knowledgenewrobustcmd{\last}{\cmdkl{\ensuremath{\textit{last}}}}
\knowledgenewrobustcmd{\aut}[1][q]{\cmdkl{\ensuremath{\mathcal{A}^{#1}}}}

\newrobustcmd{\defeq}{\mathrel{\hat{=}}}

\knowledgenewrobustcmd{\assoc}[1]{\cmdkl{\widehat{#1}}}

\knowledgenewrobustcmd{\adom}{\cmdkl{\textit{adom}}}

\newcommand\NL{\text{NL}}

\newcommand{\stylequery}[1]{\mathsf{#1}}
\knowledgenewrobustcmd{\qOne}{\cmdkl{\stylequery{q_1}}}
\knowledgenewrobustcmd{\qTwo}{\cmdkl{\stylequery{q_2}}}
\knowledgenewrobustcmd{\qTwop}{\cmdkl{\stylequery{q'_2}}}
\knowledgenewrobustcmd{\qThree}{\cmdkl{\stylequery{q_3}}}
\knowledgenewrobustcmd{\qThreep}{\cmdkl{\stylequery{q'_3}}}
\knowledgenewrobustcmd{\qFour}{\cmdkl{\stylequery{q_4}}}
\knowledgenewrobustcmd{\qFive}{\cmdkl{\stylequery{q_5}}}

\knowledgenewrobustcmd{\Aplus}[1][A]{\cmdkl{{#1}^+}}
\knowledgenewrobustcmd{\Ind}{\cmdkl{\textup{\textsc{Ind}}}}
\knowledgenewrobustcmd{\Indpi}[1][i]{\cmdkl{\ensuremath{\textup{\textsc{Ind}}}^+_{#1}}}

\definecolor{light-gray}{gray}{0.9}
\newcommand{\proofcase}[1]{\noindent\colorbox{light-gray}{#1}~}
\theoremstyle{plain} 

\graphicspath{{./img/}} 
\newcommand{\exampleqedsymbol}{{$\triangle$}}
\AtBeginEnvironment{example}{%
  \pushQED{\qed}\renewcommand{\qedsymbol}{\exampleqedsymbol}%
}
\AtEndEnvironment{example}{\popQED\endexample}
\AtBeginEnvironment{exa}{%
  \pushQED{\qed}\renewcommand{\qedsymbol}{\exampleqedsymbol}%
}
\AtEndEnvironment{exa}{\popQED\endexample}

\theoremstyle{defC}
\newtheorem{clmC}[thm]{Claim}

\crefname{thm}{Theorem}{Theorems}
\crefname{thmC}{Theorem}{Theorems}

\crefname{defi}{Definition}{Definitions}

\let\endexample\endexa
\crefname{exa}{Example}{Examples}

\crefname{rem}{Remark}{Remarks}

\crefname{obs}{Observation}{Observations}

\crefname{cor}{Corollary}{Corollaries}

\crefname{lem}{Lemma}{Lemmata}
\crefname{lemC}{Lemma}{Lemmata}

\crefname{prop}{Proposition}{Propositions}
\crefname{propC}{Proposition}{Propositions}

\crefname{clm}{Claim}{Claims}
\crefname{clmC}{Claim}{Claims}
\crefname{fact}{Fact}{Facts}
\crefname{nota}{Notation}{Notations}
\begin{document}

\title[A Simple Algorithm for CQA under Primary Keys]{A Simple Algorithm for Consistent Query\texorpdfstring{\\}{} Answering Under Primary Keys}
\titlecomment{Journal version of the paper presented at ICDT $2023$ \cite{ourICDTpaper}, see \Cref{sec:delta-conference} for a summary of the added material.}

\author[D. Figueira]{Diego Figueira}[a]
\author[A. Padmanabha]{Anantha Padmanabha}[b]
\author[L. Segoufin]{Luc Segoufin}[c]
\author[C. Sirangelo]{Cristina Sirangelo}[d]

\address{Univ. Bordeaux, CNRS,  Bordeaux INP, LaBRI, UMR 5800, Talence, France}	
\email{diego.figueira@cnrs.fr}  

\address{Indian Institute of Technology Madras, Chennai, India}	
\email{ananthap@cse.iitm.ac.in}  

\address{INRIA, ENS-Paris, PSL University, France}	
\email{luc.segoufin@inria.fr}  

\address{Université Paris Cité, CNRS, IRIF, F-75013, Paris, France}	
\email{cristina@irif.fr}  




\begin{abstract}

We consider the dichotomy conjecture for consistent query answering under
primary key constraints.
It states that, for every fixed boolean conjunctive query $q$, testing whether $q$
is certain ("ie", whether it evaluates to true over all repairs of a given inconsistent database) is either
polynomial time or \conp-complete. This conjecture has been verified for
self-join-free and path queries.

We propose a simple inflationary fixpoint algorithm for consistent query
answering 
which, for a given database, naively computes a set $\Delta$ of
subsets of facts of the database of size at most $k$,
where $k$ is the size of the query $q$. The algorithm runs in polynomial time and can be formally defined as:
\begin{enumerate}
    \item Initialize $\Delta$ with all sets $S$ of at most $k$ facts such that $S\models q$.
    \item Add any set $S$ of at most $k$ facts to $\Delta$ if there exists a block  $B$ (i.e., a maximal set of facts sharing the same key) such that 
    for every fact $a\in B$ 
    there is a set $S'\subseteq S \cup \set{a}$ such that $S' \in \Delta$.
\end{enumerate}
For an input database $D$, the algorithm answers ``$q$ is certain'' if{f} $\Delta$ eventually contains the empty set. 
The algorithm correctly computes certainty when the query $q$ falls in the polynomial time cases of the known dichotomies for self-join-free queries and path queries. For arbitrary boolean conjunctive queries, the algorithm is an under-approximation: the query is guaranteed to be certain if the algorithm claims so. 
However, there are polynomial time certain queries (with self-joins) which are not identified as such by the algorithm.


\end{abstract}

\maketitle

\bigskip


\noindent
\raisebox{-.4ex}{\HandRight}\ \ This pdf contains internal links: clicking on a "notion@@notice" leads to its \AP ""definition@@notice"".\footnote{\url{https://ctan.org/pkg/knowledge}}

\clearpage\tableofcontents

\section{Introduction}
\label{section-Introduction}


A "database" often comes with integrity constraints. Such constraints are helpful in many ways, for instance in order to help optimizing query evaluation. When the "database" violates its integrity constraints we are faced with several possibilities. A first possibility is to clean the data until all integrity constraints are satisfied. This task is not easy as it is inherently non-deterministic: there could be many equally good ways to ``"repair"'' a "database". A "repair" can be understood as a minimal way to change the "database" in order to satisfy the constraints.

Another possibility is to keep the "database" in its inconsistent state,
postponing the problem until a query is issued. In order to
evaluate the query on the inconsistent "database" $D$, the classical solution is to consider all possible "repairs"
of $D$ and return all the answers which are ``certain'', "ie", the answers that
are returned by the query when evaluated on \emph{every} "repair" of $D$ \cite{DBLP:conf/pods/ArenasBC99}. However, this method usually has an impact on the complexity of the query
evaluation problem. The impact will of course depend on the type of the integrity constraints and on the definition of a "repair", but most often the worst case complexity
increases at least by a factor which is exponential in the size of the "database", since there could be exponentially many
ways to "repair" a "database".

Depending on the type of integrity constraints, what should be considered as a ``good'' notion of "repair" may be
controversial. In this paper we consider "primary key constraints", which are arguably the most common kind of integrity constraints in "databases". For "primary keys", there is a unanimously accepted
notion of "repair". 
"Primary key constraints" identify, for each relation, a set of attributes which are considered to be the \emph{key} of this relation. 
An inconsistent "database" is therefore a "database" that has
distinct tuples sharing the same key within a relation. For such constraints,
the standard notion of a "repair" is any maximal subset of the "database" satisfying all
the "primary key constraints".  This amounts to keeping exactly one tuple for each group of tuples having the same key, in each relation. A simple analysis shows that there can be
exponentially many "repairs" of a given "database", and therefore a naive
evaluation algorithm would have to evaluate the query on each of these
exponentially many "repairs".

As query language, we consider boolean conjunctive queries, which can be evaluated in polynomial time in data complexity. 
With the
certain answer semantics described above, a query is ``certain'' on an inconsistent
"database" if it is true on all its "repairs". The data complexity of certain answers for conjunctive
queries over inconsistent "databases" in the presence of "primary key constraints" is
therefore in \conp. Indeed, in order to test whether the query is not certain, it is enough to guess a subset of the "database" which is a "repair" and which makes the query false. Further, it has been observed that for some conjunctive queries the certain query answering problem is
\conp-hard~\cite{DBLP:journals/jcss/FuxmanM07} while, for other queries, it can be solved in polynomial time. 
\AP
The main "conjecture"  for inconsistent "databases" in the presence of "primary key constraints"
is that there are no intermediate cases: for any boolean conjunctive query, the certain answering problem is either
solvable in polynomial time or it is  \conp-complete.

The "conjecture" has been proved for "self-join-free" boolean conjunctive
queries~\cite{DBLP:journals/tods/KoutrisW17} and for "path" queries~\cite{DBLP:conf/pods/KoutrisOW21}. 
However, the "conjecture" remains open for arbitrary conjunctive queries (with self-joins).
In this paper we revisit the two cases above where the "conjecture" is known to hold: "self-join-free" queries and "path queries".

\subsection{Contributions} Our main contribution is the design of a simple
fixpoint algorithm for computing certain answers of boolean conjunctive queries over inconsistent
"databases" in the presence of "primary key constraints".  For every $k \ge 1$,
we describe a fixpoint algorithm parameterized by $k$.  The algorithm is always
an under-approximation of the certain answers: on boolean queries, if it
outputs `yes' then the query is certain, "ie", it is true on all "repairs" of
the "database". But there could be ``false negatives'', that is, queries which are certain but on which the algorithm outputs `no'.

In this paper we investigate the expressive power of our fixpoint algorithm,
trying to understand when it solves the certain answering problem.

Our first result shows that in the case of "self-join-free" queries and "path"
boolean queries, we can characterize the cases when our fixpoint algorithm
computes the certain answers via a semantic condition. In other words, when the
condition holds there exists $k$ (namely, the number of atoms of the query)
such that the corresponding fixpoint algorithm correctly computes the certain
answer (taking $k$ as parameter); and conversely when the condition does not
hold, for every $k$ there exists a database instance $D$ such that the fixpoint
algorithm parameterized by $k$ outputs a false negative on $D$.

We then show that when our algorithm fails to compute the certain answers, "ie",
when our semantic condition fails, the certain answering problem is
actually \conp-hard. Hence, for "self-join-free" queries and "path queries" and
assuming $\ptime \neq \conp$, our simple fixpoint algorithm solves the certain
answering problem in all the cases where it is solvable in polynomial time. The
current approaches for \ptime solvable queries for these two classes are
mutually orthogonal and our result provides a uniform algorithm to solve all
the polynomial time solvable queries known in the literature.

A natural question is then to wonder whether our algorithm always correctly
computes the certain answering problem on all queries for which this problem is
polynomial time computable. Our second result answers negatively to this
question.  There is a simple two-atom query, named $\qFour$ in the paper, with
self-joins, whose certain answering problem is equivalent to bipartite matching
under \logspace reductions, and cannot be solved with our fixpoint
algorithm. Recall that the bipartite matching problem can be solved in
polynomial time and is \nl-hard.  This shows that in the presence of
self-joins, the classification of the complexity of the certain answering
problem is richer than in the self-join free case when it is either in
\logspace or \conp-complete~\cite{DBLP:journals/mst/KoutrisW21} .

Our fixpoint algorithm is based on a function that is expressible in
first-order logic (\fo). Hence,
when the fixpoint is bounded, the certain answering problem is expressible in
\fo. Our last result shows that for "self-join-free" and "path queries" the
converse is true: whenever the certain answering problem is expressible in
\fo, the fixpoint algorithm is bounded and the certain answering problem is
expressible by some bounded unfolding of the fixpoint algorithm.

\medskip

Though our greedy fixpoint computation algorithm is simple, the proof of
correctness when the semantic condition holds is non-trivial. In the case of "self-join-free"
queries satisfying our semantic condition, to prove that the algorithm always
computes the correct answer, we proceed by contradiction: if the algorithm fails to give the
correct answer, we use the fixpoint definition of the algorithm in order to
produce 
an infinite sequence of distinct facts
of the "database", contradicting its finiteness.

The situation is a bit simpler in the case of "path queries", where we show that our
fixpoint algorithm can simulate the polynomial time algorithm
of~\cite{DBLP:conf/pods/KoutrisOW21} for computing certain answers for a path query $q$,
assuming that certain answering for $q$ is polynomial time solvable.

For the lower bounds, we first show that our fixpoint algorithm fails to
compute the certain answering problem for $\qFour$ by constructing for every
number $k$, a database such that all its repairs satisfy $\qFour$ but the algorithm outputs `no'.  In the second step we reduce this
query $\qFour$ to all queries that falsify the semantic condition: if the fixpoint
algorithm would work for such queries, it would also work for $\qFour$. The
reduction relies on the (syntactic) condition of~\cite{DBLP:journals/tods/KoutrisW17}
characterizing the class of queries having a \conp-complete certain answering
problem.

\subsection{Related work} Our work is inspired by the results of
Koutris and
Wijsen~\cite{DBLP:journals/mst/KoutrisW21,DBLP:journals/tods/KoutrisW17}.  For
"self-join-free" queries, the authors prove the polynomial time case via a long
sequence of reductions eventually producing a simple query whose certain
answers can be solved efficiently. When unfolding the sequence of reductions
this gives a complicated polynomial time algorithm with a complex proof of
correctness. We have basically simplified the algorithm and pushed all the
difficulty into the proof of correctness. Our algorithm is simple, but the
proof of correctness is arguably as complex as theirs. Further, our algorithm
does not give, a priori, the optimal \logspace complexity result
of~\cite{DBLP:journals/mst/KoutrisW21} as we know that some of the path queries
that can be solved with our algorithm are \ptime
complete~\cite{DBLP:conf/pods/KoutrisOW21}.  The semantic condition that we
provide for characterizing the polynomial case in the "self-join-free" case can
be effectively tested, but not efficiently, unlike the simple syntactic
characterization of~\cite{DBLP:journals/tods/KoutrisW17} based on the so-called
``"attack graph"'' of the query.

In the case of "path queries", \cite{DBLP:conf/pods/KoutrisOW21} also provides
a simple fixpoint algorithm for solving the polynomial cases. Though it seems
that their algorithm is different in spirit from ours, the two algorithms have
some similarities that we use in order to ``simulate'' their fixpoint
computation using ours.

For both "self-join-free" and "path queries" the cases where the certain
answering problem is expressible in \fo is also characterized
in~\cite{DBLP:conf/pods/KoutrisOW21,DBLP:journals/tods/KoutrisW17}.  In fact,
our boundedness results use their characterizations.

Recently, the dichotomy "conjecture" has been proved for queries with two
atoms~\cite{ourPODS} and for ``rooted
tree-queries''~\cite{DBLP:journals/corr/abs-2310-19642}. Certain answers for two-atom queries in the
polynomial time cases is computed using a combination of our
fixpoint algorithm and bipartite matching.  {Bipartite matching
  hardness has also been obtained for the consistent query answering problem in the
  presence of key constraints~\cite[Lemma 6.4]{DBLP:conf/pods/KoutrisW20} or in
  the presence of negated atoms~\cite{DBLP:conf/pods/KoutrisW18}.
}

\subsection{Conference paper}
\label{sec:delta-conference}
The current article is based on the conference paper \cite{ourICDTpaper}. 
While the main results are essentially the same, though with improved explanations and figures, we have also added new material:
\begin{itemize}
    \item We have added the complete characterization of the cases where our
      fixpoint algorithm works (\Cref{sec:caseanyquery} and \Cref{path-fails-cqk}).

    \item We show the link between boundedness of our algorithm and expressibility in \fo of the certain answering problem (\Cref{section-FO}).
\end{itemize}


\section{Preliminaries}
\label{section-Preliminaries}

\AP
 A ""relational signature""
is a finite set of relation symbols associated with an arity. A finite relational
structure $D$ over a "relational signature" $\sigma$ is composed of: a finite set, the domain $\intro*\adom(D)$ of
$D$, and a function associating to each symbol $R$ of $\sigma$ a relation
$R(D)$ of the appropriate arity over $\adom(D)$. A ""database"" is a finite relational structure.

\AP
An ""$R$-fact"" of a "database" $D$ over a "relational signature" $\sigma$ is a term of the form $ R(\bar a)$ where $R$ is a
symbol of $\sigma$ and $\bar a$ a tuple in $R(D)$. A ""fact"" $u$ is an "$R$-fact" for
some $R$ if $u = R(\bar a)$ where $R$ is the symbol associated to the "fact" $u$ and $\bar a$ the tuple associated to $u$. 
A "database" can then be viewed as a finite collection of
"facts". The size of a "database" is the number of "facts" it contains. Assuming
$\sigma$ is fixed (which we will implicitly do in this paper) this is
equivalent to the usual notion of size for a "database", up to some polynomial function.

\AP A ""key constraint"" over a "relational signature" $\sigma$ specifies for every
relation symbol $R$ of $\sigma$ a certain set of indices (columns) of $R$ as a
key. A "database" satisfies the "key constraint" if for every
relation $R$ over $\sigma$, whenever two "$R$-facts" agree on the key indices
they must be equal. A set of ""primary key constraints"" has, for each relation of
$\sigma$, a unique "key constraint". Notice that if the "primary key
constraint" associated to the relation symbol $R$ contains all the columns of
$R$, then it induces no constraints on $R$. As all the sets of constraints we
consider are "primary key constraints" we will henceforth omit the `primary'
prefix. We use the letter $\Gamma$ to denote the corresponding set of key
constraints. 

\AP
Given two "facts" $u$ and $v$ and a set $\Gamma$ of "key constraints", we say that
$u$ and $v$ are ""$\Gamma$-equivalent"", denoted by $u \intro*\simblock[\Gamma] v$, if $u$ and $v$
have the same associated symbol $R$ and agree on the key of $R$ as specified by
$\Gamma$.  "$\Gamma$-equivalence" is an equivalence relation and the equivalence
classes are called ""$\Gamma$-blocks"". We will omit $\Gamma$ in our notations whenever it
is clear from the context. A "database" is then a finite collection
of "blocks", each "block" being a maximal (finite) collection of  "equivalent" "facts".
When writing a "query" $q$ we will always underline in an "atom"
$R(\bar x)$ the positions that are part of the key of $R$ as specified by
$\Gamma$. This will avoid explicitly describing $\Gamma$. For instance $R(\underline{x}~ y)$ says that the first position is the key for the binary relational symbol $R$; and
 $R'(\underline{yz}~ y)$ says that the first two positions form the key for the ternary relational symbol $R'$.

\AP
If a "database" $D$ satisfies the "key constraints" $\Gamma$, denoted by $D\models
\Gamma$, then each "block" of $D$ has size one. If not, then a ""repair"" of $D$ is a
subset of the "facts" of $D$ such that each "block" of $D$ has exactly one
representative in the "repair". In particular a "repair" always satisfies the "key
constraints". Notice that there could be exponentially many
"repairs" of a given "database" $D$.

\AP
A ""boolean conjunctive query"" over a "relational signature" $\sigma$ is a
collection of "atoms" where an ""atom"" is a term $R(\bar x)$ where $R$ is a relation
symbol from $\sigma$ and $\bar x$ is a tuple of variables of the appropriate arity. The query
being boolean, all variables are implicitly existentially quantified. 
We will consider "atoms" of a "conjunctive query" to appear in an arbitrary but fixed order.
In this paper a ``\reintro{query}'' is always a "boolean conjunctive query". 
\AP
A "database" $D$ satisfies a "query" $q$ having "atoms" $A_1, \dotsc, A_k$, denoted by 
$D \models q$, if there is a mapping $\mu$ from the variables of $q$ to the
elements of the domain of $D$ such that the "fact" $\mu(A_i) \in D$ for all $i$. 
In this case the sequence $(\mu(A_1),\dotsc, \mu(A_k))$  of (not necessarily distinct) "facts" of $D$ is called a ""solution"" to $q$ in $D$. 
Different mappings yield different solutions. The set of solutions to $q$
in $D$ is denoted by $q(D)$.
We will also write $D \models q(\bar u)$ to denote that the sequence of "facts" $\bar u$  is a solution to $q$ in $D$. 
\AP
If $\bar u = (u_1,\dotsc,u_k)$ is a solution to $q$ we also say that $u_i$  ""matches"" $A_i$ in this solution, and that any subsequence $u_{i_1}, \dotsc, u_{i_l}$ 
"matches" $A_{i_1},\dotsc, A_{i_l}$.

 \AP
  Let $r$ be a "repair" of $D$. By $|q(r)|$ we denote the number of solutions
  to $q$ in $r$, "ie", the cardinality of $q(r)$. We say that a "repair" $r$ is ""minimal"" if there is no "repair" $s$
  such that $|q(s)|< |q(r)|$.

\AP
We say that a "query" $q$ is ""certain"" for a "database" $D$ if all "repairs" of $D$
satisfy $q$. We study the complexity of determining whether a query
is certain for a "database" $D$. We adopt the \emph{data complexity} point of view. 
\AP
For
each "query" $q$ and set of "key constraints" $\Gamma$, we denote by
$\intro*\certain_\Gamma(q)$
(or simply $\certain(q)$ when $\Gamma$ is understood from the
context) the problem of determining, given a "database" $D$, whether $q$ is certain
for $D$. Clearly the problem is in \conp as one can guess a (polynomially sized) "repair" and test
whether it does not satisfy $q$. It is known that for some queries $q$ the
problem $\certain(q)$ is
\conp-complete~\cite{DBLP:journals/jcss/FuxmanM07}. However, there are queries
$q$ for which $\certain(q)$ is in \ptime or even expressible in first-order
logic (denoted by \fo in the sequel)~\cite{DBLP:journals/ipl/KolaitisP12,DBLP:journals/ipl/Wijsen10}. In this context, the following dichotomy has been conjectured (cf  \cite{DBLP:journals/jcss/FuxmanM07,DBLP:conf/icdt/AfratiK09}):

\begin{conj}[""Dichotomy conjecture""]\label{main-conjecture}
 For each "boolean conjunctive query" $q$, the problem $\certain(q)$ is either in \ptime or \conp-complete.
\end{conj}

The "conjecture" has been proved in the case of "self-join-free"
queries~\cite{DBLP:journals/tods/KoutrisW17} and of "path
queries"~\cite{DBLP:conf/pods/KoutrisOW21}; however, it remains open
in the general case. 
\AP
A "boolean conjunctive query" is ""self-join-free"" if all its
"atoms" involve different relational symbols, otherwise, it is a ""self-join query"". 
\AP
A ""path query"" is a "boolean conjunctive query" with $n+1$ distinct variables
$x_0,x_1,\cdots x_n$ and $n$ "atoms" $A_1 \cdots A_n$ such that each atom
$A_i=R_i(\underline{x_{i-1}}\, x_i)$ for some symbol $R_i$ of $\sigma$ of arity two. 
The "path query" may contain
self-joins, in other words it may be the case that $R_i=R_j$ for some $i\neq j$. 


\begin{exa}\label{first-example}
  \AP Consider the following example queries taken
  from~\cite{DBLP:journals/ipl/KolaitisP12,DBLP:conf/pods/KoutrisOW21}  (recall that all variables are implicitly existentially quantified).  For
  the "self-join-free" boolean query
  \begin{align*}
    \intro*\qOne &\defeq R_1(\underline{x}~y)\land R_2(\underline{y}~z)
  \end{align*}
 it is easy to
  see that the problem $\certain(\qOne)$ can be solved in polynomial
  time~\cite{DBLP:journals/ipl/KolaitisP12}. In fact, the first-order formula
  $\phi := \exists xyz~\Big( R_1(xy)\land R_2(yz) \land \forall y' \big(R_1(xy') \to \exists z'
  R_2(y'z')\big)\Big)$ is such that for every database instance $D$, $\qOne$ is certain for $D$ if{f} $D\models \phi$.

  \AP For the "self-join-free query" and the "path query"
  \begin{align*}
    \intro*\qTwo &\defeq R_1(\underline{x}~y)\land R_2(\underline{y}~x)\\
    \intro*\qTwop &\defeq R(\underline{x_1}~x_2)\land X(\underline{x_2}~x_3)\land
  R(\underline{x_3}~x_4) \land Y(\underline{x_4}~x_5)\land
  R(\underline{x_5}~x_6) \land Y(\underline{x_6}~x_7)
  \end{align*}
  it has been shown,
  in~\cite{DBLP:journals/ipl/Wijsen10} and~\cite{DBLP:conf/pods/KoutrisOW21}
  respectively, that $\certain(\qTwo)$ and $\certain(\qTwo')$ can be solved in
  polynomial time but cannot be expressed in first-order logic, unlike $\certain(\qOne)$.  The polynomial
  time algorithm described in the next section computes $\certain$ for $\qOne$,
  $\qTwo$ and $\qTwop$ (see also Example~\ref{second-example}).

  \AP Finally, for the "self-join-free query" and the "path query"
  \begin{align*}
    \intro*\qThree \defeq {}& R_1(\underline{x}~y)\land R_2(\underline{z}~y)\\
    \intro*\qThreep \defeq {}& R(\underline{x_1}~x_2)\land
    X(\underline{x_2}~x_3)\land R(\underline{x_3}~x_4) \land
    X(\underline{x_4}~x_5)\land {}\\
    & R(\underline{x_5}~x_6) \land
    Y(\underline{x_6}~x_7) \land R(\underline{x_7}~x_8) \land
    Y(\underline{x_8}~x_9)
  \end{align*}
  both $\certain(\qThree)$ and $\certain(\qThreep)$
  are known to be \conp-complete~\cite{DBLP:journals/jcss/FuxmanM07,
    DBLP:conf/pods/KoutrisOW21}.
\end{exa}


\section{Polynomial time algorithm}
\label{section-Cqk}

To solve $\certain(q)$, we describe a family of algorithms $\Cqk(q)$, where $k \geq 1$ is a parameter. For a fixed $k$ and query $q$, $\Cqk(q)$ takes a "database" as input and runs in
time polynomial in the size of the "database", in such a way that $\Cqk(q)$ is always an under-approximation of
$\certain(q)$, "ie", whenever $\Cqk(q)$ says `yes' then $q$ is certain for the input
"database". However, $\Cqk(q)$ could give false negative answers.

In \Cref{section-sjf} and \Cref{section-sjf-lower-bound} we will show that for
"self-join-free" queries either $\Cqk(q)$ computes $\certain(q)$ (where $k$ is
the number of "atoms" occurring in $q$) or $\certain(q)$ is $\conp$-complete
in which case $\Cqk(q)$ fails to compute $\certain(q)$ for every $k\ge 1$.
In \Cref{section-path} we show an analogous result for "path queries".

The algorithm inductively computes sets of "facts" maintaining the invariant
that every "repair" containing one of these sets makes the query true. The
algorithm returns `yes' if the empty set is eventually derived (since all
"repairs" contain the empty set).

\AP We now describe the algorithm. Assume $q,\Gamma$ and $k$ are fixed.  Let
$D$ be a "database". {A ""$k$-set"" over $D$ is a set $S$ of facts of
  $D$ of size at most $k$.}

\AP
We denote by $\Cqk(q)$ the \Cref{algo:Cqk}.
\begin{algorithm*}
  \caption{$\intro*\Cqk(q)$}\label{algo:Cqk}
  \begin{algorithmic}[1]
      \STATE \textbf{Input:} "database" $D$
      \STATE initialize $\intro*\Deltak(q,D)$ with all "$k$-sets" $S$ of $D$ such that $S\models q$
      \WHILE{there is a "$k$-set" $S \not\in \Deltak(q,D)$ and a "block" $B$ of $D$ such that for every "fact" $u\in B$ there exists $S' \subseteq S\cup \{ u\}$ where $S' \in \Deltak(q,D)$}
          \STATE update $\Deltak(q,D) \gets \Deltak(q,D) \cup \set{S}$
      \ENDWHILE
      \IF{$\emptyset \in \Deltak(q,D)$} 
        \RETURN \texttt{YES} 
      \ELSE
        \RETURN \texttt{NO}
      \ENDIF
  \end{algorithmic}

\end{algorithm*}
On a "database" input $D$, the algorithm $\Cqk(q)$ inductively computes a set $\Deltak(q,D)$
of "$k$-sets" over $D$ while maintaining the following invariant:
\begin{align}
  \textit{For every "repair" $r$ of $D$ and every $S\in \Deltak(q,D)$, if $S\subseteq r$ then $r\models q$.}  \tag{\textsc{Inv}}\label{eq:Delta:invariant}
\end{align}

\AP
Initially $\Deltak(q,D)$ contains all "$k$-sets" $S$ such that $S\models q$. In
other words,
we start with all solutions to $q$ in all "repairs" of $D$. Clearly,
this satisfies the invariant \eqref{eq:Delta:invariant}.
%
%
Now we iteratively add a "$k$-set" $S$ to $\Deltak(q,D)$ if there exists a "block" $B$ of $D$ such that 
for every "fact" $u\in B$ there exists $S' \subseteq S\cup \{ u\}$ such that $S' \in \Deltak(q,D)$. Again, it is immediate to verify that the invariant \eqref{eq:Delta:invariant} is maintained.

This is an inflationary fixpoint algorithm. Note that both the initial and
induction steps can be expressed in \fo because all sets of "facts"
$S$ computed by the fixpoint have size at most $k$, and each set can be
represented by a tuple of $k$ elements. A relation of arity $k$ can then encode $\Deltak$. The initial condition adds
any "$k$-set" that contains a solution to $q$, and the induction step again
adds only "$k$-sets".  Thus, if $n$ is the number of facts of $D$,
the fixpoint is reached
in $O(n^k)$ steps. In the end, $\Cqk(q)$ returns `yes' if{f} the empty set
belongs to $\Deltak(q,D)$. Equivalently, $\Cqk(q)$ returns `yes' if there is a
"block" $B$ of $D$ such that for all "facts" $u$ of $B$ the set $\set{u}$
belongs to $\Deltak(q,D)$. We write $D \models \Cqk(q)$ or $D\in \Cqk(q)$ to denote that
$\Cqk(q)$ returns `yes' upon input $D$.
%

The following properties are now immediate.

\begin{prop}\label{prop-underapprox}
For all $q,\Gamma,k$, $\Cqk(q)$ runs in time polynomial in the size of its input
"database" $D$ and, if $D\models \Cqk(q)$ then $D \models \certain(q)$.  
\end{prop}

{
\begin{prop}\label{prop-underapprox2}
For all $q,\Gamma,k,k'$ if $k'\ge k$ then for every database $D$, if $D\models \Cqk(q)$ then $D\models \Cqk[k'](q)$.
\end{prop}
}
In order to simplify the notations, as we will mostly consider the
  case where $k$ is the number of "atoms" in $q$, we write $\Delta(q,D)$ and $\Cqk[](q)$ to denote $\Deltak(q,D)$ and $\Cqk(q)$ respectively, where $k$ is the number of "atoms" of $q$. Also, for a "fact" 
$u$, we sometimes write $u\in\Deltak(q,D)$ instead of $\set{u}\in\Deltak(q,D)$.
{ We denote $\Deltak(q,D,i)$ to be the set $i^{th}$ step of the computation of
  $\Deltak(q,D)$. The following proposition is immediate from the definitions.

\begin{prop}\label{prop-underapprox3}
For all $q,\Gamma,k,i$ if $S\in \Deltak(q,\Gamma,i)$ and $S'\supseteq S$ such that $|S'| \le k$ then $S'\in \Deltak(q,\Gamma,i)$.
\end{prop}

}

\begin{exa}\label{second-example}
  Consider again the query
  $\qTwo: R_1(\underline{x}~y)\land R_2(\underline{y}~x)$ from
  Example~\ref{first-example}. Let $k=2$ and consider the execution of
  $\Cqk[2](\qTwo)$. For a given input database $D$, initially $\Delta(\qTwo,D)$ contains all pairs of "facts"
  $\set{R_1(ab),R_2(ba)}$ such that both $R_1(ab)$ and $R_2(ba)$ are in $D$.
  The first iterative step adds to $\Delta(\qTwo,D)$ (i) all singletons
  $\set{R_1(ab)}$ such that $R_2(ba)$ is a "fact" of $D$ whose "block" contains
  only $R_2(ba)$, and (ii) analogously all $\set{R_2(ab)}$ such that the "block"
  of $R_1(ba)$ is a singleton.
  
 In subsequent steps, the empty set is added to $\Delta(\qTwo,D)$ if at some point, there is some "block" $B$ such that for every "fact" $u\in B$ we have $u\in \Delta(\qTwo,D)$. At this point the algorithm outputs `yes' and if $\Delta(\qTwo,D)$
saturates without the empty set as its member, then the algorithm outputs `no'.



  \bigskip
  
  \label{second-example:contd}
  We show that $\Cqk[2](\qTwo)$ computes $\certain(\qTwo)$ or, in other words, $D\models \certain(\qTwo)$ if{f} $D\models \Cqk[2](\qTwo)$. 

  Observe that, for every "repair" $r$ and "fact" $\alpha$ therein, there is at most one other "fact" $\alpha'$ in $r$ such that $\set{\alpha,\alpha'} \models \qTwo$. This is because in any "repair" the first "atom" of $\qTwo$ determines the second "atom" and vice-versa. This ``mutual determinacy'' is, in fact, what makes $\Cqk[2](\qTwo)$ a complete procedure, as we shall see next.
  
  In view of Proposition~\ref{prop-underapprox}, it remains to
  show that if $\qTwo$ is certain for $D$ then  $\Delta(\qTwo,D)$  contains
  the empty set. Towards this we use the following observation about $\qTwo$.

  \begin{clm}\label{claim-qTwo}
      If $r$ is a "minimal" "repair" and both "facts" $R_1(ab)$ and $R_2(ba)$ are  in $r$
      then $R_1(ab)\in\Delta(\qTwo,D)$.    
  \end{clm}
  
   Assume that the claim is true. Now suppose $\qTwo$ is certain for $D$, then for any "minimal" "repair" $r$, we must have $r \models q$ and
  this is witnessed by two "facts" $R_1(ab)$ and $R_2(ba)$ of $r$. Let $B$ be the "block"
  of $R_1(ab)$. Let us show that all "facts" of $B$ are in $\Delta(\qTwo,D)$ as singleton sets and hence $\emptyset \in \Delta(\qTwo,D)$. Let
  $R_1(ab')$ be such a "fact" and consider the "repair" $r'$ obtained by replacing
  $R_1(ab)$ with $R_1(ab')$. As $r$ is "minimal" it follows immediately that $r'$ is
  "minimal" and must contain $R_2(b'a)$ (again, this is ensured by the mutual determinacy of $\qTwo$). From the claim it follows that $R_1(ab')\in
  \Delta(\qTwo,D)$, as desired. Thus it remains to prove the \Cref{claim-qTwo}. 
  
  \begin{proof}[Proof of \Cref{claim-qTwo}] Assume that $r$ is a "minimal" "repair" containing both $R_1(ab)$ and $R_2(ba)$.
   Towards a contradiction, suppose
  $R_1(ab)\not\in\Delta(\qTwo,D)$ then we shall construct an infinite sequence
  $u_1,u_2,\dotsc$ of distinct "facts" of $D$, contradicting the finiteness
  of $D$. Towards this we additionally construct an infinite sequence $v_1,v_2,\dotsc$ of
  "facts" of $D$ and an infinite sequence of "minimal" "repairs" $r_1,r_2,\dotsc$
  maintaining the following invariants for every $i$:
  \begin{enumerate}
    \item the $u_i$'s are pairwise distinct;
  \item $u_i\not\in\Delta(\qTwo,D)$;
  \item if $u_i=R_1(cd)$ then $v_i=R_2(dc)$ and if $u_i=R_2(cd)$ then $v_i=R_1(dc)$;
  \item $u_{i+1}\simblock v_i$ and $u_{i+1} \neq v_i$;
  \item \label{ex:it:ri} $r_i$ is "minimal" and contains $v_i$ and each $u_j$ for all $j\leq i$.
  \end{enumerate}
  
  Initially $r_1=r$, $u_1=R_1(ab)$ and $v_1=R_2(ba)$ and all  the invariant conditions are
  met: $(1)$ is trivial, $(2)$ and $(5)$ follow from the assumption, $(3)$  is true by construction and $(4)$ does not apply.
  
  Consider step $i$. Consider the block $B_i$ of $v_i$. As
  $\{u_i\}\not\in\Delta(\qTwo,D)$ it means that we cannot use any block $B$ as a witness to add $u_i$ to $\Delta(\qTwo,D)$ ("ie", for every block $B$ there is some "fact" $w\in B$ such that $\{u_i,w\}\not\in \Delta(\qTwo,D)$). Hence, in particular, $B_i$ must contain an element
  $u_{i+1}$ such that both $u_{i+1}\not\in\Delta(\qTwo,D)$ and
  $\set{u_i,u_{i+1}}\not\in\Delta(\qTwo,D)$. In particular $u_{i+1} \simblock v_i$ but
  $u_{i+1} \neq v_i$ and items~$(2)$ and~$(4)$ of our induction hypothesis are met. Towards the first item of our induction hypothesis, if $u_{i+1}=u_j$ for some $j\le i$ then by item~$(5)$ the "repair" $r_i$ would contain two "equivalent" "facts", $v_i$ and $u_{i+1}=u_j$, which is not possible since we have already established that $u_{i+1} \neq v_i$.

  Consider the "repair" $r_{i+1}$ resulting from replacing $v_i$ with $u_{i+1}$.
  Let $v_{i+1}$ be the dual "fact" of $u_{i+1}$ as required by the third item of
  the invariant. As $u_iv_i$ forms a solution to $q$ in $r_i$ and $r_i$ is
  "minimal", we must have $v_{i+1}\in r_{i+1}$ (otherwise $r_{i+1}$ has strictly fewer solutions than $r_i$). Finally, notice that $r_{i+1}$ is
  "minimal" as its solutions to $q$ are exactly the same as for $r_i$ except for
  $u_iv_i$ that has been removed and $u_{i+1}v_{i+1}$ that has been added (by the mutual determinacy of the "atoms" of $\qTwo$).
  
  Here is a depiction of how the $u_i$'s and $v_i$'s are defined, where the full and hollow arrows correspond to $R_1$ and $R_2$ respectively.
  \begin{center}
  \includegraphics[width=.5\textwidth]{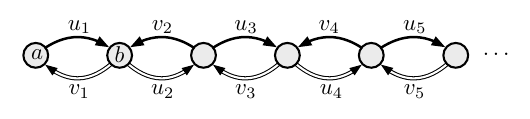}
  \end{center}
  
  This concludes the construction of the infinite sequence, showing that
  $R_1(ab)\in\Delta(\qTwo,D)$ for any "minimal" "repair" containing both $R_1(ab)$ and
  $R_2(ba)$ which proves claim.
  \renewcommand{\qedsymbol}{\usebox{\lmcsQEDSymbol}}  
\end{proof}
\renewcommand{\qedsymbol}{\exampleqedsymbol}
This concludes our \Cref{second-example}.
\end{exa}

$\Cqk$ does not always compute the certain answers. For instance,
the query $\qThree$ from Example~\ref{first-example} is so that
$\certain(\qThree)$ is \conp-complete, and hence $\Cqk(\qThree)$ must have
false negatives for all $k$, under the hypothesis that $\conp \neq \ptime$. Proving this without
relying on complexity theoretic assumptions is the goal of
\Cref{sjf-fails-cqk} for "self-join-free" queries and \Cref{path-fails-cqk} for
"path queries".


\section{Tractable self-join-free queries}
\label{section-sjf}


In this section we consider the case of "self-join-free" queries. We exhibit a
condition named $\HP$ (for Polynomial time Condition) and show that any
"self-join-free query" $q$ satisfying $\HP$ is such that $\Cqk$ computes
$\certain(q)$, where $k$ is the number of "atoms" in $q$. When $\HP$ fails, we will
see that for all values of $k$, $\Cqk$ fails to compute $\certain(q)$ and moreover,
$\certain(q)$ is \conp-hard.

We start by defining \HP, which will require some extra definitions.
Fix, for the rest of this section, a set $\Gamma$ of "primary key
constraints". Let $D$ be a "database" and $r$ a "repair" of $D$. For a "fact"
$u$ of $r$, and for an "equivalent" "fact" $v\simblock u$ from $D$, we denote
by $r[u\to v]$ the "repair" obtained from $r$ by replacing the
"fact" $u$ with $v$.

Consider a "self-join-free query" $q$ with $k$ "atoms".  Recall that we write
{$D\models q(\bu)$} when $\bu$ is a solution to $q$ in $D$.  As $q$ is
"self-join-free", for each "fact" $a$ in a solution $\bar u$ there is a unique
"atom" of $q$ that $a$ can "match", namely the only "fact" of $q$ having the same
relation symbol as $a$.  Hence, the order on both $\bu$ and the "atoms" of $q$ are
not relevant. With some abuse of notation we will therefore often treat a
solution $\bu$, or the sequence of "atoms" of $q$, as a \emph{set} rather than a
\emph{sequence}; we will often use different orders among the "facts" of a 
solution, placing up front the most relevant "facts". Also we shall
write, for a tuple $\bu$ of "facts", $\bu\in\Delta_k(q,D)$ to denote that the
set formed by the "facts" of $\bu$ is a "$k$-set" and belongs to $\Delta_k(q,D)$.

\AP
Let $A$ be an "atom" of $q$ whose associated symbol is $R$. We denote by
$\intro*\vars(A)$ the set of variables of $A$ and by
$\intro*\key(A)$ the set of variables of $A$ occurring in a
position belonging to the primary key of $R$. For instance
$\key(R(\underline{x}~y))$ is $\set{x}$, $\key(R'(\underline{yz}~ x))$ is
$\set{y,z}$ and $\key(R''(\underline{xzx}\,y\,))$ is $\set{x,z}$.

\AP
Given a set $X$ of variables of $q$ and a sequence $A_1 \dots A_n$ of "atoms" of
$q$, we say that $X~A_1 \dots A_n$ is a
""$\Gamma$-derivation"" from $X$ to $A_n$ in $q$ if
for each $ 1 \leq i\leq n$ we have that
$$\key(A_i) \subseteq X \cup \bigcup_{1 \leq j < i}\vars(A_j).$$

\AP If $X = \vars(A_0)$, for some "atom" $A_0$ of $q$, we say that the
"$\Gamma$-derivation" is from $A_0$ to $A_n$, and write it as
$A_0 A_1 \dots A_n$.  We say that an "atom" $A'$ is ""$\Gamma$-determined"" by
the "atom" $A$ if there exists a "$\Gamma$-derivation" from $A$ to
$A'$. Moreover, $A$ and $A'$ are ""mutually $\Gamma$-determined"" if $A'$ is
$\Gamma$-determined by $A$ and $A$ is $\Gamma$-determined by $A'$. This is an
equivalence relation among "atoms".  \AP A set $S$ of "atoms" of $q$ is called
""stable"" if every two distinct "facts" of $S$ are "mutually $\Gamma$-determined".
Note that if an "atom" $A$ is in a "stable set" $S$ then $S$ need not contain
all the "atoms" that are "mutually $\Gamma$-determined" by $A$. So a partition
of the atoms of $q$ into "stable sets" is a refinement of the partition induced by
"mutual $\Gamma$-determinacy".  Notice also that if two atoms $A$ and $B$ are in a
"stable set" $S$, we do not require that $S$ contains the atoms witnessing
their mutual determinacy.
 As usual, we will omit
$\Gamma$ when it is clear from the context.  

The main intuition on how we will use "$\Gamma$-derivations" is the following.
Suppose there is a query $q$ with atoms $A_1, \dotsc, A_n$ which has "solutions" in two "repairs" $r,r'$ of a "database", witnessed by valuations $\mu,\mu'$.
If there is a (one-step) "derivation" ``$X \, A_1$'', this means that $\key(A_1) \subseteq X$, so if $\mu,\mu'$ agree on $X$, in particular they agree on $\key(A_1)$ (that is, the corresponding atoms $\mu(A_1)$ and $\mu'(A_1)$ are "$\Gamma$-equivalent").
Further, if $\mu(A_1) \in r'$, we can actually say that $\mu,\mu'$ agree on all the variables in $\vars(A_1)$, since the "repairs" must necessarily have the same "$A_1$-fact@$R$-fact" under the key $\mu(\key(A_1))=\mu'(\key(A_1))$.
Now if we were to add another step in the "derivation" ``$X \, A_1 \, A_2$'', we further have $\key(A_2) \subseteq X \cup \vars(A_1)$. Since $\mu,\mu'$ agree on $X \cup \vars(A_1)$, we now have $\mu(A_2)$ and $\mu'(A_2)$ are "$\Gamma$-equivalent".
This key property relating "$\Gamma$-derivations" and query "solutions" (and extended to arbitrary length "derivations") is formally stated in the lemma below.

\begin{lem}\label{lem-determinacy}
  Let $q$ be a "self-join-free query".  Let $D$ be a "database" instance and
  $r,r'$ be two "repairs" of $D$. Let $X$ be a set of variables of $q$ and let
  $X~A_1\ldots A_n$ be a "$\Gamma$-derivation" from $X$ to $A_n$ in $q$.  Let
  $q(r)$ contain a "solution" witnessed by a valuation $\mu$ of variables of
  $q$, and $q(r')$ contain a "solution" witnessed by a valuation $\mu'$. If
  $\mu$ and $\mu'$ agree on $X$ and $\mu(A_i) \in r'$ for all $i<n$, then:
  \begin{enumerate}[1)]
    \item $\mu(A_i) =\mu'(A_i)$ for each $i < n$ and 
    \item $\mu(A_n) \simblock \mu'(A_n)$ (and therefore $\mu(A_n)=\mu'(A_n)$ if moreover $\mu(A_n)\in r'$).
  \end{enumerate}
  \end{lem}
  \begin{proof} The proof is by induction on $n$.
    For $n =1$ the statement trivially holds since $\key(A_1) \subseteq X$ thus
    $\mu(A_1) \simblock \mu'(A_1)$.
  
    Now consider a sequence $X A_1 \dots A_n$, $n>1$ satisfying the
    hypotheses. The induction hypothesis applied to the sequence
    $X A_1 \dots A_{n-1}$ implies $\mu(A_i) = \mu'(A_i)$, for all
    $1 \leq i \leq n-1$.  Then $\mu$ and $\mu'$ agree on
    $vars(A_1) \cup\dotsb \cup vars(A_{n-1}).$
  
    Since $X A_1 \dots A_n$ is a "$\Gamma$-derivation" sequence, we have
    $key(A_n) \subseteq X \cup vars(A_1) \cup \dotsb\cup vars(A_{n-1})$; hence
    $\mu$ and $\mu'$ agree on $\key(A_n)$, or in other words
    $\mu(A_{n}) \simblock \mu'(A_{n})$.
  \end{proof}

\begin{cor}
\label{cor-determinacy}
{Let $q$ be a "self-join-free query" and $S$ be a "stable set" of
  "atoms" of $q$. Let $D$ be a "database" instance and $r$ be a "repair" of
  $D$.} Assume
$r \models q(\bar\alpha \bar a \bar \beta)$ and
$r \models q(\bar\alpha' \bar a' \bar\beta')$, where $\bar a$ and $\bar a'$
"match" $S$.  If $\bar a \cap \bar a' \neq \emptyset$ then $\bar a = \bar a'$.
\end{cor}
\begin{proof}
  The statement follows directly from \Cref{lem-determinacy} using $r = r'$.
  Take $a_1$ as any "fact" in $\bar a \cap \bar a'$, $A_1$ being the "atom"
  "matched" by $a_1$, $X$ being $\vars(A_1)$ and $A_1 \dots A_k$ being a
  "$\Gamma$-derivation" sequence containing all $S$ (which exists by
  stability).
\end{proof}

\AP We are now ready to define \HP.  A ""$\Gamma$-sequence"" $\tau$ of $q$ is a
sequence $\tau=S_1S_2\cdots S_n$ where each $S_i$ is a "stable set" of "atoms" of
$q$, and the $S_i$'s form a partition of $q$.  In this context, we denote by
$S_{\leq i}$ the set $\bigcup_{j\leq i} S_j$. We define $S_0$ to be the empty set.

\AP
Let $\tau=S_1S_2\cdots S_n$ be a "$\Gamma$-sequence" of $q$. Let $1\leq i < n$
and let $A$ be an "atom" of $S_{i+1}$.
\AP
We say that the query $q$ satisfies
$\intro*\HPtauA \tau A$ and write $q \models \HPtauA \tau A$ if
the following is true for all "databases" $D$, all "repairs" $r$ of $D$ and all
solutions $\balpha u \bbeta$ and $\balpha' u' \bbeta'$ to $q$ in $D$ such that
$\balpha$ and $\balpha'$ "match" $S_{\leq i}$ and $u$ and $u'$ "match" $A$:
\begin{center}
If $
\left\{
\begin{aligned}\label{hypothesis-ha}
  r&\models q(\balpha u \bbeta ), \\
  u &\simblock u' \text{, and}\\  
  r&[u \rightarrow u']\models q(\balpha' u' \bbeta') 
\end{aligned}\right\}
\text{ then }
  r \models q(\balpha' u \bdelta)
$ for some sequence of facts $\bdelta$.
\end{center}



Note that, by symmetry, we also have $r[u\rightarrow u'] \models q(\balpha u'\bdelta')$ for some sequence of facts $\delta'$.
\AP We write $q \models \intro*\HPtaui \tau i$ if $q$ satisfies
$\HPtauA \tau A$ for all $A$ of $S_{i+1}$, and we write
$q \models \intro*\HPtau$ if $q$ satisfies $\HPtaui \tau i$ for all
$1\leq i <n$. Since the condition is restricted to indices $i<n$, $\HPtau$
trivially holds for any $\tau$ having only one "stable set".  \AP Finally, we
write $q\models \intro*\HP$ if there is a "$\Gamma$-sequence" $\tau$ of $q$
such that $q\models \HPtau$. Again, if $q$ has only one $\Gamma$-determinacy
class (for instance the query $\qTwo$ of \Cref{second-example}) then
$q\models\HP$ in a trivial way.

We illustrate the definition of $\HP$ with the following examples.
\begin{exa}\label{third-example}
  We recall the three queries from \Cref{first-example}.  The query
  $\qTwo = R_1(\underline{x}~y)\land R_2(\underline{y}~x)$ satisfies $\HP$
  since it has only one maximal "stable set".
  
  The query $\qOne = R_1(\underline{x}~y)\land R_2(\underline{y}~z)$ has two
  "stable sets": $R_1(\underline{x}~y)$ determines $R_2(\underline{y}~z)$ but the
  converse is false. For
  $\tau=\set{R_2(\underline{y}~z)}\set{R_1(\underline{x}~y)}$ we have
  $q\not\models\HPtau$ because we have the solutions $\qOne(R_2(bc)R_1(ab))$ and
  $\qOne(R_2(b'c)R_1(ab'))$ but not $\qOne(R_2(bc)R_1(ab'))$.  However for
  $\tau=\set{R_1(\underline{x}~y)}\set{R_2(\underline{y}~z)}$ it is easy to
  verify that $\qOne\models\HPtau$.
  Hence, $\qOne\models\HP$.

  The query $\qThree = R_1(\underline{x}~y)\land
  R_2(\underline{z}~y)$ has also two "stable sets", but no
  possible sequence $\tau$ makes $\HPtau$ true. This is because
  (i) $\qThree(R_1(ab)~R_2(cb))$ and $\qThree(R_1(a'b')~R_2(cb'))$ hold, but not
  $\qThree(R_1(ab)~R_2(cb'))$, 
  and (ii) $\qThree(R_2(ab) ~R_1(cb))$ and $\qThree(R_2(a'b')~R_1(cb'))$ hold, but not $\qThree(R_2(ab)~R_1(cb'))$. Therefore, $\qThree\not\models\HP$.  
\end{exa}

Our goal for the remaining part of this section is to show that $q\models \HP$
implies that $\Cqk[](q)$ computes $\certain(q)$ (Recall that $\Cqk[](q)$ denotes $\Cqk[k](q)$ where $k$ is the number of "atoms" in $q$) which is given by \Cref{theorem-HP-ptime}. In
\Cref{section-sjf-lower-bound} we will see that when $q\not\models\HP$ then
$\certain(q)$ cannot be computed by $\Cqk$ for any choice of $k$ and that $\certain(q)$ is actually
\conp-hard.

\begin{thm}\label{theorem-HP-ptime}
  Let $q$ be a "self-join-free query" with $k$ "atoms". If $q \models\HP$, then
  $\Cqk(q)$ computes $\certain(q)$.
\end{thm}
Suppose $q$ has $k$ "atoms".
Let $\tau=S_1 \cdots S_n$ be a "$\Gamma$-sequence" of $q$ such that
$q \models \HPtau$. We show that $\Cqk(q)$ computes precisely
$\certain(q)$.

If $q$ has only one "stable set" (and thus it trivially satisfies $\HP$), the
proof is similar to the proof of \Cref{second-example}, starting with a
minimal "repair" and exploiting the mutual determinacy of the "atoms" of $q$. If
$q$ has more "stable sets", then the condition of minimality needs to be more
fine-grained and we proceed by induction on the index of the "stable sets", in the
order described by $\tau$.

  We start with some extra notations. Recall that $q(r)$  denotes the set of solutions to $q$ in a "repair" $r$; we additionally denote by $\ql{i}(r)$ 
  the projection of $q(r)$ on the first $i$ "stable sets" of $\tau$. More precisely 
  \[
    \ql{i}(r) = \{ \bar v~|~ \exists \bar u \in q(r)\textrm{  s.t.\ }\bar v 
  \textrm{ is the subset of }\bar u\textrm{ "matching"  }S_{\leq i}\},
  \]
  and if $\bar v \in \ql{i}(r)$ we write equivalently $r \models \ql{i}(\bar v)$. 
  Let $D$ be a "database" and $r$ a "repair" of $D$. 
  \AP
  We say that $r$ is
  ""$i$-minimal"" if there is no "repair" $r'$ such that
  $\ql{i}(r') \subsetneq \ql{i}(r)$.
  We say that a "fact" $u$ of a "database" $D$ is ""$i$-compatible"", if it "matches" some
  "atom" of $S_i$.
  We will need the limit case when $i=0$. In that case both $S_0$ as well as $S_{\leq 0}$  are empty sets (and hence $\HPtaui \tau 0$ is always true),
  $q_{\leq 0}(r)$ contains only the empty sequence $\epsilon$ for all $r$, and therefore all "repairs" are $0$-minimal.
   The proof of the theorem makes use of an induction based on the following invariant property of the "database", for each $0 \leq i \leq n$:
   \AP
  \begin{align*}
    \intro*\Ind_i = \text{For all "$i$-minimal" "repairs" $s$ and "facts" $\bar u$ s.t.\ $s
    \models \ql i(\bar u)$, we have $\bar u \in \Deltak(q,D)$.}
  \end{align*}

\begin{lem}
\label{induc-prop}
Given $q$, $D$ and a "$\Gamma$-sequence" $\tau$ for $q$, for every
$0 \leq i < n$, if $\Ind_{i+1}$ and $\HPtaui \tau i$, then $\Ind_{i}$.
\end{lem}

We first show how this statement already implies \Cref{theorem-HP-ptime}.
\begin{proof}[Proof of \Cref{theorem-HP-ptime}]
  From \Cref{prop-underapprox}, we know that if $D$ is a "database" such that
  $D\models\Cqk(q)$ then all "repairs" of $D$ satisfy $q$. It remains to show
  the converse.

Assume all "repairs" satisfy $q$ and that $q \models\HPtau$ for some sequence
$\tau$ of length $n$, which means that $\HPtaui \tau i$ holds for all $i$.
Observe that $\Ind_{n}$ holds true by the base case definition of
$\Delta_k(q,D)$. Hence by $n$ repeated applications of \Cref{induc-prop} we
obtain that $\Ind_0$ holds true.  Now take any "repair" $r$. By definition $r$
is "$0$-minimal" and by hypothesis it satisfies the query $q$. By $\Ind_0$ it
follows that the empty set (denoted by the empty tuple) is in $\Delta_k(q,D)$,
and hence $D\models \Cqk(q)$.
\end{proof}

We are now left with the proof of \Cref{induc-prop}, which is the main
technical content of the section.  {Towards this, we define a stronger version
  of "$i$-minimality".}  \AP For $1 \leq i < n$, we say that an "$i$-minimal"
"repair" $s$ is ""strong $i$-minimal"" if there exists no "repair" $s'$ such
that $\ql{i}(s') = \ql{i}(s)$ and $|\ql{i+1}(s')| < |\ql{i+1}(s)|$.  Note that
if $\ql{i+1}(s') \subsetneq \ql{i+1}(s)$ then either
$\ql{i}(s') \subsetneq \ql{i}(s)$ or $\ql{i}(s') = \ql{i}(s)$ but
$|\ql{i+1}(s')| < |\ql{i+1}(s)|$. In particular, every "strong $i$-minimal" "repair"
is "$(i+1)$-minimal".

  \begin{clm}\label{claim-strong-minimal}
    If there exists an "$i$-minimal" "repair" $s$ such that
    $s\models q_{\leq i}(\bar u)$, then there exists a "strong $i$-minimal"
    "repair" $s'$ such that $s'\models q_{\leq i}(\bar u)$.
  \end{clm}
  \begin{proof}
    Among all "repairs" $s''$ having $\ql{i}(s'') = \ql{i}(s)$, choose $s'$ as
    having minimal $|\ql{i+1}(s')|$. In other words, $s'$ is a "repair" having
    $\ql{i}(s') = \ql{i}(s)$ and for every "repair" $s''$ with
    $\ql{i}(s'') = \ql{i}(s)$ we have $|\ql{i+1}(s')| \leq
    |\ql{i+1}(s'')|$. Hence, $s'$ is "strong $i$-minimal".
  \end{proof}

  \AP
  For a given "database" $D$, for a "repair" $r$ of $D$, we denote by
  $\intro*\midi r {i+1}$ the set of "facts" of $r$ which are not
  "$(i+1)$-compatible".  \AP A sequence $\bar p$ of "facts" of the "database"
  is ""connected"" with respect to $D' \subseteq D$ if for every "repair" $r$
  containing $\bar p$ and $D'$, and for every two consecutive "facts" $a\, b$
  of $\bar p$, if $r \models q(\bar\alpha a \bar\beta)$ for some
  $\bar\alpha, \bar\beta$, then $b \in \bar\alpha\bar\beta$.
  Note that if $\bar p$ is the empty tuple (or a tuple of size 1), then
  $\bar p$ is trivially "connected" with respect to every $D' \subseteq D$.

\begin{proof}[Proof of \Cref{induc-prop}]
  By contradiction, suppose the statement of the lemma is false. Then, there is
  some $i$ such that $\Ind_{i+1}$ and $\HPtaui \tau i$ holds, but for some
  "$i$-minimal" "repair" $s$ and tuple $\bar u$ we have
  \begin{align}
    s \models q_{\leq i}(\bar u) \text{ but }\bar u \not\in \Delta_k(q,D). \tag{h1}\label{hyp:contr:h1}
  \end{align}

  From \Cref{claim-strong-minimal}, we can assume that $s$ is "strong
  $i$-minimal".  We will build an infinite sequence of pairwise distinct
  "facts" $ p_1,p_2, \dotsc$ from $D$, contradicting the finiteness of $D$. We
  also maintain another sequence of "repairs" $r_1,r_2,\dotsc$. We set
  $r_0 = s$. For all $l > 0$, we define $\bar p_l = p_1,\ldots ,p_l$ with
  $\bar p_0$ as the empty "fact" sequence.

  The sequence is constructed by induction with the following invariant for
  every $l \geq 0$, assuming $\bar p = \bar p_l$ and $r = r_l$:
\begin{enumerate}[(a)]
\item\label{pi} $\bar p$ contains only "$(i+1)$-compatible" "facts" of $D$;
\item\label{pd} the elements of $\bar p$ are pairwise distinct;
\item\label{pc} $\bar p$ is "connected" with respect to $\midi r {i+1}$;
\item\label{ru} $r$ is "strong $i$-minimal", $r \models q_{\leq i}(\bu)$ and,
  if $\bar p$ is not empty and $v$ is the last "fact" of $\bar p$, then
  $r \models q(\bu v \bar \delta)$, for some $\bar \delta$;
\item\label{suf} $\bu\bc\not\in \Delta_k(q,D)$, where $\bc$ is the maximal suffix
  of $\bar p$ satisfying $r' \models q(\bu\bc\bbeta)$ for some $\bbeta$ and
  "strong $i$-minimal" "repair" $r'$ containing $\bar p$ and $\midi r {i+1}$.
\end{enumerate}

\subparagraph*{Base case} When $l = 0$, we have $r_0 = s$ and $\bar p$ is the
empty sequence.  Hence \ref{pi}, \ref{pd} and \ref{pc} are trivially true by
emptiness of $\bar p$; \ref{ru} holds since $s \models q_{\leq i}(\bu)$ (note
that we have assumed $s$ to be "strong $i$-minimal"); finally \ref{suf} holds
with empty $\bar c$ since $\bu \notin \Delta_k(q,D)$ by \eqref{hyp:contr:h1}.

\subparagraph*{Induction step} Assume we have $r = r_{l-1}$ and
$\bar p_{l-1} = p_1,\ldots ,p_{l-1}$ (possibly empty) satisfying the five properties
above (for the rest of the proof we denote $\bar p_{l-1}$ as $\bar p$ for brevity).  Consider the maximal suffix $\bar c$ concerned by
property~\ref{suf}. That is, for some $\bbeta$ and "strong $i$-minimal"
"repair" $r'$ containing $\bar p$ and $\midi r {i+1}$ we have:
  \begin{align}
    \text{$\bu \bc \not\in \Delta_k(q,D)$ and   $r'\models q(\bu\bc\bbeta)$}  \tag{h2}\label{hyp:contr:h12}
  \end{align}
  First let $\bbeta =d_1d_2,\ldots d_t$. Since $\bu \bc \not\in \Delta_k(q,D)$,
  by definition of $\Delta_k(q,D)$ there exists some $d_1'\simblock d_1$ such
  that $\bu\bc d_1'\not\in \Delta_k(q,D)$. This again implies that there exists
  some $d_2'\simblock d_2$ such that $\bu\bc d_1'd_2' \not\in
  \Delta_k(q,D)$. Since $k$ is the number of "atoms" in $q$, we can continue this
  to obtain $\bbeta' = d_1'd_2'\ldots d_t'$ where $d_i' \simblock d_i$ but
  $\bu\bc\bbeta'$ contains no "$k$-set" in $\Delta_k(q,D)$.
 
  We now show that $\bar c$ cannot "match" the entire set $S_{i+1}$.  Suppose,
  by means of contradiction, that $r' \models q_{\leq i+1}(\bu \bc)$. As $r'$
  is "strong $i$-minimal", it is "$(i+1)$-minimal". Hence by
  \eqref{hyp:contr:h1}, since $\Ind_{i+1}$ holds by hypothesis,
  $\bu \bc \in \Delta_k(q,D)$, which is in contradiction with
  \eqref{hyp:contr:h12}.  Then,
  \begin{align}
      r' \models \lnot\ql{i+1}(\bu \bc). \tag{h3}\label{hyp:contr:h2}
  \end{align}
 
  This means that, since $r' \models q (\bu \bc \bbeta)$ by
  \eqref{hyp:contr:h12}, and $\bc$ "matches" a subset of $S_{i+1}$, there must be
  an "atom" $C$ of $S_{i+1}$ that is not "matched" by any "fact" of $\bar
  c$. Consider the "atom" $A$ of $S_{i+1}$ "matching" the last element of $\bar
  c$. If instead $\bar c$ is empty, choose $A$ as an arbitrary "atom" of
  $S_{i+1}$.

  Since $A$ and $C$ are both in $S_{i+1}$, which is "stable", there exists a
  "$\Gamma$-derivation" $\sigma$ from $A$ to $C$. (Notice that $\sigma$ may
  contain "atoms" outside $S_{i+1}$.)  Consider the first "atom" $B$ of $\sigma$
  which is in $S_{i+1}$ and which is not "matched" by any "fact" of $\bar
  c$. The following depiction may help to see the situation:
  
  \begin{center}
   \includegraphics[width=.8\textwidth]{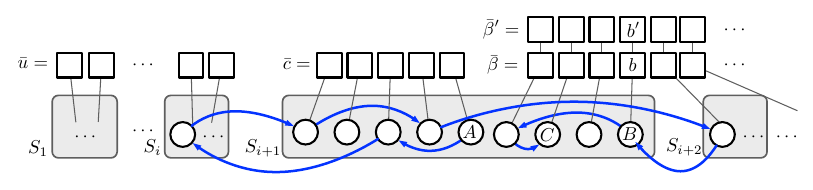}
  \end{center}
  
  In the picture directed edges connecting "atoms" of the query represent the successor relation in the 
  "$\Gamma$-derivation" from $A$ to $C$.
  
  Let $b$ be the "fact" of $\bar \beta$ "matching" $B$ and $b' \simblock b$ be
  the corresponding "fact" in $\bar \beta'$.  We show that
  \begin{align}
  \text{$b \not\in \bar p$. \tag{h4}\label{hyp:h34}}
  \end{align}
  Suppose $b$ is in $\bar p$. By construction, $b$ is not in $\bar c$, thus it
  must occur before $\bar c$ in $\bar p$ and hence the suffix $\bar b$ of $\bar
  p$ starting with $b$ strictly contains
  $\bar c$. By "connectedness" of $\bar p$ with respect to $\midi r {i+1}$, as $b
  \in \bbeta$ is part of the
  solution $\bar u \bar c \bar \beta$, we must have
  $r' \models q(\bar u \bar b \bar \gamma)$ for some $\bar \gamma$.  This
  contradicts the maximality of $\bar c$ imposed by \ref{suf}, thus proving
  that \eqref{hyp:h34} holds.  Note that this also implies $b'\not\in \bar p$,
  otherwise if $b' \in \bar p$, we have that $b'\simblock b$ are both in $r'$,
  thus $b = b' \in \bar p$, contradicting~\eqref{hyp:h34}.

  Assign $p_l = b'$, so we have $\bar p' = \bar p \cdot p_l$ and let
  $ r_l = r'[b\to b']$.  (To avoid many subscripts, let $r_l = s'$).  Observe
  that
  \begin{align}
  \text{$s'$ contains $\bar p'$ and $\midi r {i+1}$.}  \tag{h5}\label{hyp:contr:h4}
  \end{align}
  In fact $s'$ contains $\bar p$, as observed earlier, and $s'$ contains $b'$
  by construction; moreover $s'$ contains $\midi{r'}{i + 1}$ which contains
  $\midi r {i+1}$ by \ref{suf}.
  We now show that $\bar p'$ and $s'$ have all the desired properties.
  
  \proofcase{\ref{pi}} By construction $b'$ is "$(i+1)$-compatible".
  
  \proofcase{\ref{pd}} The elements of $\bar p'$ are pairwise distinct, as $b' \not\in \bar p$.
  
  \proofcase{\ref{pc}} By our choice of $b$ we show that $\bar p'$ is connected
  with respect to $\midi{s'}{i+1}$. Without loss of generality assume that
  $\bar p'$ has at least size 2 (otherwise it is trivially
  connected). Therefore, $\bar p$ is not empty. Since $\midi{s'}{i+1}$ contains
  $\midi r {i+1}$ by \eqref{hyp:contr:h4}, the connectedness property of
  $\bar p$ with respect to $\midi r {i+1}$ implies that for every "repair"
  containing $\bar p'$ and $\midi {s'} {i+1}$ and for every pair of consecutive
  facts $a\, b$ in $\bar p$, every solution in $s'$ containing $a$ also
  contains $b$.

  It remains to show the same property for the last "fact" $a$ of $\bar
  p$. Consider a "repair" $t$ containing $\bar p'$ and $\midi{s'}{i+1}$ and
  suppose $t \models q(\bar \gamma a \bar \delta)$ for some $\bar\gamma$ and
  $\bar\delta$.  We have to show $b' \in \bar\gamma\bar\delta$.
  Let $\sigma_{AB}$ be the prefix of the "$\Gamma$-derivation" $\sigma$ going
  from $A$ to $B$ in the "$\Gamma$-derivation" from $A$ to $C$.  (Notice that, since $\bar p$ is not empty we have $A \neq B$.)
  By property \ref{ru} of $\bar p$, since $\bar p$ is not empty, $a$ is the
  last "fact" in $\bar c$.  Recall that by~(\ref{hyp:contr:h12})
  $r' \models q(\bar u \bar c \bar \beta)$; thus in this solution the "atom" $A$
  is "matched" by $a$.  So we can apply \Cref{lem-determinacy} to $r'$ and $t$
  with solutions $(\bar u \bar c \bar \beta)$ and $(\bar \gamma a \bar \delta)$
  respectively, and "$\Gamma$-derivation" $\vars(A)~\sigma_{AB}$. The hypotheses of
  \Cref{lem-determinacy} are satisfied since:
 \begin{itemize}
 \item in both solutions $A$ is "matched" by $a$;
 \item  by construction of $B$, for each "atom" $D$  strictly preceding $B$ in $\sigma_{AB}$, the "fact" "matching" $D$ in $(\bar u \bar c \bar \beta)$ is either in $\bar c$ or in $\midi{r'}{i+1}$, both contained in $t$ (in fact $t \supseteq \bar p' \supseteq \bar p \supseteq \bar c$ and $t \supseteq \midi{s'}{i+1} \supseteq \midi{r'}{i+1}$).
 \end{itemize}
 
 We conclude, by \Cref{lem-determinacy}, that the "facts" "matching" $B$
 in the two solutions are "equivalent", "ie", the "fact" "matching" $B$ in
 $(\bar \gamma a \bar \delta)$ is "equivalent" to $b$ (which is the "fact" "matching"
 $B$ in $\bar u \bar c \bar \beta$).

 In $t$ the unique "fact" "equivalent" to $b$ is $b'$ (since
 $b' \in \bar p' \subseteq t$), thus the "fact" "matching" $B$ in
 $(\bar \gamma a \bar \delta)$ is $b'$.  We have thus proved that any solution
 in $t$ containing the last "fact" $a$ of $p$ also contains $b'$.

 \proofcase{\ref{ru}} The following claim, together with "strong
 $i$-minimality" of $r'$ and $r'\models q(\bar u b \bar\gamma)$ for some
 $\bar \gamma$, shows that
  \begin{enumerate}[(I)]
  	\item $s'$ is also
  "strong $i$-minimal", 
  \item $s' \models \ql{i}(\bar u)$, and 
  \item \label{it:bprime-in-sol}
  $s' \models q(\bar u b' \bar \delta)$ for some $\bdelta$.
  \end{enumerate}
  
  \begin{clm}\label{claim-strong-delta}
  Assume $\HPtaui \tau i$. Let $s$ be a "strong $i$-minimal" "repair" such that
  $s\models q(\bar \alpha a \bar \beta)$ where $\bar \alpha$ "matches"
  $S_{\leq i}$ and $a$ is "$(i+1)$-compatible".  
  Then for any $a' \simblock a$ we have that $s' = s[ a \mapsto a']$ is
  "strong $i$-minimal" and $s'\models q(\bar \alpha a'\bar \delta)$ for some $\bar\delta$.
\end{clm}
\begin{proof}
  Notice that $s$ and $s'$ agree on all their solutions to $q$ that contain
  neither $a$ nor $a'$. Hence, if $a'$ is in no solution for $q$ in $s'$, we
  have $\ql{i+1}(s') \subsetneq \ql{i+1}(s)$ and therefore
  $|\ql{i+1}(s')| < |\ql{i+1}(s)|$ and $\ql{i}(s') \subseteq \ql{i}(s)$.  The
  latter implies $\ql{i}(s') = \ql{i}(s)$ by "$i$-minimality" of $s$.  This
  contradicts "strong $i$-minimality" of $s$.  Hence
  $s'\models q(\bar \alpha' a' \bar \beta')$, for some $\bar \alpha'$,
  $\bar \beta'$. By $\HPtaui \tau i$ this implies that
  $s'\models q(\bar \alpha a' \bar \delta)$ for some $\bar\delta$.

  It remains to prove that $s'$ is "strong $i$-minimal".  To this end, we
  exhibit a bijection from $\ql{i+1}(s)$ to $\ql{i+1}(s')$ preserving the
  $(S_1 \cup \dotsb \cup S_i)$-projection of the solutions. The existence of
  such bijection implies $\ql{i}(s) = \ql{i}(s')$ and
  $|\ql{i+1}(s)| = |\ql{i+1}(s')|$, thus showing that $s'$ is "strong
  $i$-minimal", provided $s$ is too.
  
  The mapping is the identity for the solutions that do not contain the "fact"
  $a$.

  It remains to map bijectively solutions in $s$ containing the "fact" $a$ to
  solutions in $s'$ containing the "fact" $a'$.  Let $\bar a$ (resp. $\bar a'$)
  be the "facts" "matching" $S_{i+1}$ in $\bar \alpha a \bar \beta$
  (resp. $\bar \alpha' a' \bar \beta')$.  By \Cref{cor-determinacy} in
  all solutions of $s$ containing $a$, $S_{i+1}$ is "matched" by $\bar
  a$. Similarly in all solutions of $s'$ containing $a'$, $S_{i+1}$ is
  "matched" by $\bar a'$.

  Moreover, by $\HPtaui \tau i$ for each $\bar u$,
  $\bar u\bar a \in q_{\leq i+1}(s)$ if{f}
  $\bar u\bar a' \in q_{\leq i+1}(s')$. Hence mapping each
  $\bar u \bar a \in \ql{i+1}(s)$ to $\bar u\bar a' \in \ql{i+1}(s')$ forms a
  bijection.
\end{proof}

\proofcase{\ref{suf}} Let $\bar e$ be the maximal suffix of $\bar p'$ such
that, for a "strong $i$-minimal" "repair" $t$ containing $\bar p'$ and
$\midi{s'}{i+1}$ we have $t \models q(\bar u \bar e \bar\delta)$ for some
$\bar \delta$.  Since $s' \models q(\bar u b' \bar \delta')$ for some
$\bar \delta'$ by Item~\ref{it:bprime-in-sol} above, $\bar e$ cannot be
empty.  Then let $\bar e = \bar d b'$, where $\bar d$ is a suffix of $\bar p$.
  
Since $\midi{s'}{i+1}$ contains $\midi{r}{i+1}$ by \eqref{hyp:contr:h4}, in
particular $t$ is a "strong $i$-minimal" "repair" containing $\bar p$ and
$\midi{r}{i+1}$.  Then, by maximality of $\bar c$, $\bar d$ must be a suffix of
$\bar c$, implying that $\bar u \bar d b'$ is a subset of $\bu \bc
\bbeta'$. Since by definition $\bu\bc \bar \beta'$ does not contain any
"$k$-set" in $\Delta_k(q,D)$, we have $\bar u \bar d b' \not\in \Delta_k(q,D)$ as
needed.
 
 This completes the proof of \Cref{induc-prop}.
\end{proof}


\section{Lower bounds}
\label{section-sjf-lower-bound}
\label{sjf-fails-cqk}


Now we turn to the lower bounds. In particular, we consider the
"self-join-free" queries $q$ falsifying $\HP$  and show that for such queries, $\certain(q)$ is \conp-hard. 
Towards proving this we also obtain other results  of independent interest. The results can be
summarized as follows.

\begin{enumerate}
  \item We first consider the query with "self-join"
    \AP
  \[
    \intro*\qFour \defeq R(\underline{x}\, yz) \land R (\underline{z}\, xy).
  \]

  We first show, in \Cref{thm-X+lowerBound}, via a combinatorial argument, that there
  is no $k$ such that $\certain(\qFour)$ can be computed using $\Cqk(\qFour)$
  ("ie", for every $k$ there exists a database $D$ such that
  $D\models \certain(q)$ but $\Cqk(\qFour)$ ouputs a false negative for $D$).
  
  On the other hand we prove, in \Cref{thm:q4-sbm-complete}, that $\certain(\qFour)$ is
  equivalent (modulo \logspace reductions) to a "matching problem@SBM problem"
  whose precise complexity is a long-standing open problem. This confirms the
  difficulty to obtain a complete complexity classification of $\certain(q)$
  for "self-join" queries $q$ even on two atoms.\footnote{The dichotomy for two-atom queries has recently been proved
  in \cite{ourPODS}. Remarkably, the polynomial cases are solved via a combination of
  the $\Cqk$ algorithm and the "bipartite matching algorithm@SBM problem".}

  \item In \Cref{sec:caseq5} we show that for the
  "self-join-free query" 
  \AP
  \[
    \intro*\qFive \defeq R_1(\underline{x}\, y) \land S_1(\underline{y z}\, x),
  \]
  $\certain(\qFive)$ 
  cannot be computed using $\Cqk(\qFive)$ for any choice of $k$. This is shown
  by reducing the case of $\qFive$ to the case of $\qFour$.
\item   In \Cref{sec:koutriswijsen-hard} we describe some of the techniques developed
  in~\cite{DBLP:journals/tods/KoutrisW17} and show that they imply that for
  "self-join-free" queries $q$, $\certain(q)$ is
  \conp-hard when $\HP$ fails for $q$. In particular, assuming $\ptime\neq\conp$, this
  implies that when $q\not\models\HP$, $\certain(q)$ cannot be solved using $\Cqk(q)$,
  for any choice of $k$.

\item Finally, in \Cref{sec:caseanyquery} we prove, without any complexity theoretic
  hypothesis, that for any "self-join-free query" $q$ such that
  $q\not\models\HP$, that $\certain(q)$ 
  cannot be computed using $\Cqk(q)$ for any choice of $k$ by 
  reducing the case of $\qFive$ to such query $q$.
\end{enumerate}

%


\newcommand{\XX}{\mathbb{X}}

\subsection{The case of \texorpdfstring{$\qFour$}{q4}}
\label{sec:caseq4}
\AP
Note that the query $\qFour$ is not a "self-join-free query". In particular, we
can have $D \models \qFour(aa)$ for some fact $a\in D$. We call such "solutions" ""self-loops"".

To prove that $\Cqk(\qFour)$ does not compute $\certain(\qFour)$ for any choice
of $k$, we actually prove a stronger statement: not even an extension
``$\Cqkp(\qFour)$'' of $\Cqk(\qFour)$ can capture $\certain(\qFour)$. This
stronger form will be needed later for reducing $\qFive$ to $\qFour$.

Let us first explain the extension, which is tailored to two-"atom" queries
(thus, strictly speaking, it `extends' $\Cqk(\qFour)$ only in this
context). Recall that the definition of $\Cqk(q)$ iteratively adds a "$k$-set"
$S$ to $\Deltak(q,D)$ if there exists a "block" $B$ of $D$ such that for every
"fact" $u\in B$ there exists $S' \subseteq S\cup \{ u\}$ such that
$S' \in \Deltak(q,D)$.  \AP This rule is henceforth called the \AP""first
derivation rule"".

\AP The algorithm $\intro*\Cqkp(q)$ initializes the set $\intro*\Deltakp(q,D)$
of "$k$-sets" as in $\Cqk(q)$, and inherits the aforementioned derivation rule
of $\Cqk(q)$, but it also contains the following \AP""second derivation rule"".
A "$k$-set" $S$ is
also added to $\Deltakp(q,D)$ if there exists a "fact" $a$ of $D$ which 
is not a "self-loop" and for
every "fact" $u \in \set{ b \in D : D \models q(ab) \lor q(ba) \lor a=b }$
there exists $S' \subseteq S\cup \set{u}$ such that $S' \in \Deltakp(q,D)$. 
As before, we define $\Cqkp(q)$ to accept if the empty set is eventually derived and denote it by $D\in \Cqkp(q)$.

Note that, like $\Cqk(q)$, $\Cqkp(q)$ also runs in polynomial time
{but it no
  longer satisfies the inductive property \eqref{eq:Delta:invariant} and may
  therefore give false positive answers. However, since the "first derivation rule" is present in both $\Cqk(q)$ and $\Cqkp(q)$,
  whenever $\emptyset \in \Cqk(q)$ we also have $\emptyset \in \Cqkp(q)$, "ie", if $D\in \Cqkp(q)$ then $D\in \Cqk(q)$.}  
  We  now formally  state the first result of this subsection:
\begin{thm}
\label{thm-X+lowerBound}
  \AP
  For every choice of $k$, there exists a database $D$ such that $D\models \certain(\qFour)$ but $D\not\in \Cqkp(\qFour)$ (and hence $D\not\in\Cqk(\qFour)$).
\end{thm}

Before proving \Cref{thm-X+lowerBound} we discuss some special properties of the query
$\qFour$. Note that $\qFour$ is a "self-join" two-"atom" query. 
\AP
We define the ""solution graph"" of $D$, denoted by $\intro*\solgraph$, to be an undirected graph whose vertices are the "facts" of $D$ and it contains an edge $\set{a,b}$ whenever $D \models \qFour(ab)$ holds. In the context of "solution graphs", a ""triangle"" is just a clique on three vertices (without self-loops).

\begin{remark}\label{rk:properties-solgraph-q4}
We first state some key properties of $\qFour$ for any "database" $D$ and "facts" $a,b,c\in D$, which are easy to verify:
\begin{enumerate}
  \item if $D \models \qFour(ab) \land \qFour(ac)$ then $b= c$,
  \item if $D \models \qFour(ba) \land \qFour(ca)$ then $b= c$,
  \item if $D\models \qFour(ab) \land \qFour(bc)$ then $D\models \qFour(ca)$,
  \itemAP hence, every connected component of $\solgraph$ is either a "triangle", a 2-clique (without self-loops), or a single vertex (with or without a self-loop).\label{shape-solgraph-q4}
\end{enumerate}
\end{remark}

\knowledgenewrobustcmd{\Dn}[1][n]{\cmdkl{D_{#1}}}
Towards proving \Cref{thm-X+lowerBound}, for every $n \geq 4$ we exhibit a "database"
$\Dn$ such that
\begin{itemize}
  \item $\Dn \models \certain(\qFour)$ (\Cref{prop-Dk in certain}), and
  \item $\Dn[n] \not\in \Cqkp[n-2](\qFour)$ (\Cref{prop-Dk not in Delta(k-2)}).
\end{itemize}
Intuitively, the "database" $\Dn$ has two kinds of "blocks". The first kind has $n$ "blocks" denoted by
$B_1,\cdots,B_n$ where each $B_i$ consists of $n-1$ "facts" denoted
$b_i^1,\cdots,b_i^{n-1}$. The second kind has $(n-1)(n-3)$ "blocks" denoted by $E^j_l$ for every $1\leq j \leq n-1$ and $1\leq l \leq n-3$, where each
$E_l^j$ has two "facts" denoted by $u_l^j$ and $v_l^j$.  The "solution graph" of
$\Dn$ is depicted in \Cref{fig:counterexample}.

  \begin{figure}
    \centering
    \includegraphics[width=\textwidth]{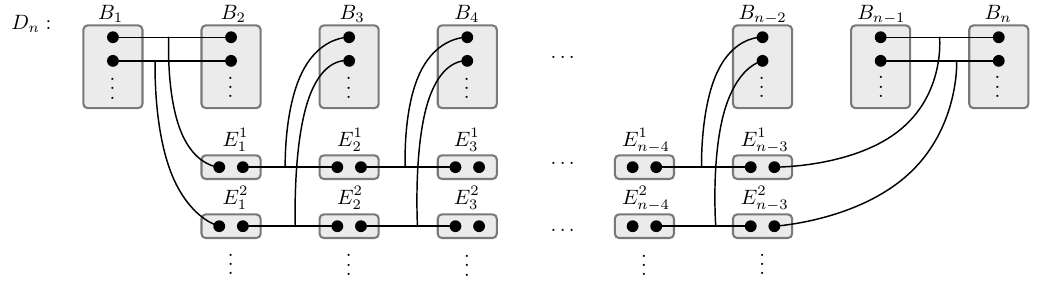}
    \caption{"Solution graph" for "database" $\Dn$. Black dots denote "facts",
      rectangles denote "blocks", and three-pointed edges denote "triangles" ("ie", 3-cliques) in the
      "solution graph" of $\Dn$.  There are $n-1$ "facts" in each "block" $B_i$ and two "facts" in each "block" $E^j_k$.}
    \label{fig:counterexample}
  \end{figure}

\AP
\paragraph{Definition of $\intro*\Dn$}
We now formally define the "facts" of the "database" $\Dn$. Fix some $n\ge 4$.
Define $n$ "blocks" of the form
$B_1,\cdots,B_n$ in $\Dn$
where each $B_i$ consists of $n-1$ "facts" denoted
$b_i^1,\cdots,b_i^{n-1}$ and $(n-1)(n-3)$ "blocks" of the form
$E^j_l$ for $1\leq j \leq n-1, 1\leq l \leq n-3$, where each $E_l^j$ have two "facts" denoted by
$u_l^j$ and $v_l^j$. All "facts" in $\Dn$ are $R$-"facts", and for this reason we shall henceforth drop the $R$ from "facts" $R(\bar c)$ and simply write $\bar c$.

Let $a_1,\cdots,a_n$ and  $e_i^j$ where $1\leq j \leq n-1, 1\leq i \leq n-3$ be fresh active domain elements.
The "facts" of $\Dn$ are defined as follows.
\begin{itemize}
\item $b_1^j=(\underline{a_1}~~ a_2e^j_1)$,
  $b_2^j = (\underline{a_2}~~e^j_1a_1)$;
\item for every $3 \le i \le n-2$:~~~
  $b^j_i = (\underline{a_i}~~e^j_{i-2}e^j_{i-1})$;
\item $b^j_{n-1} = (\underline{a_{n-1}}~~ e^j_{n-3}a_{n})$ and
$b^j_{n} = (\underline{a_{n}} ~~ a_{n-1}e^j_{n-3})$;
\item $u_1^j=(\underline{e^j_1}~~a_1a_2)$ and
  $v_{n-3}^j=(\underline{e^j_{n-3}}~~a_{n}a_{n-1})$;
\item for every $1\le l < n-3$:~~ $v_l^j=(\underline{e^j_l}~~e^j_{l+1}a_{l+2})$
  and $u_{l+1}^j=(\underline{e^j_{l+1}}~~a_{l+2}e^j_{l})$.
\end{itemize}

As shown in \Cref{fig:counterexample:triangles}, it can be verified that for every $1\le j \le n-1$ we have the "triangles" $\{b^j_1,b^j_2,u^j_1\}$ and
$\{b^j_{n-1},b^j_{n}, v^j_{n-3}\}$ and for every $1\le l < n-3$ we have a "triangle"
$\{v^j_l,u^j_{l+1},b^j_{l+2}\}$. Thus we get the "solution graph" described in Figure~\ref{fig:counterexample}. Notice that $\Dn$ is defined in such a way that every "fact" of $\Dn$ is part of a
"triangle". 

\begin{figure}
  \centering
  \includegraphics[width=\textwidth]{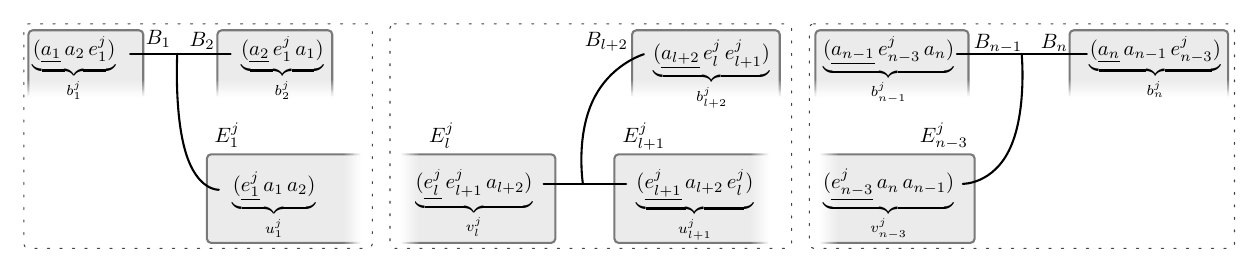}
  \caption{"Triangles" in the "database" $\Dn$.}
  \label{fig:counterexample:triangles}
\end{figure}

\begin{prop}
\label{prop-Dk in certain}
For every $n\ge 4$, $\Dn \models \certain(\qFour)$.
\end{prop}
\begin{proof}
Fix some $n\ge 4$ and consider the "database" $\Dn$. Let $B_1,\ldots B_n$ be the "blocks" of $\Dn$. Now for every $1\le j \le n-1$ if a "repair" $r$
  contains two indices $i,i'$ such that $i\ne i'$ and $b^j_i, b^j_{i'}\in r$ then it is easy to
  verify that $r\models q$ (\cf~Figure~\ref{fig:counterexample}).
  
  Now since each $B_i$ contains exactly $n-1$ "facts" of the form
  $b^1_i,\ldots ,b^{n-1}_i$, by the pigeonhole principle, every "repair" $r$ must
  contain two $i,i'$ such that $i\ne i'$ and $b^j_i, b^j_{i'}\in r$. Hence, every "repair" $r$ of $\Dn$ verifies $r\models q$.
\end{proof}

It remains to prove the following:
\begin{prop}
\label{prop-Dk not in Delta(k-2)}
Let $k\geq 2$. $\Dn[k+2] \not\in \Cqkp(\qFour)$.
\end{prop}

%
For showing this, we first need to set up some definitions and prove some useful properties.
\AP
\knowledgenewrobustcmd{\allB}{\cmdkl{\mathbb{B}}}
When $\Dn$ is clear from the context, we denote $\intro*\allB \defeq \{B_1,\ldots B_n\}$ and, for every $1\le j \le n-1$, we denote
\knowledgenewrobustcmd{\allE}[1]{\cmdkl{\mathbb{E}^{#1}}}
\AP
$\intro*\allE{j} \defeq \{ E^j_l \mid 1\le l \le n-3\}$ 
where the $B_i$'s and $E^j_l$'s are the "blocks" of $\Dn$ defined above.

\AP
If $\XX =\{X_1, \dotsc, X_k\}$ is a set of "blocks" of $\Dn$, a set of "facts" $W = \{ w_1,\ldots w_k\}$ is called a ""partial repair"" of $\XX$
if $w_i\in X_i$ for every $1\le i \le k$. 
\knowledgenewrobustcmd{\Wprep}[1]{\cmdkl{W[}#1\cmdkl{]}}
\AP We denote by $\intro*\Wprep{\allE{j}}$ to be the set
of "facts" from $W$ in the "blocks" $\XX \cap \allE{j}$ (for $1\le j\le n-1$),
and $\reintro*\Wprep{\allB}$ to be the set of facts from $W$ in the "blocks" $\XX \cap \allB$.

Recall that for every $3\le l \le n-2$ and every $1\le j\le n-1$ we have $b^j_l \in B_l$
and the "triangle" $\{b^j_l, u^j_{l-1}, v^j_{l-2}\}$.
For each such $j$ and $l$ we define 
\knowledgenewrobustcmd{\back}{\cmdkl{U}}
\knowledgenewrobustcmd{\front}{\cmdkl{V}}
\AP
$\intro*\back(j,l) \defeq \{ u^j_k \mid 1\le k \le l-2\}$ and
$\intro*\front(j,l) \defeq \{ v^j_k \mid l-1 \le k \le n-3\}$, which are depicted in \Cref{fig:counterexample:U-V}. 

\begin{figure}
  \centering
  \includegraphics[width=\textwidth]{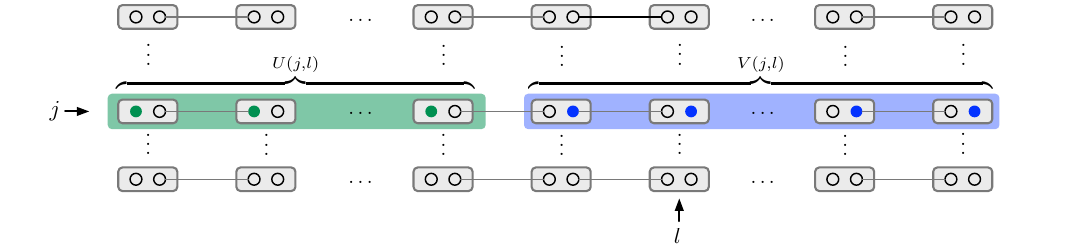}
  \caption{Depiction of $\back(j,l)$ as  green solid discs, and $\front(j,l)$ as  blue solid discs.}
  \label{fig:counterexample:U-V}
\end{figure}

Intuitively, $\back(j,l)$ and $\front(j,l)$ represent the "facts" that need to be picked in "blocks" of $\allE{j}$ if one wants to construct a "repair" of $\Dn$ containing $b_l^j$ and having no query "solutions".
In "fact" note that $\back(j,l) \cup \front(j,l)$ picks exactly one "fact" for every "block"
in $\allE{j}$. But $u^j_{l-1},~v^j_{l-2} \not\in \back(j,l) \cup
\front(j,l)$. Consequently, there are no "solutions" in the set of "facts"
$\back(j,l)\cup \front(j,l)\cup \{ b^j_l\}$ (\Cref{fig:counterexample:U-V}).

In the same spirit, recall that for every $1\le j\le n-1$, $\Dn$ contains the "triangle" $\{b^j_1, b^j_2, u^j_1\}$, so if a "repair" contains 
either $b^j_1$ or $b^j_2$ and contains no "solution" to $\qFour$, this "repair" must pick no $u^j_i$ "fact" in the "blocks" of $\allE{j}$. Thus, we define
 $\back(j,1) = \back(j,2) = \emptyset$ and
$\front(j,1) = \front(j,2) = \{ v^j_k \mid 1 \le k \le n-3\}$, so that there are no "solutions" in the set of "facts"
$\back(j,l)\cup \front(j,l)\cup \{ b^j_l\}$, for $l\in\{1,2\}$.

Dually,
$\{b^j_{n-1}, b^j_{n}, v^j_{n-3}\}$ forms a "triangle", so define
$\back(j,n) = \back(j,n-1) = \{ u^j_k \mid 1\le k \le n-3\}$ and
$\front(j,n) = \front(j,n-1) = \emptyset$. Again
there are no "solutions" in the set of "facts"
$\back(j,l)\cup \front(j,l)\cup \{ b^j_l\}$, for $l\in\{n-1,n\}$.

%
 Overall, for every $1\le l \le n$ and every $1\le j\le n-1$ we have:
 \begin{align*}
  \begin{aligned}
    &\text{$\back(j,l)\cup \front(j,l)\cup \{ b^j_l\}$ forms a "partial repair" over the
    "blocks" $\{B_j\} \cup \allE{j}$, and } \\
    &\text{$\back(j,l)\cup \front(j,l)\cup \{ b^j_l\}$ contains no "solution".} 
    \end{aligned}
 \end{align*}

 
\AP
A set of "facts" $W$ of size $k$ is called a ""$k$-obstruction set"" if $W$ is
a "partial repair" of some $\XX = \{X_1,\ldots X_k\}$ and if
$\XX \cap \allB = \{B_{i_1},B_{i_2}\ldots B_{i_l}\}$ for some $l\le k$ with
$\Wprep{\allB} = \{ b^{j_1}_{i_1}, b^{j_2}_{i_2}\ldots b^{j_l}_{i_l}\}$, then the
following conditions hold:

\begin{enumerate}
\item \label{condneq}  $j_1,j_2,\ldots j_l$ are pairwise distinct. 

{\noindent(Informally: elements of $W$ in $\allB$ are in different `rows' in \Cref{fig:counterexample}.)}

\item \label{condt} For every $1\le m \le l$ we have
  $\Wprep{\allE{j_m}} \subseteq \back(j_m,i_m) \cup \front(j_m,i_m)$.
  
{\noindent(Informally: if $W$ contains an element at row $j_m$ connected to the block $B_{i_m}$ then it must be painted green or blue in \Cref{fig:counterexample:U-V}, under the renaming $j \mapsto j_m$ and $l \mapsto i_m$.)}

\item \label{condj} For every $1\le j \le n-1$ there exists $1\le l \le n$ such that $\Wprep{\allE{j}} \subseteq \back(j,l) \cup \front(j,l)$.

{\noindent(Informally: all the elements of $W$ at a row $j$ must be painted green or blue in \Cref{fig:counterexample:U-V}, for some choice of indices.)}
\end{enumerate}

{
\begin{lemma}
\label{prop-extend-singleton-to-obst-set}
If $b$ is a "fact" in $D_{k+2}$ then there always exists a "$k$-obstruction set" that contains~$b$.
\end{lemma}
\begin{proof}
Let $B$ be the "block" such that $b\in B$.  Now if $B$ is of the form $B_i$ then let $b = b^j_i$ and we can choose the $k-1$ other blocks of the form $B_{i'}$ and pick $b^j_i$ from $B_i$ and pick $b^{j'}_{i'}$ from each chosen $B_{i'}$ such that \Cref{condneq} is satisfied which forms a "$k$-obstruction set" (this is always possible since each $B_l$ has $k+1$ facts).
 
Suppose $B$ is of the form $E^j_i$ then if $b = u^j_i$ then we can pick $\front(j,l-1) \cup \back(j,l-1) \cup \{b^j_{l-1}\}$ that forms an obstruction set for $\{B_j\} \cup \allE{j}$. If $b = v^j_i$ then we can pick $\front(j,l+1) \cup \back(j,l+1) \cup \{b^j_{l+1}\}$ that again forms an obstruction set for $\{B_j\} \cup \allE{j}$.
\end{proof}
}

\begin{lem}
\label{lemma-no solution in obstruction}
If $W$ is a "$k$-obstruction set", then there are no "solutions" to the query $\qFour$
within $W$.
\end{lem}
\begin{proof}
  Let $W$ be a "partial repair" of $\XX = \{X_1,\ldots X_k\}$. By
  Condition~\eqref{condj}, for every $1\le j \le n-1$ there are no "solutions"
  within $\Wprep{\allE{j}}$. By construction, it is also not possible to have
  "solutions" involving one "fact" from $\Wprep{\allE{j}}$ and another from
  $\Wprep{\allE{j'}}$ for $j\ne j'$ (\cf~Figure~\ref{fig:counterexample}).

  So if a "solution" exists, it has to involve some "fact" of the form
  $b^j_l \in \Wprep{\allB}$. By construction of $\Dn$ this "solution" must involve  either
  $u^j_{l-1}$ or $v^j_{l-2}$. But by 
  Condition~\eqref{condt} none of them are in $W$.
\end{proof}

\smallskip

We are now in shape to prove \Cref{prop-Dk not in Delta(k-2)}, which is a consequence of the following two claims.
  %
  \begin{clm}
    \AP\label{claim1} For every set of $k$ "blocks" $\XX = \{X_1,\ldots X_k\}$ of $\Dn[k+2]$, there exists  a "partial repair" $W$ of $\XX$ such that $W$ is a "$k$-obstruction set".
  \end{clm}
  \begin{clm}
    \AP\label{claim2} For every set of "facts" $W$, if $W$ is a "$k$-obstruction set", then $W\not\in \Deltakp(\qFour,\Dn[k+2])$.
  \end{clm}
  \begin{proof}[Proof of \Cref{prop-Dk not in Delta(k-2)}]
  Assuming that the two claims are true, they together imply that for every set of "blocks"
  $\set{X_1, \dotsc, X_k}$ of $\Dn[k+2]$ there exists a "partial repair" $W$ for
  $\set{X_1, \dotsc, X_k}$ such that $W\not\in \Deltakp(\qFour,\Dn[k+2])$. 
  
  Suppose $\emptyset \in \Deltakp(\qFour,\Dn[k+2])$ was obtained by using the "first derivation rule". Then, there is some block $B$ such that for every "fact" $a\in B$ we have $\set a\in \Deltakp(\qFour,\Dn[k+2])$. Let $\XX$ be an arbitrary set of blocks of size $k$ such that $B\in \XX$ (such $\XX$ can always be picked since $\Dn[k+2]$ has sufficiently many blocks). By \Cref{claim1} there exists a "$k$-obstruction set" $W$ which is a "partial repair" of $\XX$ and by \Cref{claim2} $W\not\in \Deltakp(\qFour,\Dn[k+2])$. This is a contradiction for the "fact" $a \in B\cap W$.  
  Thus, $\Dn[k+2]\not\models \Cqkp(\qFour)$. 

  If, on the other hand, $\emptyset \in \Deltakp(\qFour,\Dn[k+2])$ was obtained by using the "second derivation rule", then there is a "fact" $a$ such that for every $a' \in A=\set{ b \in \Dn[k+2] : \Dn[k+2] \models q(ab) \lor q(ba) \lor a=b }$ we have $\set{a'} \in \Deltakp(\qFour,\Dn[k+2])$. But from \Cref{prop-extend-singleton-to-obst-set} there is a "$k$-obstruction set" $W$ such that $a\in W$.  Moreover, from the two rules to compute $\Deltakp(\qFour,\Dn[k+2])$, it follows that if $S\in \Deltakp(\qFour,\Dn[k+2])$ and $S\subseteq S'$ where $|S'| \le k$ then $S'\in\Deltakp(\qFour,\Dn[k+2])$. Hence, $W\in \Deltakp(\qFour,\Dn[k+2])$ which contradicts \Cref{claim2}.
\end{proof}

Hence, we are only left with the proofs of \Cref{claim1,claim2}.

\begin{proof}[Proof of \Cref{claim1}] Recall that $\Dn[k+2]$ has $k+2$ "blocks"
of the form $B_i$ and each $B_i$ has $k+1$ "facts". Further, for every
$1 \le j \le k+1$ we have the sets of "blocks" of the from
$E^j_1\ldots E^j_{k-1}$.

First let $\XX \cap \allB = \{ B_{i_1},\ldots B_{i_m}\}$ for some $m \le k$ and
$i_1,i_2\ldots i_m \le k+2$. We let $W$ to contain
$\set{b^1_{i_1},b^2_{i_2},\ldots b^m_{i_m}}$. Moreover for every $1\le j \le m$
we add to $W$ the "partial repair" induced by $\back(j,i_j) \cup \front(j,i_j)$. Formally, we add $(\back(j,i_j) \cup \front(j,i_j)) \cap {\bigcup} \XX$ to $W$ ("ie", all "facts" of $\back(j,i_j) \cup \front(j,i_j)$ that are in a "block" of $\XX$). This ensures that conditions~\eqref{condneq} and~\eqref{condt} are satisfied.

Now for all $m < j \le k+1$ we add to $W$ the "partial repair" induced by
$\front(j,1)$ ("ie", $\front(j,1) \cap {\bigcup} \XX$).  Hence, condition~\eqref{condj} is also satisfied and $W$ is a "$k$-obstruction set".
\end{proof}

\begin{proof}[Proof of \Cref{claim2}] Let $\intro*\Deltakpi(\qFour,\Dn[k+2],i)$ be
  the set computed by the algorithm of $\Cqkp(\qFour)$ at step~$i$ of the fixpoint computation.

  Suppose the claim is false, and let $n$ be the least index such that
  $W\in \Deltakpi(\qFour,\Dn[k+2],n)$ for some "$k$-obstruction set" $W$ over the
  "blocks" $\XX=\set{X_1,\ldots, X_k}$.

  Note that $n=0$ is not possible since $W$ does not contain any "solution"
  (cf.\ \Cref{lemma-no solution in obstruction}); hence $n >0$. By definition
  of $\Deltakpi(\qFour,\Dn[k+2],n)$, this implies that there exists a set of "facts" $A$ such that either
  \begin{itemize}
    \item ("first derivation rule") $A$ is a "block" in $\Dn[k+2]$, or 
    \item ("second derivation rule") for some "fact" $a$ we have $D\not\models q(aa)$ and $A = \{a\} \cup \{b\mid D\models q(ab)\cup a(ba)\}$,
  \end{itemize}
   and for all
  $a'\in A$ there is a subset $W'\subseteq W \cup \{a'\}$ of size at most $k$ such
  that $W'\in \Deltakpi(\qFour,\Dn[k+2],n-1)$ and further $W'$ is not a
  "$k$-obstruction set", nor a subset thereof.

  Let $\XX \cap \allB = \{B_{i_1},\ldots B_{i_m}\}$ for some $m\le k$ and
  $\Wprep{\allB} = \{b^{j_1}_{i_1},b^{j_2}_{i_2},\ldots b^{j_m}_{i_m}\}$. So there
  exists $k+1-m$ many distinct indices $j$ that are not in
  $\set{j_1, \dotsc, j_m}$.  Let $j_{m+1}, \dotsc, j_{k}, j_{k+1}$ be those
  indices. Since $\XX$ is of size $k$, among those indices there can be at most
  $k-m$ many indices $s$ such that $\XX \cap \allE{s} \ne \emptyset$. Thus,
  there exists at least one index $s$ such that $s\not\in \set{j_1,\ldots j_m}$
  and $\XX \cap \allE{s} = \emptyset$.

\begin{figure}
  \centering
  \includegraphics[width=\textwidth]{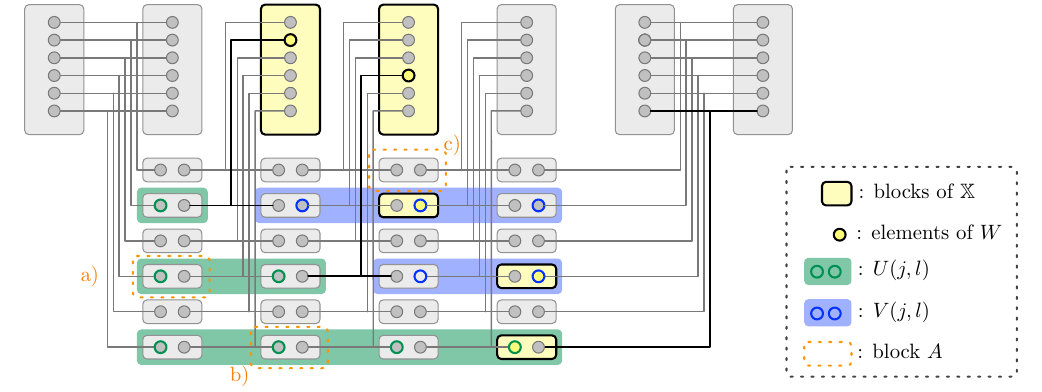}
  \caption{"Database" $\Dn[k+2]$ in the proof of Claim~\ref{claim2}, with  "$k$-obstruction set" $W$ over "blocks" $\XX$. Dotted-line boxes depict the "block" $A$ in cases a) b) and c) of the proof.}
  \label{fig:counterexample:3cases}
\end{figure}
Now we consider various candidates for $A$ and show that each case leads to a contradiction.
\begin{itemize}
\item If $A = B_l$ for some $1\le l\le k+2$ then consider $b_l^{s} \in
  B_l$. Observe that every "$k$-set" of size $k$ over $W \cup \{ b_l^{s}\}$ forms a
  "$k$-obstruction set" since condition~\eqref{condt} holds vacuously for
  $s$. Hence, there must exist a "$k$-obstruction set" inside
  $\Deltakpi(\qFour,\Dn[k+2],n-1)$, which is in contradiction with the
  minimality of $n$.

\item If, otherwise, $A=E^\lambda_l$ for some $1\le \lambda \le k+1$ and $1\le
  l \le k-1$, there are three cases to consider depending on where $A$ lies in relation to $\XX$ (see Figure~\ref{fig:counterexample:3cases}):
\begin{itemize}
\item[a)] If $\lambda \in \set{j_1,\ldots j_m}$, say $\lambda = j_t$, then by condition~\eqref{condt}
  $\Wprep{\allE{\lambda}} \subseteq \back(j_t,i_t) \cup \front(j_t,i_t)$. Let $w$ be the
  "fact" in $E^\lambda_l$ picked by $\back(j_t,i_t) \cup \front(j_t,i_t)$. Then every
  $k$-sized subset of $W\cup \{ w\}$ is a "$k$-obstruction set" (all conditions
  follow since $W$ is already a "$k$-obstruction set"). This is in contradiction with the minimality of $n$.

\item[b)] If $\lambda \not\in \set{j_1,\ldots j_m}$ but $\XX \cap \allE{r} \ne \emptyset$ then by
  condition~\eqref{condj}, $\Wprep{\allE{\lambda}} \subseteq \back(\lambda,\lambda')\cup \front(\lambda,\lambda')$ for some
  $1\le \lambda'\le k+2$. Let $w$ be the "fact" in $E^\lambda_l$ picked by
  $\back(\lambda,\lambda')\cup \front(\lambda,\lambda')$. Then, again every $k$-size subset of
  $W\cup \{ w\}$ is a "$k$-obstruction set" (all conditions follow since $W$ is
  already a "$k$-obstruction set"). This is a contradiction.

\item[c)] Otherwise, $\lambda\not\in \set{j_1,\ldots j_m}$ and $\XX \cap \allE{\lambda} = \emptyset$. In this
  case, pick $v^\lambda_l\in E^\lambda_l$ for some arbitrary $l$. We have
  $\{v^\lambda_l\} \subseteq \back(\lambda,1) \cup \front(\lambda,1)$. Hence, it can be verified
  that every $k$-sized subset of $W \cup \{v^\lambda_l\}$ is a "$k$-obstruction set",
  which again is a contradiction.
\end{itemize}
\item Otherwise $A =  \set{ b \in D : D \models q(ab) \lor q(ba) \lor a=b }$ for some fact $a\in \Dn[k+2]$ which implies that $A$ is of the form $\set{b^\lambda_l, u^\lambda_{l-1}, v^\lambda_{l-2}}$ or $\{b^\lambda_1,b^\lambda_2,u^\lambda_1\}$ or $\{b^\lambda_{n-1},b^\lambda_{n},v^\lambda_{n-3}\}$. We only prove the case where $A = \set{b^\lambda_l, u^\lambda_{l-1}, v^\lambda_{l-2}}$ and the other two cases are analogous. The argument is similar to the previous case.

  \begin{itemize}
\item[a)] If $\lambda\in \set{j_1,\ldots j_m}$, say $\lambda = j_t$, then by condition~\eqref{condt}
  $\Wprep{\allE{\lambda}} \subseteq \back(j_t,i_t) \cup \front(j_t,i_t)$. Let $w$ be the
  "fact" in $A$ picked by $\back(j_t,i_t) \cup \front(j_t,i_t)$. Then every
  $k$-sized subset of $W\cup \{ w\}$ is a "$k$-obstruction set" (all conditions
  follow since $W$ is already a "$k$-obstruction set"). This is in contradiction with the minimality of $n$.

\item[b)] If $\lambda\not\in \set{j_1,\ldots j_m}$ but $\XX \cap \allE{r} \ne \emptyset$ then by
  condition~\eqref{condj}, $\Wprep{\allE{\lambda}} \subseteq \back(\lambda,\lambda')\cup \front(\lambda,\lambda')$ for some
  $1\le \lambda'\le k+2$. Let $w$ be the "fact" in $A$ picked by
  $\back(\lambda,\lambda')\cup \front(\lambda,\lambda')$. Then, again every $k$-size subset of
  $W\cup \{ w\}$ is a "$k$-obstruction set" (all conditions follow since $W$ is
  already a "$k$-obstruction set"). This is a contradiction.

\item[c)] Otherwise, $\lambda\not\in \set{j_1,\ldots j_m}$ and $\XX \cap \allE{\lambda} = \emptyset$. In this
  case, pick $v^\lambda_{l-2}\in A$. We have
  $\{v^\lambda_{l-2}\} \subseteq \back(\lambda,1) \cup \front(\lambda,1)$. Hence, it can be verified
  that every $k$-sized subset of $W \cup \{v^\lambda_l\}$ is a "$k$-obstruction set",
  which again is a contradiction.\qedhere
\end{itemize}
\end{itemize}
\end{proof}

 As we have seen, $\certain(\qFour)$
cannot be computed using our fixpoint algorithm.  \AP We will next show that
$\certain(\qFour)$ is complete for the ""Saturating Bipartite Matching
problem"" (or \reintro{SBM problem} for short). This is the problem of, given
bipartite graph $(V_1 \cup V_2, E)$, whether there is an injective function
$f: V_1 \to V_2$ such that $(v,f(v)) \in E$ for every $v \in V_1$. The "SBM
problem"\footnote{In turn, the "SBM problem" is equivalent to the Bipartite
  Perfect Matching problem under \logspace reductions, which corresponds to the
  restriction to instances in which $|V_1|=|V_2|$.} is known to be in
\ptime~\cite{DBLP:journals/siamcomp/HopcroftK73}, hence $\certain(\qFour)$ can
be solved in polynomial time, although its precise complexity class is open.%

\begin{thm}
  \label{thm:q4-sbm-complete}
  $\certain(\qFour)$ is complete for the (complement of the) "SBM problem" under
  \logspace-reductions.
  In particular $\certain(\qFour)$ is in \ptime.
  \end{thm}
\begin{proof}
First we prove that there is a \logspace-reduction from $\certain(\qFour)$ to the (complement of the)
  "SBM problem".
Fix an input "database" $D$. We reduce $\certain(\qFour,D)$ to the "SBM problem".
 
First we check for "self-loops" in $D$. If there is a "fact" $a \in D$ such that
$D\models q(aa)$ then we first check if $a$ is a singleton "block". If so, then
$D \models \certain(\qFour)$. Otherwise, $D \models \certain(\qFour)$ iff
$D\setminus\set{a} \models \certain(\qFour)$, so we can consider a smaller
"database" and repeat the argument.  Moreover this pre-processing can be performed in \logspace.
So assume that there is no "fact"
$a \in D$ such that $ D\models \qFour(aa)$.
    

Now consider the bipartite graph $G = (V_1 \cup V_2, E)$ where $V_1$ is the set
of all "blocks" of $D$ and $V_2$ is the set of all maximal cliques in the "solution graph"
$\solgraph$. Note that, by property \eqref{shape-solgraph-q4} (in page
\pageref{shape-solgraph-q4}), $V_2$ forms a partition of $D$, namely the set of
all maximal connected components of $\solgraph$. Let $(v_1,v_2) \in E$ if the
"block" $v_1$ contains a "fact" which is in the clique $v_2$.

\AP
Suppose that there is a ""$V_1$-saturating matching"", that is, an injective
 function $f: V_1 \to V_2$ such that $(v_1,f(v_1)) \in E$ for every
 $v_1 \in V_1$. We construct a "repair" $r$ where for every "block" $B$ of $D$, we
 pick the "fact" (or one of the "facts", if there are more than one) which is in
 $f(B)$. In this way, no two chosen "facts" will be in the same clique, and also
 since there is no "solution" of the form $\qFour(aa)$ in $D$, no two chosen "facts"
 will form a "solution" to $\qFour$. Thus, $r\not\models \qFour$.

Conversely, if $\qFour$ is not "certain" in $D$, let $r$ be a "repair" such that $r \not\models \qFour$. For each "block" $B$ of $D$ let $r(B)$ be the "fact" of $B$ belonging to $r$.
Note that, since $V_2$ is a partition of $D$, each $r(B)$ belongs to a unique clique in $V_2$.
Define $f: V_1\to V_2$ such that each "block" $B \in V_1$ is mapped to the clique in $V_2$ where $r(B)$ lies. To verify that $f$ is a witness function of a 
"$V_1$-saturating matching" for $G$, note that for every $B \in V_1$ we have $(B, f(B)) \in E$, as $B$ and $f(B)$ both contain $r(B)$. Moreover, $f$ is injective, otherwise if $f$ maps two distinct "blocks" to the same clique, this clique must contain at least two "facts" $a,b$ from $r$. These two "facts" are neighbors in $\solgraph$, and then $r \models \qFour(ab)$ or $r \models \qFour(ba)$, contradicting the hypothesis $r \not\models \qFour$.

Thus, to check if $D\in \certain(\qFour)$, it is sufficient to check if there is a "$V_1$-saturating matching" for $\solgraph$.

\bigskip

For the other direction, 
  given a bipartite graph $G = (V_1\cup V_2, E)$, let $V_1 = \{ s_1,\ldots s_n\}$ and $V_2 = \{ t_1,\ldots t_m\}$. 
  We will define a "database" $D_G$ such that there exists a "$V_1$-saturating matching" in $G$ if{f} $\qFour$ is not "certain" in $D_G$. 
    
  \knowledgenewrobustcmd{\Neighbour}{\cmdkl{N}}
  \AP
  For all $s_j\in V_1$ let $\intro*\Neighbour(s_j) \subseteq V_2$ denote the neighbours of $s_j$ and similarly for all $t_i \in V_2$ let $\reintro*\Neighbour(t_i)\subseteq V_1$ denote the neighbours of $t_j$.
  
  First note that if there is some $s_j\in V_1$ such that $\Neighbour(s_j) = \emptyset$, then clearly there cannot be a "$V_1$-saturating matching". Similarly, if there is some $t_i\in V_2$ such that $\Neighbour(t_i) = \emptyset$, then $t_i$ does not contribute to any matching and hence can be removed from the input. Further, suppose there is some $t_i \in V_2$ such that $|\Neighbour(t_i)| = 1$, let $s_j$ be the single neighbour of $t_i$. In this case, in every "$V_1$-saturating matching" maps $s_j$ to $t_i$.  So we can remove the vertices $s_j$ and $t_i$ from the input graph and conside a smaller instance.  Note that all these checks can be done in \logspace. 
    Hence we assume that for every $u\in V_1\cup V_2$, $\Neighbour(u) \ne \emptyset$ and  $|\Neighbour(s_j)| \ge 1$ for all $s_j\in V_1$ and $|\Neighbour(t_i)| \ge 2$ for all $t_i\in V_2$.

  Now we define the  "database" $D_G$. Note that this construction is very similar to the construction of $\Dn$ that we used to prove Theorem \ref{thm-X+lowerBound}.
  \begin{itemize}
  \item For every vertex in $s_j\in V_1$ create a "block" $B_j$ in $D_G$.
  \item For every $s_j\in V_1$ and $t_i\in V_2$, if $t_i\in \Neighbour(s_j)$ then there is a "fact" denoted by $b^i_j$ in the "block" $B_j$. By assumption $\Neighbour(s_j) \ge 1$ and hence every "block" $B_j$ is non-empty.
  \item For every $t_i\in V_2$ if $|\Neighbour(t_i)| = l$ then let $s_{i_1},\ldots s_{i_l} \in V_1$ be the neighbours of $t$. By the above construction, for every $j\le l$,  there is a "fact" of the form $b^i_{i_j}$ in $B_{i_j}$ that corresponds to the vertex $t_i$.
  
  Now if $l=2$ then define $b^i_{i_1}$ and $b^i_{i_2}$ such that they form a "solution" to $\qFour$. Otherwise, if $l = 3$ then define $b^i_{i_1}, b^i_{i_2}$ and $b^i_{i_3}$ such that they pair-wise form a "solution" to $\qFour$ (the three facts form a "triangle").
  
  If $l\ge 4$ then create $l-3$ new "blocks" denoted by $E^i_1,\ldots E^i_{l-3}$ where each $E^i_j$ contains exactly two "facts" $u^i_j$ and $v^i_j$. Moreover, in the same way as described in the definition of $\Dn$ earlier (\cf~\Cref{fig:counterexample:triangles,fig:counterexample}), define the "facts" appropriately such that $\{b^i_{i_1},b^i_{i_2},u^i_1\}$ and $\{b^i_{i_{l-1}},b^i_{i_l},v^i_{l-3} \}$ form "triangles" and for every $1\le j< l-3$ we have a "triangle" $\{v^i_j,u^i_{j+1},b^i_{j+2}\}$.
  
  \end{itemize}
  The reader can verify that this is exactly the construction used to define $\Dn$. 
  \AP
  Again, this construction is in \logspace. For each such $j$ and $l$  define the analogous
  $\back(i,l) \defeq \{ u^i_k \mid 1\le k \le l-2\}$ and
  $\front(i,l) \defeq \{ v^i_k \mid l-1 \le k \le l-3\}$.
  
  Now suppose there is a "$V_1$-saturating matching" then let us show that $\qFour$ is not "certain" for $D_G$. Consider the "repair" $r$  where for each "block" $B_j$ we  pick $b^i_j$ if $s_j$ is matched with $t_i$. Further, pick $\back(i,l) \cup \front(i,l)$ which gives a "partial repair" over $E^i_1\ldots E^i_l$.
  
  If some $t_i \in V_2$ is not matched with any vertex in $V_1$ then pick $\back(i,1) \cup \front(i,1)$ which gives a "partial repair" over $E^i_1\ldots E^i_l$.  It can be verified that the obtained "repair" does not contain any "solution".
  
  Conversely, suppose there is a "repair" $r$ of $D_G$ that falsifies $\qFour$, and let us show that there is a "$V_1$-saturating matching" in $G$.
  For any such "repair" $r$, note that if $b^i_j$ is picked in "block" $B_j$ then for all other "blocks" $B_{j'}$, the "fact" $b^i_{j'}$ cannot be in $r$ since that would make $\qFour$ true. Also $b^i_j\in B_j$ only if there is an edge between $s_j$ and $t_i$. Hence, we can define the "$V_1$-saturating matching" that maps every $s_j\in V_1$ to $t_i\in V_2$, where  $b^i_j$ is the "fact" in $r$ from the "block" $B_j$.
  \end{proof}

\subsection{The case of \texorpdfstring{$\qFive$}{q5}}
\AP
\label{sec:caseq5}

We now show that $\certain(\qFive)$ cannot be computed by $\Cqk(\qFive)$. This
is shown by reduction to the case of $\qFour$ based on the following construction:

\begin{prop}
\label{theorem-UnconditionX+-Lowerbound-sjf}
For every "database"  $D$ over the "signature" of $\qFour$ we can construct a
"database" $D'$ over the "signature" of $\qFive$ such that:
\begin{enumerate}
\itemAP if $D\models \certain(\qFour)$  then $D'\models \certain(\qFive)$;\label{theorem-UnconditionX+-Lowerbound-sjf:1}
\itemAP for every $k\ge 2$, if $D'\models \Cqk(\qFive)$ then
  $D\models \Cqkp(\qFour)$. \label{theorem-UnconditionX+-Lowerbound-sjf:2}
\end{enumerate}
\end{prop}

\begin{proof}
  Let $D$ be a "database" over the "signature" of $\qFour$.
%
  Consider the "solution graph" $\solgraph$ of $D$. Recall that by property
  \eqref{shape-solgraph-q4}, every connected component in $\solgraph$ is always
  a clique of size at most $3$. Thus, every "fact" of $D$ can be part of
  exactly one maximal clique.

  Let $B_1,B_2\ldots B_m$ be the set of all "blocks" of $D$ and
  $C_1,C_2\ldots C_l$ be the set of cliques in $\solgraph$. Let
  $e_1,e_2\ldots e_m, f_1,f_2\ldots f_l$ be fresh and pairwise distinct domain
  elements.  Define $D'$ as follows:
\begin{itemize}
\item \label{item-uncond-preserves-certain} For every "fact"
  $u$ occurring in a block $B_i$ of $D$ and in the clique $C_n$ add the "fact"
  $u_{R_1}= R_1(\underline{e_i}\, f_n)$ to $D'$. Notice that there could be two
  distinct facts occurring in the same block of $D$ and in the same
  clique. They are then both associated to the same fact of $D'$.
  
\item
  \label{item-uncond-preserves-X+} For every two facts $u,v$ from a
  clique $C_n$ such that $D\models \qFour(uv)$ and $u,v$ are in two distinct
  "blocks" $B_i$ and $B_j$, create two new "facts" in $D'$ as
  $u^v_{S_1}=S_1(\underline{f_ng}~e_i)$ and
  $v^u_{S_1}=S_1(\underline{f_ng}~e_j)$, where $g$  is a fresh element depending only on $u$ and $v$.
\item
  \label{item-uncond-preserves-X+-2} For every fact $u$ from a
  clique $C_n$ and a block $B_i$ such that $D\models \qFour(uu)$ add a "fact" in $D'$ as
  $u^u_{S_1}=S_1(\underline{f_ng}~e_i)$, where $g$  is a fresh element depending only on $u$.
\end{itemize}
\begin{figure}[ht]
  \centering
  \includegraphics[scale=0.9]{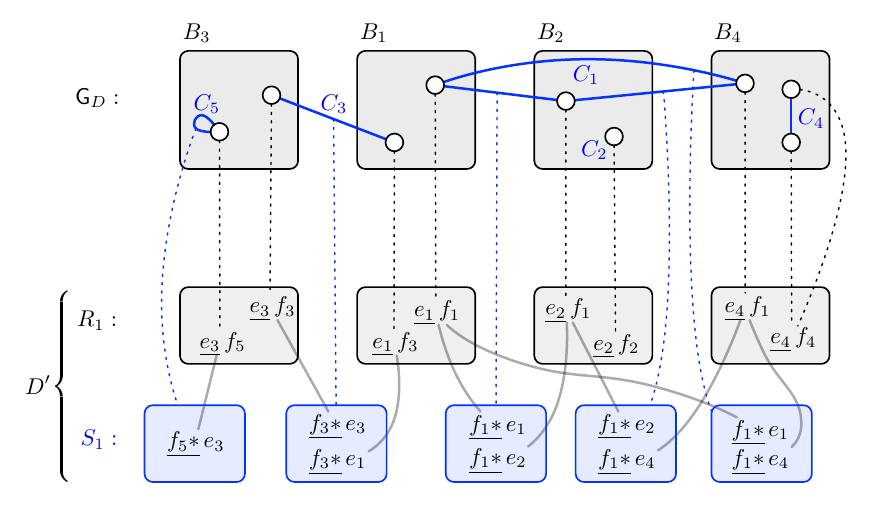}
  \caption{Example of construction of $D'$ from $D$ in the proof of \Cref{theorem-UnconditionX+-Lowerbound-sjf}.  The symbols ``$*$'' stand for the necessary constants ``$g$'' used in order to obtain the depicted "blocks". $D$ has $4$ "blocks" and $5$ cliques. Observe that neither the edge-less singleton clique $C_2$ nor the clique $C_4$ ---which is internal to the "block" $B_4$--- intervene in the relation $S_1$. Light gray edges depict the pairs of "facts" from $D'$ which form a "solution" to $\qFive$.}
  \label{fig:q5q4}
\end{figure}
See \Cref{fig:q5q4} for an example.

Notice that for any two "facts" $u,v \in D$ if
$D\models \qFour(uv)$ and $u,v$ are in two distinct blocks, then $u^v_{S_1}\simblock v^u_{S_1}$,
$D'\models \qFive(u_{R_1}u^v_{S_1}) \land \qFive(v_{R_1}v^u_{S_1})$, and the
"block" containing $u^v_{S_1},v^u_{S_1}$ in $D'$ does not contain any other
fact.
Moreover if for some fact  $u \in D$, $D\models \qFour(uu)$ then $D'\models
\qFive(u_{R_1}u^u_{S_1})$  and the
"block" containing $u^u_{S_1}$ in $D'$ does not contain any other
fact.
Finally, the only "solutions" in $D'$ for $\qFive$ are of the form
$\qFive(u_{R_1}u^v_{S_1})$ for some "facts" $u,v\in D$ such that
$D\models \qFour(uv)$ and $u \not\simblock v$ or $u=v$. Now we prove both claims of
the statement.

\medskip

\proofcase{\eqref{theorem-UnconditionX+-Lowerbound-sjf:1}} Suppose
$D\models \certain(\qFour)$. To verify that $D'\models \certain(\qFive)$, pick
any "repair" $r'$ of $D'$. By construction, every $R_1$-"fact" $a$ of $r'$ is of the form $u_{R_1}$
for some "fact" $u$ in $D$. Recall that there could be two such "facts" $u$. We
arbitrarily select one of them that we denote by $f_R(a)$, and we consider the set $r = \set{u \mid \exists a \in r' ~ u=f_R(a)}$. 

First, we prove that $r$
is a "repair" of $D$, that is: (i) it contains no two distinct facts $u \simblock v$, and (ii) it contains at least one "fact" for each "block" $B_i$.
\begin{enumerate}[(i)]
  \item By means of contradiction, suppose $r$ contains two "facts" $u,v\in r$
  such that $u\simblock v$. 
  By definition of $r$, $u=f_R(a)$
  and $v=f_R(b)$ for some "facts" $a,b$ of $r'$. 
  By construction of $D'$, it
  follows from $u\simblock v$ that $a\simblock b$ and therefore $a=b$ since we assumed $r$ to be a
  "repair". Hence, $u=v$.
  \item Further, for every
"block" $B_i$ in $D$ there is a $R_1$-"block" in $D'$ whose "facts" have "primary key"
$e_i$ and one of those "facts" is selected in $r'$, which gives rise (via $f_R$) to a "fact" of $B_i$.
\end{enumerate}

Since $D\models \certain(\qFour)$ by hypothesis, there are $u,v\in r$ such that $D\models \qFour(uv)$, let us show that this implies $r' \models \qFive$:
\begin{itemize}
  \item If $u=v$ then, by definition of $r$, we have $u_{R_1} \in r'$. Since $u^u_{S_1}$ belongs to a singleton $S_1$-"block" of $D'$, it is in every "repair" of $D'$, and thus $r'\models \qFive(u_{R_1} u^u_{S_1})$ as desired.
  \item If $u\neq v$ then, by definition of $r$, we have $u_{R_1},v_{R_1}\in r'$, and $\set{u^v_{S_1}, v^u_{S_1}}$ form a
$S_1$-"block" of $D'$. Hence, $r'$ contains one of
$\set{u^v_{S_1}, v^u_{S_1}}$. If $u^v_{S_1}\in r'$ we have
$r'\models \qFive(u_{R_1}u^v_{S_1})$; otherwise, if $v^u_{S_1}\in r'$, we have
$r'\models \qFive(v_{R_1}v^u_{S_1})$.
\end{itemize}
This concludes the proof for item \eqref{theorem-UnconditionX+-Lowerbound-sjf:1}.

  \bigskip

  \proofcase{\eqref{theorem-UnconditionX+-Lowerbound-sjf:2}}
Fix some $k\ge 2$.  Let $\reintro*\Deltakp(\qFour,D,i)$ and $\reintro*\Deltaki(\qFive,D',i)$ be the
  sets computed at step $i$ of $\Cqkp(\qFour)$ and $\Cqk(\qFive)$
  respectively. First we prove the following claim.

\begin{clm}\label{claim-R-block-only}
  Let $u,v,w$ be "facts" of $D$ in the same clique of the "solution graph"
  $\solgraph$. Possibly $v \simblock w$ or $v=w$ but we assume
  $u\not\simblock v$ and $u\not\simblock w$. In particular
  $D\models \qFour(uv) \lor \qFour(vu)$.
  Then for every "$k$-set" $S$ over $D'$
  if
  $S\cup \set{u^v_{S_1}} \in \Deltaki(\qFive,D',i)$ then
  $S \cup \set{v_{R_1}} \in \Deltaki(\qFive,D',i+1)$ and
  $S \cup \set{w_{R_1}} \in \Deltaki(\qFive,D',i+1)$. 
\end{clm} 
Note that in the case where $v_{R_1}\in S$ the claim implies
$S\in \Deltaki(\qFive,D',i+1)$. Similarly for $w_{R_1}$.

\begin{proof}[Proof of \Cref{claim-R-block-only}] We assume that  $D\models \qFour(uv)$, since the case where
  $D\models \qFour(vu)$ is symmetric. To facilitate the understanding of the following proof, refer to \Cref{fig:claim-R-block-only}.
  \begin{figure}
    \centering
    \includegraphics[scale=.9]{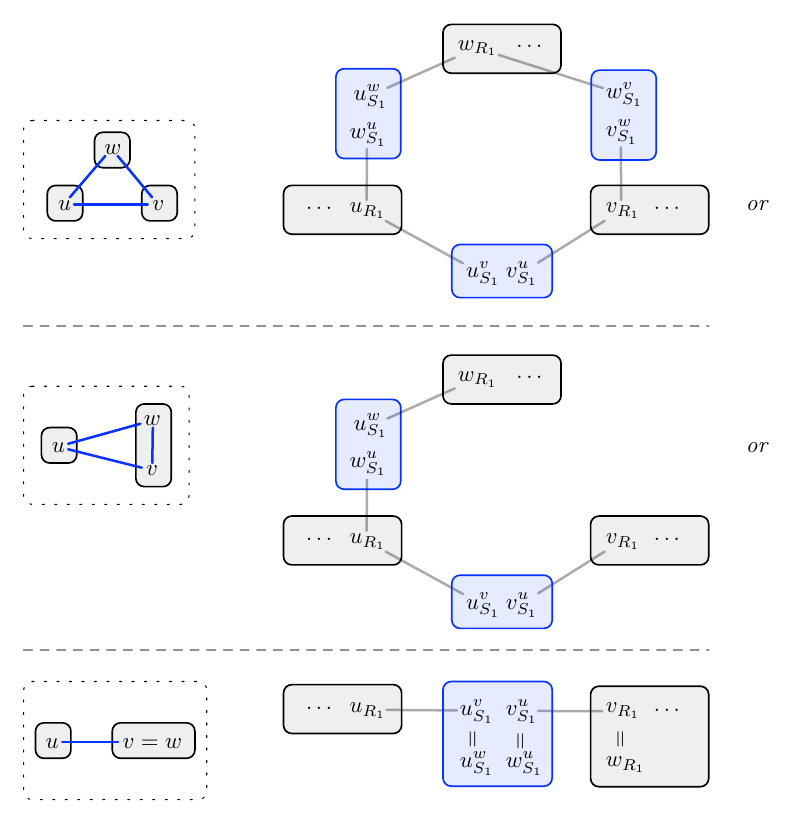}
    \caption{Schema for proof of \Cref{claim-R-block-only}.}
    \label{fig:claim-R-block-only}
  \end{figure}

  Assume $S \cup \set{u^v_{S_1}} \in \Deltaki(\qFive,D',i)$. If
  $S \in \Deltaki(\qFive,D',i)$ then the conclusion is immediate since
  $S \cup \set{u} \in \Deltaki(\qFive,D',i)$ for every "fact" $u$ such that
  $S\cup \set{u}$ is a $k$-set. So, let us assume
  $S \not\in \Deltaki(\qFive,D',i)$.
  
  \begin{itemize}
  \item The base case $i=0$ implies
    $S\cup\set{u^v_{S_1}}\in \Deltaki(\qFive,D',0)$ and
    $S\not\in\Deltaki(\qFive,D',0)$. This can only happen if $u_{R_1}\in S$,
    since $\set{u^v_{S_1},u_{R_1}}$ is the only $\qFive$-"solution" of $D'$
    involving $u^v_{S_1}$.  Let $B$ be the block
    $B=\set{u^v_{S_1}, v^u_{S_1}}$. Since $\set{u_{R_1},u^v_{S_1}}$ and
    $\set{v^u_{S_1},v_{R_1}}$ are $\qFive$-"solutions"
    (\cf~\Cref{fig:claim-R-block-only}), a simple analysis shows that the block
    $B$ is a witness for $S\cup\set{v_{R_1}} \in
    \Deltaki(\qFive,D',1)$. Similarly the block $B'=\set{u^w_{S_1},w^u_{S_1}}$
    witnesses the fact that $S\cup\set{w_{R_1}} \in
    \Deltaki(\qFive,D',1)$. 

  \item For the induction step, suppose
    $S\cup\set{u^v_{S_1}} \in \Deltaki(\qFive,D',i)$ and $i>0$. As $i>0$, let $B$ be the "block" of
    $D'$ used as the witness. By definition, for all "fact" $b\in B$, there is
    a subset $S_b$ of $S$ such that
    $S_b \cup \set{u^v_{S_1},b} \in \Deltaki(\qFive,D',i-1)$. By induction this
    implies that for all $b\in B$,
    $S_b \cup \set{v_{R_1},b} \in \Deltaki(\qFive,D',i)$. Hence $B$ witnesses
    the fact that $S\cup\set{v_{R_1}}\in\Deltaki(\qFive,D',i+1)$, as
    desired. The proof that $S\cup\set{w_{R_1}}\in\Deltaki(\qFive,D',i+1)$ is similar. \qedhere
  \end{itemize}
  \end{proof}

  With \Cref{claim-R-block-only} in place, assume that $D' \models \Cqk(\qFive)$ and let us show
  $D\models \Cqkp(\qFour)$.
  Let $S$ be a "$k$-set" of $D'$. Let $\tilde S$ be 
  any "$k$-set" of
  $D$ containing every "fact" $u$ such that either $u_{R_1} \in S$ or $v^u_{S_1} \in S$ for some "fact" $v$. 

  We will prove, by induction on $i$, that if $S\in \Deltaki(\qFive,D',i)$ then $\tilde S\in\Deltakpi(\qFour,d,i+1)$ (hence, in particular, $D'\models \Cqk(\qFive)$ implies $D\models \Cqkp(\qFour)$, concluding the proof of \eqref{theorem-UnconditionX+-Lowerbound-sjf:2}).
  \begin{itemize}
    \item For the base case $i=0$, if $S\in \Deltaki(\qFive,D',0)$ then it contains a "solution" to $\qFive$
    of the form $u_{R_1}u^v_{S_1}$ for some facts $u,v$ of $D$. By construction
    $D\models \qFour(uv)$ and $u,v\in \tilde S$. Hence
    $\tilde S\in\Deltakpi(\qFour,d,0)$.
    \item For the induction step, assume $S\in \Deltaki(\qFive,D',i)$ as witnessed by a
  "block" $B'$ of $D'$. 
  \begin{itemize}
    \item If $B'$ is an $R_1$-"block" then the "block"
    $\tilde B=\set{u\mid u_{R_1}\in B'}$ is a witness for
    $\tilde S\in\Deltakpi(\qFour,D,i)$. 
    \item If $B'$ is an $S_1$-"block" of the form
    $\set{u^v_{S_1},v^u_{S_1}}$. Set $w$ to the third "fact" in the clique of
    $u,v$ in the "solution graph" $\solgraph$ (set $w=v$ if the clique has size
    two). From \Cref{claim-R-block-only} we get that each of
    $S\cup\set{u_{R_1}}$, $S\cup\set{v_{R_1}}$, $S\cup\set{w_{R_1}}$ contains a
    "$k$-set" in $\Deltaki(\qFive,D',i-1)$. Hence, by induction each of
    $\tilde S\cup\set{u}$, $\tilde S\cup\set{v}$, $\tilde S\cup\set{w}$ contains
    a "$k$-set" in $\Deltakpi(\qFour,D,i)$. 
    Hence, by $\Deltakp$'s "second derivation rule" on the "fact" $u$,
    $\tilde S\in\Deltakpi(\qFour,D,i+1)$.
    \item It remains to consider the case where $B'$ is an $S_1$-"block" of the form
  $\set{u^u_{S_1}}$. But then $u$ is a "self-loop" and is contained in $\tilde
  S$. By definition we then have $\tilde S \in\Deltakpi(\qFour,D,0)$. \qedhere
  \end{itemize}
  \end{itemize}
\end{proof}

\begin{thm}\label{thm-q5-cqk}
 $\certain(\qFive)$ cannot be computed by
  $\Cqk(\qFive)$, for any choice of $k$.
\end{thm}
\begin{proof}
  By means of contradiction, assume that $\certain(\qFive)$ is equivalent to $\Cqk(\qFive)$ for some $k \geq 2$. We show that this implies that $\certain(\qFour)$ is equivalent to $\Cqkp(\qFour)$, contradicting \Cref{thm-X+lowerBound}. 
  
  It is enough to show that if a "database" $D$ is such that $D\models\certain(\qFour)$ then
  $D\models \Cqkp(\qFour)$. To this end, assume $D$ is a "database" such that $D\models\certain(\qFour)$. Let $D'$ be the "database" constructed from $D$
  given by \Cref{theorem-UnconditionX+-Lowerbound-sjf}. As
  $D\models\certain(\qFour)$, we get from
  \Cref{theorem-UnconditionX+-Lowerbound-sjf}-\eqref{theorem-UnconditionX+-Lowerbound-sjf:1} that
  $D'\models\certain(\qFive)$. From our hypothesis it follows that
  $D'\models\Cqk(\qFive)$. By \Cref{theorem-UnconditionX+-Lowerbound-sjf}-\eqref{theorem-UnconditionX+-Lowerbound-sjf:2} this
  implies that $D\models\Cqkp(\qFour)$ as desired.
\end{proof}

\subsection{The case of all self-join-free queries falsifying \texorpdfstring{\HP}{PCond}: \conp-hardness}
\AP
\label{sec:koutriswijsen-hard}
We describe in this section techniques
from~\cite{DBLP:journals/tods/KoutrisW17}, together with their immediate
consequence: $\certain(q)$ is \conp-hard as soon as $q\not\models\HP$.
We build on the dichotomy result of~\cite{DBLP:journals/tods/KoutrisW17} based
on their notion of an ``"attack graph"''. First we recall this notion using our
notation.

Let $q$ be a query, let $\Gamma$ be a set of "primary key constraints". Given an
"atom" $A$ of $q$
let
\AP
\begin{align*}
  \intro*\Aplus \defeq \{ B \textrm{ "atom" of } q \mid{}& \textrm{there exists a } 
   \textrm{"$\Gamma$-derivation" }X~B_1\ldots B_n \\
  &\textrm{where } X = \key(A)\text{, } B_n = B \textrm{, and for all }i,~B_i \ne A \}  
\end{align*}

\AP
Let $\vars(\Aplus) =~\bigcup_{B \in \Aplus} \vars(B)$.
Given two "atoms" $A$ and $B$ of $q$ we say that $A$ ""attacks"" $B$ 
if there exists a sequence $F_0,F_1,\dotsc, F_n$ of "atoms" of $q$
and $x_1,x_2,\dotsc,x_n$ of variables not in $\vars(\Aplus)$ such that $A=F_0, B=F_n$ and
for all $i>0$, $x_i$ is a variable occurring both in $F_{i-1}$ and $F_i$.
\AP
The "attack"  from $A$ to $B$  is said to be ""weak@weak attack"" if $B$ is
"$\Gamma$-determined" by $A$. The ""attack graph"" of $q$ and $\Gamma$ is the graph
whose vertices are the "atoms" of $q$ and whose edges are the "attacks". \AP
A cycle in
this graph is ""weak@weak cycle"" if all the "attacks" involved are "weak@weak attack", otherwise it is a \reintro{strong cycle}.

The dichotomy result of~\cite{DBLP:journals/tods/KoutrisW17} can be stated as:

\begin{thmC}[{\cite[Theorem~3.2]{DBLP:journals/tods/KoutrisW17}}]\label{thm-dichotomy-sjf}
  Let $q$ be a "self-join-free query" and $\Gamma$ a set of "primary key
  constraints". If every cycle in the "attack graph" of $q$ and $\Gamma$ is "weak@weak cycle", then $\certain(q)$ can be computed in polynomial time; otherwise $\certain(q)$ is \conp-complete.
\end{thmC}

We prove that if the "attack graph" of $q$ and $\Gamma$ contains only
"weak cycles", then $\HP$ holds. It then follows from \Cref{thm-dichotomy-sjf}
then whenever $q\not\models\HP$ then $\certain(q)$ is \conp-hard as desired:

\begin{thm}\label{thm-sjf-lowerbound}
Assume $q$ is a "self-join-free query"  and $\Gamma$ a set of "primary key
  constraints". If $q\not\models\HP$, then $\certain(q)$ is \conp-hard.
\end{thm}

In view of \Cref{thm-dichotomy-sjf}, the proof of \Cref{thm-sjf-lowerbound} is
an immediate consequence of the following lemma.

\begin{lemma}\label{lemma-sjf-lowerbound}
Assume $q$ is a "self-join-free query"  and $\Gamma$ a set of "primary key
  constraints". If the "attack graph" of $q$ and $\Gamma$ contains only "weak
  cycles" then $q\models \HP$. 
\end{lemma}
\begin{proof}\label{proof:thm-sjf-lowerbound}
\AP
Let $\mathcal{X}$ be the set of all the
strongly connected components of the "attack graph" of $q$ and $\Gamma$. 
We
define the ""core graph"" of $q$ and $\Gamma$ as the directed graph whose
vertices are the elements of $\mathcal{X}$ and there is an edge from $S$ to
$S'$ if $S$ contains an "atom" $A$ "attacking" an "atom" $B$ of $S'$.  Note that, by
definition, a "core graph" is always acyclic. Let $\tau$ be any topological
ordering of this graph, that is, $\tau=S_1,\dotsc,S_n$ is an ordering on $\mathcal{X}$, and for every $i,j$, if there is an edge from $S_i$ to $S_j$ in the "core graph", then $i<j$.   Note that
  since the "attack graph" contains only "weak cycles", any two facts belonging to
  the same strongly connected component are mutually "$\Gamma$-determined". Hence,
  every $S \in \mathcal{X}$ is a "stable set".  In particular $\tau$ is a 
  "$\Gamma$-sequence". We claim that $q\models \HPtau$.
  
To prove $q\models \HPtau$, we need to show that $q\models \HPtaui \tau i$ for every $i<n$. So, fix some $i$ and consider an arbitrary "atom" $A$ of $S_{i+1}$.
Let $D$ be some "database" with "repair" $r$ such that
$r\models q(\balpha a \bbeta)$, where $a$ "matches" $A$ and $\balpha$ "matches"
$S_{\leq i}$. Let $\mu$ be the valuation of the variables of $q$ witnessing the
"solution" $\balpha a \bbeta$. By abuse of notation, for each "atom" $B$ of $q$, we
write $\mu(B)$ to denote the "fact" of the "database" witnessing the "solution" for
the relation symbol of $B$. In particular, $a=\mu(A)$.

Let $a'\simblock a$, $r'=r[a\to a']$, and assume that
$r'\models q(\balpha' a' \bbeta')$ as witnessed by the valuation $\mu'$. We need to show that $r'\models
q(\balpha a' \bdelta)$ for some $\bdelta$.
Notice that the hypotheses of \Cref{lem-determinacy} are satisfied by $r$,
$r'$, $\mu$, $\mu'$ and there is a "$\Gamma$-derivation" from $X=key(A)$ to every
"atom" in $\Aplus$. Hence, from \Cref{lem-determinacy} it follows that 
\begin{align}
  \text{for any
  variable $x \in \vars(\Aplus)$ we have $\mu(x)=\mu'(x)$. }
  \label{eq:mu-mu'} \tag{$\dag$}
\end{align}
To show $r'\models
q(\balpha a' \bdelta)$, we define a satisfying valuation
$\nu$ as follows:
If $x$ is a variable occurring in an "atom" $B\neq A$ that is not "attacked" by
$A$, then set $\nu(x)=\mu(x)$; otherwise, set $\nu(x)=\mu'(x)$. We show that
$\nu$ witnesses the "solution" $\balpha a' \bdelta$ in $r'$. 

For this it suffices to show that for all "atom" $B$ of $q$, $\nu(B)$ is a fact
of $r'$. If $B$ is different from $A$ and not "attacked" by $A$ then this is clear
as $\nu(B)=\mu(B)$ and $\mu(B)$ belongs to both $r$ and $r'$. If $B$ is $A$ or
"attacked" by $A$ then we show that $\nu(B)=\mu'(B)$. To see this consider a
variable $x$ occurring in $B$ such that $\nu(x)=\mu(x)$. By definition of
$\nu$, this is because $x$ also belongs to some "atom" $C$ that is not "attacked"
by $A$. Hence, by definition of "attack", this implies that $x\in\vars(\Aplus)$. By
\eqref{eq:mu-mu'} this implies that $\mu(x)=\mu'(x)$ as desired.
\end{proof}

\subsection{The case of all self-join-free queries falsifying \texorpdfstring{\HP}{PCond}: \texorpdfstring{$\Cqk$}{Cert\_k} fails}
\AP
\label{sec:caseanyquery}

We finally reduce the case of $\qFive$ to any arbitrary query falsifying $\HP$. Let $q$ be a "self-join-free query" such that
$q\not\models\HP$. We show that for all $k$, if $\certain(q)$ is equivalent to
$\Cqk(q)$ then $\certain(\qFive)$ is equivalent to $\Cqk(\qFive)$, a contradiction with \Cref{thm-q5-cqk}.
The following is the analog of \Cref{theorem-UnconditionX+-Lowerbound-sjf}.

\begin{prop}
\label{proposition-q5toq}
For every "database" $D$ over the "signature" of $\qFive$ we can construct a
"database" $D'$ over the "signature" of $q$ such that:
\begin{enumerate}
\itemAP  if $D\models \certain(\qFive)$  then $D'\models \certain(q)$; \label{proposition-q5toq:1}
\itemAP  For every $k\ge 2$, if $D'\models \Cqk(q)$ then 
  $D\models \Cqk(\qFive)$. \label{proposition-q5toq:2}
\end{enumerate}
\end{prop}
\begin{proof}

  The construction of $D'$ from $D$ is actually taken from~\cite[proof of Theorem $6.1$]{DBLP:journals/tods/KoutrisW17}.
Since $q\not\models\HP$, it follows from \Cref{thm-sjf-lowerbound} that the
  "attack graph" of $q$ contains a "strong cycle". Then the following result follows from~\cite[proof of Theorem $6.1$]{DBLP:journals/tods/KoutrisW17}.
  
  \begin{clmC}[{\cite[proof of Theorem $6.1$]{DBLP:journals/tods/KoutrisW17}}]\label{lemma-conp-reduc-KW}
Let $q$ be a "self-join-free" query such that the "attack graph" of $q$ contains a strong cycle. Then  $q$ contains two "atoms" $A_1$ and $A_2$ such that for every "database"  $D$ over the "signature" of $\qFive$ there is  a "database" $D'$
over the "signature" of $q$ and functions $f$ and $g$ such that:
\begin{enumerate}[(i)]
\item If $D\models \certain(\qFive)$, then $D'\models
  \certain(q)$. \label{lemma-conp-reduc-KW:certain}

  \itemAP $f$ is a
  bijection from the $R_1$-"facts" of $D$ to the $A_1$-"facts" of $D'$ such
  that $u\simblock v$ iff $f(u)\simblock
  f(v)$. \label{lemma-conp-reduc-KW:bijection} 
  
  \itemAP $g$ is a function that maps a pair of the form $(C,u)$ where $C$ is an atom of $q$ such that $C\ne A_1$ and $u$ is an $S_1$-"fact" of $D$  to a $C$-"fact" of $D'$ such that:
  \label{lemma-conp-reduc-KW:gbijection}
  \begin{enumerate}
    \item $\set{u \mapsto v \mid g(A_2,u)=v}$ is a bijection from $S_1$-"facts" to $A_2$-"facts" such that $u\simblock v$ if{f} $g(A_2,u)\simblock g(A_2,v)$; 
    \item for every "atom" $C\not\in\set{A_1,A_2}$ and $S_1$-"fact" $u$ in $D$
      the block of $g(C,u)$ contains only one "fact" in $D'$.
        \end{enumerate}
  
\itemAP If $f(u)$ and $g(A_2,v)$ are part of a "solution" to $q$ in $D'$ then
  $D\models \qFive(uv)$.\label{lemma-conp-reduc-KW:last}
\end{enumerate}
  \end{clmC}

  Given a "database" $D$ over the "signature" of $\qFive$, let $D'$ be the
  "database" constructed in \Cref{lemma-conp-reduc-KW}. 
  Note that
  \Cref{lemma-conp-reduc-KW}-\ref{lemma-conp-reduc-KW:certain} already  proves our first item \eqref{proposition-q5toq:1}. 
  It remains to show that if $D'\models \Cqk(q)$
  then $D\models \Cqk(\qFive)$.  
  
  For any "$k$-set" $S$, let $\hat S=\set{u\mid f(u) \in S} \cup \set{v \mid g(A_2,v) \in S}$. 
  We show that if $S\in\Deltaki(q,D',i)$ then $\hat S \in \Deltaki(\qFive,D,i)$, by induction on $i$.

  For $i=0$, $S$ contains necessarily a "solution" to $q$ and therefore an
  $A_1$-"fact" $u'$ and an $A_2$-"fact" $v'$. Let $u$ and $v$ be such that
  $v'=g(A_2,v)$ and $u'=f(u)$. By construction of $\hat S$, both $u$ and $v$ are in
  $\hat S$. By~\Cref{lemma-conp-reduc-KW}-\ref{lemma-conp-reduc-KW:last} this
  implies that $D\models \qFive(uv)$ and therefore $\hat S\in\Deltaki(\qFive,D,0)$.

  For the induction step, let $B'$ be the "block" witnessing the membership of
  $S$ into $\Deltaki(q,D',i)$.
\begin{itemize}
\item Assume first that $B'$ is a $A_1$-"block". Let
  $B=\set{u \mid f(u)\in B'}$. By
  \Cref{lemma-conp-reduc-KW}-\ref{lemma-conp-reduc-KW:bijection} $B$ is a
  "block" of $D$. Moreover, it witnesses the membership of $\hat S$ in
  $\Deltaki(\qFive,D,i)$. Indeed, consider an element $b\in B$. Let $b'=f(b)$. By
  hypothesis, $S\cup\set{b'}$ contains a "$k$-set" $T$ in
  $\Deltaki(q,D',i-1)$. By induction $\hat T \in \Deltaki(\qFive,D,i-1)$ and
  $\hat T$ is a
  subset of $\hat S \cup \set{b}$. The result follows.

\item Assume now that $B'$ is a $A_2$-"block". We conclude as above using
  $B=\set{v \mid g(A_2,v)\in B'}$ and
  \Cref{lemma-conp-reduc-KW}-\ref{lemma-conp-reduc-KW:gbijection}-(a).

\item Finally, if $B'$ is a $C$-"block" for some "atom"
  $C\not\in\set{A_1,A_2}$, then $B'$ contains only one element by
  \Cref{lemma-conp-reduc-KW}-\ref{lemma-conp-reduc-KW:gbijection}-(b). Therefore, $S$ was actually in $\Deltaki(q,D',i-1)$ and we can conclude by
  induction hypothesis. \qedhere
\end{itemize}
\end{proof}

\begin{thm}\label{thm-nothp-cqk}
  Let $q$ be a "self-join-free query" such that $q\not\models\HP$. Then
  $\certain(q)$  cannot be computed by
  $\Cqk(q)$, for any choice of $k$.
\end{thm}

The proof of \Cref{thm-nothp-cqk} is identical to the proof of \Cref{thm-q5-cqk} using \Cref{proposition-q5toq}
instead of \Cref{theorem-UnconditionX+-Lowerbound-sjf}.


\section{Path queries}
\label{section-path}

The "dichotomy conjecture" has also been shown to hold for "path
  queries"~\cite{DBLP:conf/pods/KoutrisOW21}.  In this section we show that the \Cqk~algorithm works for \ptime solvable path queries and if \Cqk~does not compute $\certain(q)$ for $k= |q|$ then the problem is \conp-hard.
  
  For this section, assume that the "relational signature" $\sigma$ contains
only symbols of arity two and that the set $\Gamma$ of constraints assigns to
each symbol $R$ of $\sigma$ its first component as the "primary key".
Recall that a "path query" is a "Boolean conjunctive query" of the form
$R_1(\underline x_1\,x_2)\land R_2(\underline x_2\, x_3) \land R_3(\underline x_3\, x_4)
\land \cdots \land R_n(\underline x_{n} x_{n+1})$ that may contain "self-join", "ie", $R_i = R_j$ for some $i\ne j$.
Note that a "path query" can be described by a word over the alphabet of relation names of $\sigma$ 
(\eg, the word describing $q$ as $R_1 \dotsb R_n$). 
For simplicity, we will henceforth blur the distinction between "path queries" and words over $\sigma$.

\AP
Following~\cite{DBLP:conf/pods/KoutrisOW21} we define the language $\intro*\Lq$ as the
regular language defined by the following finite state automaton $\intro*\aut$ with
$\epsilon$-transitions\footnote{An $\epsilon$-transition in an automata makes a transition from one state to the other without reading any symbol.}
 (we use $s,t,\ldots$ to denote words over $\sigma$). The
set of states of $\aut$ is the set of all prefixes of $q$, including the
empty prefix $\epsilon$, which is the initial state. There is only one
accepting state, which is $q$. There is a transition reading $R$ from state $s$
to the state $sR$. Moreover, there is an $\epsilon$-transition in $\aut$
from any state $sR$ to any state $tR$ such that $tR$ is a prefix of $s$.

We say that the query $q$ satisfies
\AP$\intro*\FactorCond $ and write $q \models \FactorCond$ if
$q$ is a factor of all the words in the language $\Lq$.

The dichotomy result of~\cite{DBLP:conf/pods/KoutrisOW21} can be formulated as
follows\footnote{\cite{DBLP:conf/pods/KoutrisOW21} provides a much
  finer `tetrachotomy' between \fo, \nl-complete, \ptime-complete and
  \conp-complete. In this section we restrict our attention to the dichotomy between
  \ptime and \conp-complete.}:

\begin{thmC}[{\cite[Theorem~3.2]{DBLP:conf/pods/KoutrisOW21}}]\label{dichotomy-paths}
 Let $q$ be a "path query". If $q \models\FactorCond$, then $\certain(q)$ can be evaluated in \ptime; otherwise,
 $\certain(q)$ is \conp-complete.
\end{thmC}

Thus, $\certain(\qTwop)$ and $\certain(\qThreep)$ for the queries described in \Cref{first-example} are in \ptime and \conp-complete respectively (refer to~\cite{DBLP:conf/pods/KoutrisOW21} for detailed explanation). As in the "self-join-free" case we will show that, for "path queries", there is some
$k$ such that $\Cqk(q)$ captures $\certain(q)$ iff $q \models\FactorCond$. In view of \Cref{dichotomy-paths} this implies that when $\Cqk(q)$ fails to
capture $\certain(q)$, then $\certain(q)$ is \conp-complete.

\subsection{Tractable path queries}\label{section-tracatble-path-queries}

In this section we show the first part of our result: if the "path query" $q$
satisfies $\FactorCond$, then $\certain(q)$ can be computed via \Cqk.

\begin{thm}\label{cqk-path-queries}
  Let $q$ be a "path query" of length $k$. If $q \models \FactorCond$, then $\certain(q)=\Cqk(q)$.
\end{thm}

The rest of this section is devoted to the proof of
\Cref{cqk-path-queries}.
We will make use of the following fixpoint computation introduced
by~\cite[Fig.~5]{DBLP:conf/pods/KoutrisOW21}. 
\knowledgenewrobustcmd{\NqD}{\cmdkl{N}}
\knowledgenewrobustcmd{\NqDi}{\cmdkl{N}}
\knowledgenewrobustcmd{\Nq}{\cmdkl{N}}
\AP
For a fixed "path query" $q$ and
"database" $D$, let $\intro*\NqD(q,D)$ be the set of pairs of the form 
$\langle c,s\rangle$, where $c\in\adom(D)$ and $s$ is a prefix of $q$, computed
via the following fixpoint algorithm.

\begin{description}
  \item[Initialization Step] $\NqD(q,D) \leftarrow \{ \langle c,q\rangle \mid c \in \adom(D)\}$
  \item[Iterative Step] If $s$ is a prefix of $q$, add $\langle c,s \rangle$ to $\NqD(q,D)$ if one of the following holds:
    \begin{enumerate}
      \item
      \label{item-ForwardN} $sR$ is a prefix of $q$ and
      there is a "fact" $R(\underline{c}~a)$ of $D$ such that for every "fact" 
      $R(\underline{c}~b)$ of $D$ we have $\langle b, sR \rangle \in \NqD(q,D)$;
      \item
      \label{item-BackwardN} 
      There is an $\epsilon$-transition from $s$ to $t$ in $\aut$ 
      and there is a "fact" $R(\underline{c}~a)$ of $D$ such that for every "fact" $R(\underline{c}~b)$ of $D$ we have $\tup{b,tR} \in \NqD(q,D)$.
      \end{enumerate}
\end{description}

\AP
Let $\intro*\Nq(q)$ be the set of all "databases" $D$ such that there exists $c\in \adom(D)$ with $\langle c, \epsilon \rangle \in \NqD(q,D)$.

\begin{lemC}[{\cite[(proof of) Lemma 6.4]{DBLP:conf/pods/KoutrisOW21}}]
\label{lemma-Paths P Time}
For every "path query" $q$, if $q \models \FactorCond$, then $\certain(q) = \Nq(q)$.
\end{lemC}

In view of \Cref{lemma-Paths P Time}, 
\Cref{cqk-path-queries} is a direct consequence of the following proposition.

\begin{prop}
\label{prop-N is subseteq of Cqk}
For every "path query" $q$ of length $k$ such that $q \models \FactorCond$, we have $\Nq(q) = \Cqk(q)$.
\end{prop}

 Note that  $\Cqk(q) \subseteq \Nq(q)$ follows from 
  $\Cqk(q) \subseteq \certain(q)$ (\Cref{prop-underapprox})
  combined with
  $\certain(q) = \Nq(q)$ (\Cref{lemma-Paths P Time}).
  So we are left with proving $\Nq(q) \subseteq \Cqk(q)$. Let $D\in \Nq(q)$. We will prove that $D\in \Cqk(q)$.
\medskip

\AP
  For all $l \ge 0$ let $s_l$ be the prefix of $q$ of length $l$ ("ie",
  $s_0=\epsilon$ and $s_k=q$). For every "database" $D$ and "fact"
  $u=R(\underline{a}~b)$ in $D$, let us define $\trace(u) \defeq R$, $\key(u) = a$,
  and $\last(u) = b$. For a sequence of (possibly repeating) "facts"
  $\Pi = u_1,\dotsc, u_l$ of a "database" $D$, we define
  $\intro*\trace(\Pi) \defeq \trace(u_1) \dotsb \trace(u_l) \in \sigma^*$ and, if $\Pi$ is not the empty sequence, we define $\intro*\last(\Pi) \defeq \last(u_l)$. Also we let
  $S_{\Pi} = \{ u_1,\dotsc, u_l\}$ be the \emph{set} of "facts" in the
  sequence. 
  \AP
  Further, $\Pi$ is called a ""valid path"" if the following conditions hold:
  \begin{enumerate}[(i)]
    \item $S_{\Pi}$ is a "partial repair" of $D$, 
    \item for all $i<l$ we have
    $\last(u_i) = key(u_{i+1})$, and 
    \item $\trace(\Pi)$ is a prefix of $q$. 
  \end{enumerate}
  In particular, for any "valid path" $\Pi$ of length $l$, we have
    $\trace(\Pi) = s_l$.
  \knowledgenewrobustcmd{\simtrace}{\mathrel{\cmdkl{\approx}}}
  \AP
  Moreover, for any prefix $s_L$ of $q$ we write
  {$\trace(\Pi) \intro*\simtrace s_L$} if there exists a run of the automaton
  $\aut$ on $\trace(\Pi)$ ending in state $s_L$.

  \bigskip
  
  \AP
  For any $D\in \Nq(q)$, let $\intro*\NqDi(q,D,i)$ and $\reintro*\Deltaki(q,D,i)$ be the fixpoint
  computations of $\NqD(q,D)$ and $\Deltak(q,D)$ at step $i$ respectively. To
  prove that $D\in \Cqk(q)$, we will use the following claim:
  \begin{clm}
\label{claim:Induction for N is subseteq of Cqk}
    For all $i\geq 0$, For $c\in \adom(D)$ and for all non-empty prefix $s_l$ of $q$ if
    $\tup{c,s_l}\in \NqDi(q,D,i)$ then for all non-empty "valid path" $\Pi$ where
    $\last(\Pi) = c$ and $\trace(\Pi) \simtrace s_l$, we have
    $S_\Pi \in \Deltaki(q,D,i)$.
  \end{clm}
  Let us show that the claim implies $D \in \Cqk(q)$.  As $D\in \Nq(q)$, there
  exists $c\in \adom(D)$ such that $\tup{c,\epsilon} \in \NqDi(q,D,m)$ for some
  step $m$.  But note that $\tup{c,\epsilon} \in \NqDi(q,D,m)$ can only be produced
  by application of Rule \ref{item-ForwardN} in the \textit{Iteration step}
  (Rule \ref{item-BackwardN} is not possible since $\epsilon$-transitions do
  not start at the state $\epsilon$).  This implies that if $R$ is the first
  relation occurring in $q$ then there exists a "fact" of the form
  $R(\underline{c}~a)$ and for all "facts" of the form $R(\underline{c}~b)$ in
  $D$ we have $\langle b, R\rangle \in \NqDi(q,D,m-1)$.  For each such $b$ we can
  apply the claim with the "valid path" $\Pi = R(\underline{c}~b)$, obtaining
  $R(\underline{c}~b) \in \Deltaki(q,D,m-1)$ for every
  $R(\underline{c}~b)$.  Hence, $\emptyset \in \Deltaki(q,D,m)$ which implies
  $D\in \Cqk(q)$.

The proof of the claim will make use of the following consequence of $q \models \FactorCond$.

\begin{lem}\label{lemma-suffix}
  Let $q$ be a "path query" such that $q \models \FactorCond$. Then for any prefix of $q$ of the form $s_1PRs_2PR'$ where $s_1,s_2\in \sigma^*$ and $R\ne R'$, we have that $s_1P$ is a suffix of $s_1PRs_2P$.
  \end{lem}
  \begin{proof}
  Assume $q=s_1PRs_2PR't$ and consider the word $w = s_1PRs_2PRs_2PR't$. A
  simple observation shows that $w \in \Lq$. Hence,  by hypothesis, $w$ contains $q$ as factor.   
  Let $w(i)$ denote the symbol of $\sigma$ occurring at the $i$-th position of $w$, for any $1 \leq i \leq |w|$.
   Note that by definition of $w$, for all positions $l$ such that $|s_1| < l\le  |s_1|+|s_2|+3$ we have:
  \begin{align}
     w(l) = w(l+|s_2|+2) \tag{$\star$} \label{hyp:suffix}
  \end{align}
  Let $w(i),w(i+1),\dotsc, w(i+n-1)$ be the factor of $w$ that "matches" $q$, and observe that $1 \leq i \leq |w| - |q| +1  = |s_2|+3$. 
  If $q$ is a suffix of $w$, it follows that $s_1P$ is a suffix of $s_1PRs_2P$ and we are done.
  If $q$ is not a suffix of $w$, then $i < |s_2|+3$. 
  Observe that
  $w((i -1) + |s_1 PR| ) = R$
  and
  $w((i -1) + |s_1 PR s_2 PR'| ) = R'$.
  Hence, setting $l = (i -1) + |s_1 PR| \leq |s_1| + |s_2| + 3$, we have
  $w(l) = R$ and
  $w(l+ |s_2|+2) = R'$. But then by \eqref{hyp:suffix} we would obtain $R=R'$, which is in contradiction with our hypothesis.
  \end{proof}

\begin{proof}[Proof of Claim \ref{claim:Induction for N is subseteq of Cqk}]
 The proof is  by induction on $i$. For the
base case $i=0$, we have $\NqDi(q,D,0) = \set{\tup{c,q} : c\in \adom(D)}$.  Note
that $\aut$ has no $\epsilon$-transitions from a prefix of $q$ to
$q$. Hence, for any "valid path" $\Pi=u_1, \dotsc, u_l$ such
that $\trace(\Pi) \simtrace q$ we must have $k=l$ and $\trace(\Pi) = q$. In this case
$(u_1,\dotsc, u_l)$ forms a "solution" to $q$ and hence
$S_\Pi \in \Deltaki(q,D,0)$.

For the induction step, let $\Pi = u_1,\dotsc, u_l$ be any "valid path" such
that $\trace(\Pi) \simtrace s_L$ and $\last(u_l) = c$. Assuming
$\tup{c,s_L} \in \NqDi(q,D,i+1)$, we will prove $S_\Pi \in \Deltaki(q,D,i+1)$. If
$\tup{c,s_L} \in \NqDi(q,D,i)$ then by induction hypothesis we have
$S_\Pi \in \Deltaki(q,D,i)$ and we are done since
$\Deltaki(q,D,i) \subseteq \Deltaki(q,D,i+1)$. So assume that $\tup{c,s_L}$
is newly added into $\NqDi(q,D,i+1)$.

By definition of the iterative step (regardless of which rule is applied),
there is a state $s'$, a partial run of $\aut$ from state $s_L$ to state $s'$
reading $R \in \sigma$, and a "fact" $R(\underline{c}~a)$ of $D$ such that for
every "fact" $R(\underline{c}~b)$ of $D$ we have
$\langle b, s' \rangle \in \NqDi(q,D,i)$.

Let $s_l=\trace(\Pi)$. Since $\Pi \simtrace s_L$, there is a run of $\aut$ on $\trace(\Pi)$ that ends at $s_L$. We consider three cases depending on the successor of $s_l$ in $q$.
\begin{enumerate}
\item Case $s_l = q$. This case is similar to the base case. Since
  $(u_1,\dotsc, u_l)$ forms a "solution" to $q$, we have
  $S_\Pi \in \Deltaki(q,D,0) \subseteq \Deltaki(q,D,i+1)$.
  \item \label{case:sk+1hasR} Case $s_{l+1} = s_lR$.  
  
     Note that if there is already a "fact" of the form $u_i = R(\underline{c}~b)$ in
    $\Pi$ then the new path $\Pi' = u_1,u_2,\dotsc, u_l, u_i$ is also a "valid
    path" where $\last(\Pi') = b$.   Otherwise, for every "fact" $u_i$ in $\Pi$ if $\trace(u_i) = R$ then
    $\key(u_i) \ne c$.  Take any arbitrary "fact" of the form
    $v = R(\underline{c}~b)$ of $D$. The new path
    $\Pi' = u_1,u_2,\dotsc, u_l, v$ is also a "valid path". 
    
    So in both cases we have
    $\last(\Pi') = b$. Also, since there is a run of $\aut$ on
    $\trace(\Pi)$ that ends at $s_L$, there is a run of $\aut$ on
    $\trace(\Pi')$ that ends at $s'$. Hence, $\trace(\Pi') \simtrace s'$.
    
      Thus,
    by induction hypothesis, if there is already a "fact" of the form $u_i = R(\underline{c}~b)$ in
    $\Pi$ then $S_{\Pi'} \in \Deltaki(q,D,i)$. But since $u_i$
    is already present in $\Pi$, we obtain $S_\Pi = S_{\Pi'}$, and therefore
    $S_{\Pi} \in \Deltaki(q,D,i) \subseteq \Deltaki(q,D,i+1)$.
  
  Otherwise,
    by induction hypothesis we have
    $S_{\Pi'} = S_{\Pi} \cup \{R(\underline{c}~b)\} \in
    \Deltaki(q,D,i)$. Since this holds for any "fact" of the form
    $R(\underline{c}~b)$ we obtain, by definition of $\Deltaki(q,D)$, that
    $S_{\Pi} \in \Deltaki(q,D,i+1)$.

%
%
    

  \item Case $s_{l+1} = s_lR'$ for some $R' \ne R$.

    Since $\trace(\Pi) = s_l$, there is a run of $\aut$ on $s_l$
    that ends in state $s_L$. Let $P$ be the last symbol of $s_L$ (since $s_L$ is non-empty by
    assumption). Observe that, by definition of $\aut$, a run on a word cannot
    end at $s_L$ unless the word also ends with $P$;
    therefore, $s_l$ has to end with $P$.  Altogether we have
  $s_{l+1} = \overbrace{\underbrace{wP}_{s_L}R \, w'P}^{s_l}R'$ for some $w,w' \in \sigma^*$ and $R\ne R'$.
  Applying Lemma \ref{lemma-suffix}, we obtain that $s_L$ is a suffix of
  $s_l$.  Let $\Pi' = u_{l-L+1},\dotsc, u_l$ be the suffix of $\Pi$ such that
  $\trace(\Pi') = s_L$. Note that $\Pi'$ is a "valid path" and
  $S_{\Pi'} \subseteq S_{\Pi}$. Hence, it is sufficient to prove that
  $S_{\Pi'} \in \Deltaki(q,D,i+1)$. We are then in the situation of the
  already treated Case~\ref{case:sk+1hasR} above, since $s_{L+1}=s_L R$. Hence,
  $S_{\Pi'} \in \Deltaki(q,D,i+1)$.\qedhere
\end{enumerate}
\end{proof}


\subsection{Inexpressiblity results for path queries}
\label{path-fails-cqk}
In this section we show  our second main result for "path queries": if there
exists a word $w\in\Lq$ such that $q$ is not a factor of $w$, then for all $k$,
$\Cqk$ fails to capture $\certain(q)$.

Together with \Cref{cqk-path-queries} this gives a complete characterization of
when our fixpoint algorithm computes $\certain(q)$ for "path queries" $q$.
Note that from \Cref{dichotomy-paths} we already know that for such queries, $\certain(q)$ is \conp-complete. So
assuming $\ptime\neq \conp$ no polynomial time algorithm can compute
$\certain(q)$. Our result is unconditional but only applies for the \Cqk~algorithm.

\begin{thm}\label{theorem-UnconditionX+-Lowerbound-path}
  Let $q$ be a "path query" such that $q \not\models \FactorCond$. Then for all $k$, $\certain(q)$ is not computed by
  $\Cqk(q)$.
\end{thm}

From the existence of a word in $\Lq$ that does not have $q$ as a factor, it follows that
$q$ is of the form $uTvTw$ but $q$ is not a factor of
$uTvTvTw$~\cite[Lemma 5.4]{DBLP:conf/pods/KoutrisOW21}. The proof of
\Cref{theorem-UnconditionX+-Lowerbound-path} is again a reduction to
the case of the query $\qFour$ (refer to~\Cref{thm-X+lowerBound}) using the following proposition.

\begin{prop}
\label{proposition-UnconditionX+-Lowerbound-path}
Let $q$ be a "path query" of the form $uTvTw$ and $q$ is not a factor of $uTvTvTw$.
For every "database" $D$ over the "signature" of $\qFour$ we can construct a "database" $D'$
over the "signature" of $q$ such that:
\begin{enumerate}
\item If $D\models \certain(\qFour)$  then $D'\models \certain(q)$.
\item For every $k\ge 2$,  if $D'\models \Cqk(q)$ then
  $D\models \Cqkp(\qFour)$.
\end{enumerate}
\end{prop}

Before we prove the proposition we show why it implies \Cref{theorem-UnconditionX+-Lowerbound-path}.

\begin{proof}[Proof of \Cref{theorem-UnconditionX+-Lowerbound-path}]
Recall that, since $q$ is not a factor of every word in $\Lq$, we have that $q$
is of the form $uTvTw$ but $q$ is not a factor of $uTvTvTw$~\cite[Lemma 5.4]{DBLP:conf/pods/KoutrisOW21}.
  
  Assume towards a contradiction that there is a $k$ such that
  $\certain(q)=\Cqk(q)$.  We then show that $\certain(\qFour)=\Cqkp(\qFour)$,
  contradicting \Cref{thm-X+lowerBound}.
  To prove this it is enough to show that if a "database" $D$ is such that
  $D\models\certain(\qFour)$ then $D\models\Cqkp(\qFour)$.

  Consider such a "database" $D$ and let $D'$ be the "database" constructed by
  \Cref{proposition-UnconditionX+-Lowerbound-path}. From our hypothesis on $D$
  and the first item of \Cref{proposition-UnconditionX+-Lowerbound-path}, it
  follows that $D'\models \certain(q)$. From our hypothesis on $q$ it follows
  that $D'\models\Cqk(q)$. From the second item of
  \Cref{proposition-UnconditionX+-Lowerbound-path}, it follows that
  $D\models\Cqkp(\qFour)$, and the desired contradiction.  
\end{proof}

We now turn to the proof of \Cref{proposition-UnconditionX+-Lowerbound-path}.

\begin{proof}[Proof of \Cref{proposition-UnconditionX+-Lowerbound-path}]

  From the hypotheses we have that $u\ne \epsilon$ (otherwise $q = TvTw$ is a
  factor of $TvTvTw$). So let $u = A_0 u'$ and let us denote $q$ as
  \begin{multline*}
    A_0(\underline{x}~x_1) ~u'~ T(\underline{y}~y_0) ~v~ T(\underline{z}~z_0)  ~w , \quad \text{where}
    \begin{cases}
      u'&=A_1(\underline{x_1}~x_2)\ldots A_i(\underline{x_i}~y),\\
      v&=A'_0(\underline{y_0}~y_1) A'_1(\underline{y_1}~y_2)\ldots A'_j(\underline{y_j}~z),\\
      w&= A''_0(\underline{z_0}~z_1) A''_1(\underline{z_1}~z_2)\ldots
      A''_k(\underline{z_k}~z_{k+1}),
    \end{cases}
  \end{multline*}
%
and $u',v,w$ are possibly empty. Note that if $v = \epsilon$ then $y_0 = z$.

\medskip Let $D$ be a "database" for $\qFour$. Consider the "solution
graph" $\solgraph$ of $D$. Recall that every connected component in $\solgraph$
is always a clique of size less than or equal to $3$ and that every "fact" of $D$
can be part of exactly one maximal clique (\cf~\Cref{rk:properties-solgraph-q4}).

Let $B_1,B_2\ldots B_m$ be the set of all "blocks" of $D$ and
$C_1,C_2\ldots C_l$ be the set of maximal cliques in the "solution graph"
$\solgraph$. Notice that a clique may contain two "facts" in the same "block". For instance
the "facts" $R(a\, ba)$ and $R(a\, ab)$ form a "solution" to $\qFour$.  \AP
Recall that a clique $C_t$ is called \reintro{self-loop} if it contains only
one "fact" $h$ such that $D\models \qFour(hh)$.

Define $D'$ as follows:

If the "block" $B_s$ of $D$ contains a "fact" in a non-"self-loop" clique $C_t$ then we add the following facts to
$D'$, (all domain elements are fresh):
\begin{multline*}
  A_0(\underline{\alpha^s}~\alpha_{1}^{st})~
  \underbrace{A_1(\underline{\alpha^{st}_{1}}~\alpha^{st}_{2})~\ldots
    A_i(\underline{\alpha^{st}_{i}}~\beta^{st})}_{u'}\\
    T(\underline{\beta^{st}}~\beta_0^{st})\underbrace{A'_0(\underline{\beta_0^{st}}~\beta_1^{st})
    A'_1(\underline{\beta^{st}_1}~\beta_2^{st}) \ldots
    A'_j(\underline{\beta^{st}_j}~\gamma^{st})}_{v}\\
    T(\underline{\gamma^{st}}~\delta_0^{st})\underbrace{A'_0(\underline{\delta_0^{st}}~\delta_1^{st})
    A'_1(\underline{\delta^{st}_1}~\delta_2^{st}) \ldots
    A'_j(\underline{\delta^{st}_j}~\eta^{st})}_{v}
  \\T(\underline{\eta^{st}}~\eta_0^{st})\underbrace{A''_0(\underline{\eta_0^{st}}~\eta_1^{st})
    A''_1(\underline{\eta^{st}_1}~\eta_2^{st}) \ldots
    A''_k(\underline{\eta^{st}_k}~\eta^{st}_{k+1})}_{w}
\end{multline*}

If the "block" $B_s$ of $D$ contains a "fact" in the "self-loop" clique $C_t$ then then we add the following facts to
$D'$.
\begin{multline*}
  A_0(\underline{\alpha^s}~\alpha_{1}^{st})~
  \underbrace{A_1(\underline{\alpha^{st}_{1}}~\alpha^{st}_{2})~\ldots
    A_i(\underline{\alpha^{st}_{i}}~\beta^{st})}_{u'}T(\underline{\beta^{st}}~\beta_0^{st})\underbrace{A'_0(\underline{\beta_0^{st}}~\beta_1^{st})
    A'_1(\underline{\beta^{st}_1}~\beta_2^{st}) \ldots
    A'_j(\underline{\beta^{st}_j}~\gamma^{st})}_{v}\\
    T(\underline{\gamma^{st}}~\gamma_0^{st})\underbrace{A''_0(\underline{\gamma_0^{st}}~\gamma_1^{st})
    A''_1(\underline{\gamma^{st}_1}~\gamma_2^{st}) \ldots
    A''_k(\underline{\gamma^{st}_k}~\gamma^{st}_{k+1})}_{w} 
\end{multline*}
      

Further, for every maximal clique $C_t$ containing at least two facts in
different "blocks", and for every facts $h,g \in C_t$
where $h,g$ are in "blocks" $B_{s_1},B_{s_2}$, with $s_1\neq s_2$, add
the following facts to $D'$:\\

\begin{tabular}{l}
  $A_0(\underline{\theta^{s_1s_2t}}~\theta_{1}^{s_1t})~
    \underbrace{A_1(\underline{\theta^{s_1t}_{1}}~\theta^{s_1t}_{2})~\ldots
      A_i(\underline{\theta^{s_1t}_{i}}~\gamma^{s_1t})}_{u'}$\\
  \hspace{5cm}$T(\underline{\gamma^{s_1t}}~\gamma_0^{s_1t})\underbrace{A''_0(\underline{\gamma_0^{s_1t}}~\gamma_1^{s_1t})
      A''_1(\underline{\gamma^{s_1t}_1}~\gamma_2^{s_1t}) \ldots
      A''_k(\underline{\gamma^{s_1t}_k}~\gamma^{s_1t}_{k+1})}_{w}$\\
  $ A_0(\underline{\theta^{s_1s_2t}}~\theta_{1}^{s_2t})~
    \underbrace{A_1(\underline{\theta^{s_2t}_{1}}~\theta^{s_2t}_{2})~\ldots
      A_i(\underline{\theta^{s_2t}_{i}}~\gamma^{s_2t})}_{u'}$\\
  \hspace{5cm}$ T(\underline{\gamma^{s_2t}}~\gamma_0^{s_2t})\underbrace{A''_0(\underline{\gamma_0^{s_2t}}~\gamma_1^{s_2t})
      A''_1(\underline{\gamma^{s_2t}_1}~\gamma_2^{s_2t}) \ldots
      A''_k(\underline{\gamma^{s_2t}_k}~\gamma^{s_2t}_{k+1})}_{w}$
\end{tabular}

\begin{figure}
  \centering
  \includegraphics[width=\textwidth]{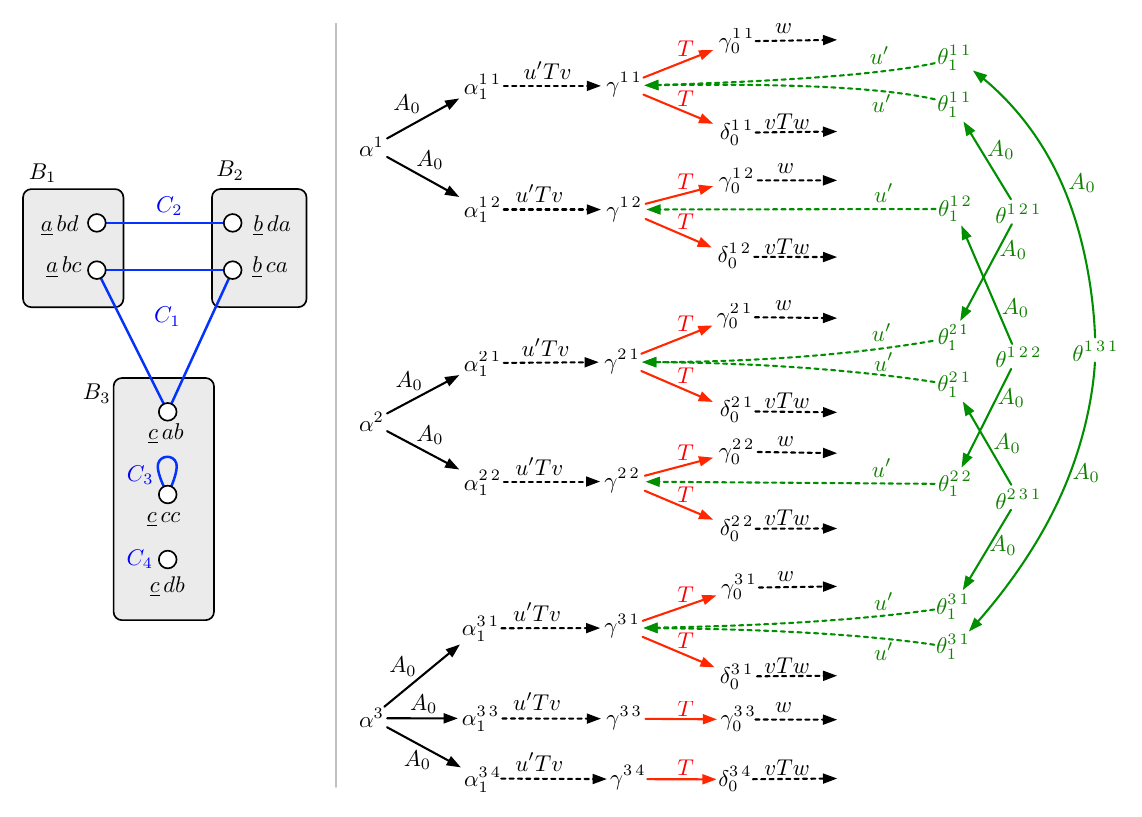}
  \caption{Illustration of the reduction from $\qFour$ to $q$. The black arrows
    starting from the same node correspond to a block of the first kind. The
    red arrows and the green arrows correspond to the blocks of the second and third kind respectively. Some domain elements are represented twice (with the same name, like $\theta^{21}_{1}$) to declutter the figure.}
  \label{fig-path-unconditional}
\end{figure}

This concludes the constructions of $D'$, we refer to
\Cref{fig-path-unconditional} for an illustration. Notice that most "blocks" of
$D'$ have size one except for three kinds. The first kind are $A_0$-"blocks"
with key $\alpha^s$ for some "block" $B_s$ of $D$. These "blocks" are in
one-to-one correspondence with the "blocks" of $D$ and contain as many facts as
there are cliques in $D$ intersecting the "block". The second kind are
$T$-"blocks" whose key is $\gamma^{st}$ for some "block" $B_s$ and maximal
clique $C_t$ where $C_t$ is not a "self-loop". These "blocks" contain at most
two "facts", one that starts a path $TvTw$ and one that starts a path $Tw$. The
third and last kind are $A_0$-"blocks" whose key is $\theta^{s_1s_2t}$ where
two facts from $B_{s_1}$ and $B_{s_2}$, $s_1\neq s_2$, form a solution and
belong to the non-"self-loop" clique $C_t$.  These "blocks" have two "facts",
each of them starting a path $u$ reaching a "block" of the second kind. The
"blocks" of size one play no role in the rest of this proof.

We now prove that $D'$ has the desired properties.

\begin{enumerate}
\item Suppose $D\models \certain(\qFour)$. Now pick any repair $r'$ of
  $D'$. Define the following repair $r$ on $D$:

  For any $A_0$-"block" in $D'$ of the first kind, pick a corresponding "fact" in
  $D$. "ie", if $A_0(\underline{\alpha^s}~\alpha_{1}^{st}) \in r'$ then
  choose $R(\underline{a}~bc)$ for $r$, where $R(\underline{a}~bc)$ is any "fact" in "block" $B_s$ that belongs to the maximal clique $C_t$.

  Now $r$ is a repair of $D$ since every $R$-"block" $B_s$ in $D$ corresponds to
  an $A_0$-"block" of the first kind. By assumption $r\models \qFour$ hence there
  exists facts $h,g$ that belong to "blocks" $B_{s_1}$ and $B_{s_2}$ and also
  belong to the same clique $C_t$ such that $h,g\in r$.

  Assume first that $h=g$, hence $s_1=s_2$ and $C_t$ is a "self-loop". By construction of $r$ this implies that
  $A_0(\underline{\alpha_s}~ \alpha_1^{st})$ belongs to $r'$ and by construction of $D'$, as
  $C_t$ is a "self-loop" all repairs of $D$ containing  $A_0(\underline{\alpha_s}~\alpha_1^{st})$ have a solution to $q$.

  Assume now that $h\neq g$. As $r$ is a repair this implies $s_1\neq s_2$.
  This implies that both $A_0(\underline{\alpha^{s_1}}~\alpha_{1}^{s_1t})$ and
  $A_0(\underline{\alpha^{s_2}}~\alpha_{1}^{s_2t})$ are in $r'$.
  Consider now the "block" of the third kind whose key is
  $\theta^{s_1s_2t}$. The repair $r'$ must contain a "fact" from this
  "block". Without loss of generality we assume that it is the "fact" that starts a
  path $u$ reaching $\gamma^{s_1t}$, the key of a "block" of the second
  kind. The repair $r'$ must contain a "fact" from this "block", either starting a
  path $Tw$ or a path $TvTw$. In the first case the query $q$ is true because of
  the path starting with $A_0(\underline{\alpha^{s_1}}~\alpha_{1}^{s_1t})$; in
  the second case the query $q$ is true because of the path starting with the $A_0$-"fact" of key $\theta^{s_1s_1t}$.

\bigskip

\item Assume that for some $k$, $D'\models\Cqk(q)$. We show that
$D\models\Cqkp(\qFour)$.

In order to show this, for every "$k$-set" $S$ of "facts" of $D'$ we relate "$k$-set" $\intro*\assoc S$ of $D$ such that  the following conditions hold:
\begin{enumerate}
\item If $A_0(\underline{\alpha^s} \alpha_1^{st}) \in S$ then $\assoc S$
  contains a "fact" $h$ that belongs to both the "block" $B_s$ and the clique
  $C_t$.
\item If $T(\underline{\gamma^{st}} \delta_0^{st})\in S$ then $\assoc S$
  contains a "fact" $h$ that belongs to both the "block" $B_s$ and the clique
  $C_t$.
\item If $T(\underline{\gamma^{st}} \gamma_0^{st})\in S$ and $h$ is the "fact"
  that belongs to both the "block" $B_s$ and the clique $C_t$ then $\assoc S$
  contains one "fact" $g$ that belongs the clique $C_t$ such that
  $D\models \qFour(gh)\lor \qFour(hg)$. By construction this is always possible.
\item If $A_0(\underline{\theta^{s_1s_2t}} \theta_1^{st}) \in S$. Recall that
  this can only happen if $s_1\neq s_2$ and $s=s_1$ or $s=s_2$. Let $h_1$ be
  one "fact" that belongs to both the "block" $B_{s_1}$ and the clique $C_t$,
  and $h_2$ be one "fact" that belongs to both the "block" $B_{s_2}$ and the clique $C_t$.
  Then $\assoc S$ contains $h_1$ if $s=s_2$ and $h_2$ if $s=s_1$.
\end{enumerate}

Note that for every "$k$-set" $S$ of "facts" of $D'$ there may be several  "$k$-set" $\assoc S$ of $D$ that satisfy the conditions.
\end{enumerate}

We claim that if $S \in \Delta_k(q,D')$ then all associated $\assoc S$ are
in $\Delta^+_k(\qFour,D)$. In particular, as $D'\models\Cqk(q)$ then by
definition the empty set is in $\Delta_k(q,D')$, and therefore it is also in
$\Delta^+_k(\qFour,D)$, hence $D\models\Cqkp(\qFour)$ as desired.

The proof of the claim is by induction on the iteration where $S$ is added to
$\Delta_k(q,D')$.

The base case is when $S\in\Delta_k(q,D',0)$. Then $S$ contains a solution in
$q(D')$. By construction, as $q$ is not a factor of $uTvTvTw$, this can only
happen if $S$ contains both $A_0(\underline{\alpha^{s}} \alpha_1^{st})$ and
$T(\underline{\gamma^{st}} \gamma_0^{st})$ or $S$ contains both
$A_0(\underline{\theta^{s_1s_2t}} \theta_1^{st})$ and
$T(\underline{\gamma^{st}} \delta_0^{st})$. In the first case any associated
$\assoc S$ by definition contains two "facts" $h$ and $g$ (possibly equal) of
$D$ such that $h$ is in the "block" $B_s$ and clique $C_t$ and $g$ also belongs
to the clique $C_t$ and $D\models \qFour(hg) \lor \qFour(gh)$ (by properties
$(a)$ and $(c)$).  In either case, we have
$\assoc S \in\Delta^+_k(\qFour,D,0)$.

In the second case, by definition, all associated $\assoc S$ contain two
distinct "facts" $h$ and $g$ of $D$ such that $h$ is in the "block" $B_{s_1}$
and clique $C_t$ and $g$ belongs to the ``block'' $B_{s_2}$ and also to the
clique $C_t$ (using properties $(b)$ and $(d)$). In particular
$D\models \qFour(hg) \lor \qFour(gh)$ and $\assoc S \in\Delta^+_k(\qFour,D,0)$.

Assume now that $S\in\Delta_k(q,D',i)$ for some $i>0$. By definition this is
because there exists a "block" $B'$ of $D'$ such that for all $b\in B'$,
$S\cup\set{b}$ contains a $k$-set in $\Delta_k(q,D',i-1)$. By induction, this
implies that for all $b' \in B'$, $\assoc{S\cup\set{b'}}$ contains a $k$-set in
$\Delta^+_k(\qFour,D,i-1)$. We do a case
analysis depending on $B'$.
\begin{enumerate}
\item $B'$ is a "block" of the first kind: it contains all elements
  $A_0(\underline{\alpha^s} \alpha_1^{st})$ for some "block" $B=B_s$. Then $B$ can be used to
  show that $\assoc S$ belongs to $\Delta^+_k(\qFour,D,i)$. Indeed consider $b\in
  B$, let $t$ be such that $b\in C_t$ and consider $b'=A_0(\alpha^s
  \alpha_1^{st})$. From the fact that $\assoc{S\cup\set{b'}}$ contains a
  "$k$-set" in $\Delta_k(\qFour,D,i-1)$ it follows by induction that
  $\assoc{S}\cup\set{b}$ contains a "$k$-set" in $\Delta^+_k(\qFour,D,i-1)$.
 Hence, $\assoc S\in\Delta^+_k(\qFour,D,i)$ as desired.

\item $B'$ is a "block" of the second kind: it contains
  $T(\underline{\gamma^{st}} \gamma_0^{st})$ and
  $T(\underline{\gamma^{st}} \delta_0^{st})$. Let $h$ be a "fact" of $D$ in
  "block" $B_s$ and clique $C_t$ and let $g,p$ be the other "facts" in $C_t$
  (possibly $g=p$, but as $C_t$ cannot be a "self-loop" we assume $g$ and $h$
  to be in distinct "blocks"). Notice that by definition
  $\{T(\underline{\gamma^{st}} \gamma_0^{st})\}$ is associated to both $\{g\}$
  and $\{p\}$ while $\{T(\underline{\gamma^{st}} \delta_0^{st})\}$ is
  associated to $\{h\}$. Now since
  $S \cup \{T(\underline{\gamma^{st}} \gamma_0^{st})\} \in \Delta_k(q,D,i-1)$
  we have by induction, $\hat{S} \cup \{g\} \in \Delta^+_k(\qFour,D,i-1)$ and
  $\hat{S} \cup \{p\} \in \Delta^+_k(\qFour,D,i-1)$. Further,
  $S \cup \{T(\underline{\gamma^{st}} \delta_0^{st})\} \in \Delta_k(q,D,i-1)$
  implies that $\hat{S} \cup \{h\} \in \Delta^+_k(\qFour,D,i-1)$. Thus the
  clique $C_t$ is a witness for $\assoc S \in\Delta^+_k(\qFour,D,i)$ by
  application of the new rule.
  

\item $B'$ is a "block" of the third kind: it contains
  $A_0(\underline{\theta^{s_1s_2t}} \theta_1^{s_1t})$ and
  $A_0(\underline{\theta^{s_1s_2t}} \theta_1^{s_2t})$.  Let $h,g$ and $p$ be the "facts" that
  are in clique $C_t$ and respectively in the "blocks" $B_{s_1}$ and $B_{s_2}$
  and $B_{s_3}$ (possibly $g=p$). Notice that by definition
  $\{A_0(\underline{\theta^{s_1s_2t}} \theta_1^{s_1t})\}$ is associated to both $g$
  and $p$, while $\{A_0(\underline{\theta^{s_1s_2t}} \theta_1^{s_2t})\}$ is associated
  to $h$ and $p$.
  As in the previous case, we obtain by induction that
  $\assoc S \cup \set{x} \in \Delta_k(\qFour,D')$ for all
  $x\in\set{h,g,p}$. Therefore, the clique $C_t$is a witness for  $\assoc S \in\Delta^+_k(\qFour,D,i)$ by application of the new rule.

\item $B'$ is a "block" of size one, "ie", containing one "fact" $b'$. By induction hypothesis
  $\assoc{S\cup\set{b'}}$ contains a $k$-set in $\Delta_k(\qFour,D)$. Moreover, by
  construction  $\assoc{S\cup\set{b'}}=\assoc S$. Hence $\assoc S \in\Delta_k(\qFour,D)$ as desired.
\end{enumerate}
This concludes the induction step, hence the proof.
\end{proof}


\section{First order definability}\label{section-FO}

\newcommand\FO{{\sc FO}\xspace}

Let $q$ be a query. We say that $\certain(q)$ is in \FO if there is a
first-order sentence $\varphi$ such that for all database $D$,
$D\models \certain(q)$ iff $D\models \varphi$. For instance, the query $\certain(\qOne)$ is in \FO where $\qOne$ is described in \Cref{first-example}.

The goal of this section is to provide a characterization of queries
whose certainty can be expressed in \FO in terms of the \Cqk~algorithm, assuming $q$ is either "self-join-free"
or a "path query".  Notice that if $\certain(q)$ is in \FO then in particular it can be solved
in AC$_0$ and therefore cannot be \conp-hard. It then follows from our results that it can be solved by
our fixpoint algorithm $\Cqk(q)$ for some $k$.

Recall the definition of $\Cqk(q)$. It computes in an inflationary way a set
$\Deltak(q,D)$ of "$k$-sets" satisfying a certain property, where the property can be
specified in first-order logic. Starting from the "solutions" to $q$ in $D$, at each
step it adds to $\Deltak(q,D)$ a new set of "$k$-sets" satisfying a first-order
property. 
\AP
Let $\intro*\Deltaki(q,D,i)$ be the set of "$k$-sets" computed this way after
$i$ iterations. Hence $\Deltaki(q,D,0)$ the set of "solutions" to $q$ in $D$ and $\bigcup_i \Deltaki(q,D,i)$ is the set of "$k$-sets" obtained
when the fixpoint is reached. By definition, $\Cqk(q)$ returns `yes' if this set
contains the empty set.

As every step can be defined in first-order logic, each set $\Deltaki(q,D,i)$ can be
defined using a first-order formula. 
\AP
Let $\intro*\psii(q)$ be the first-order sentence
such that $D \models \psii(q)$ if and only if $\Deltaki(q,D,i)$ contains the empty set.
It is clearly an
under-approximation of $\Cqk(q)$ and therefore of $\certain(q)$. Whenever 
$\Cqk(q)$ is equivalent to $\psii(q)$ for some $i$ depending only on $q$ and $k$, we say that $\Cqk(q)$ is ""bounded"".\footnote{This notion is sometimes referred to in the literature as ``goal bounded''.}
We show that whenever $\certain(q)$ is in \FO then $\Cqk(q)$, with $k$ the number of atoms of $q$, is bounded and computes $\certain(q)$.

\begin{exa}
  Recall the query $\qOne = R_1(\underline{x} ~ y)\land R_2(\underline{y}~ z)$ from Example~\ref{first-example}. We showed there that $\certain(\qOne)$ can be expressed in \FO.
  We verify that for $k = 2$, for any database $D$, if $D\models \certain(\qOne)$ then $\emptyset \in \Deltak[2](\qOne,D,2)$.
  
  First note that $\Deltak[2](\qOne,D,0)$ contains sets of the form $\{R_1(\underline{a}~b), R_2(\underline{b}~c)\}$ which are the solutions to $\qOne$. Notice that if $\{R_1(\underline{a}~b), R_2(\underline{b}~c)\} \in \Deltak[2](\qOne,D,0)$ then for all $ R_2(\underline{b}~c') \sim  R_2(\underline{b}~c)$ we have $\{R_1(\underline{a}~b), R_2(\underline{b}~c')\} \in \Deltak[2](\qOne,D,0)$. Hence for every $\{R_1(\underline{a}~b), R_2(\underline{b}~c)\} \in \Deltak[2](\qOne,D,0)$ we have $\{R_1(\underline{a}~b)\} \in \Deltak[2](\qOne,D,1)$. 
  
  In particular for every minimal "repair" $r$ if $r \models q(uv)$ then $u\in \Deltak[2](\qOne,D,1)$. Now since $r$ is minimal, for every $u'\sim u$ there exists some $v'\in r$ such that $r[u\to u']\models q(u'v')$. This implies that for every $u'\sim u$ we have $u'\in \Deltak[2](\qOne,D,1)$. Hence $\emptyset \in \Deltak[2](\qOne,D,2)$ as required. 
  
\end{exa}

\subsection{Self-join-free case}

In this case we make use of the following characterization
of~\cite{DBLP:journals/tods/KoutrisW17} based on the notion of "attack graph".

\begin{thmC}[{\cite[Theorem~3.2]{DBLP:journals/tods/KoutrisW17}}]\label{thm-fo-char-attack}
  Let $q$ be a "self-join-free query" and $\Gamma$ a set of "primary key constraints". The "attack graph" of $q$ and $\Gamma$ is acyclic iff
  $\certain(q)$ is in \FO.
\end{thmC}

We obtain the following:

\begin{prop}\label{prop-FO-char}
  Let $q$ be a "self-join-free query" and $\Gamma$ a set of "primary key constraints". Let $k$ be the number of atoms of $q$. The following are equivalent:
  \begin{enumerate}
  \item\label{item-attack} The "attack graph" of $q$ and $\Gamma$ is acyclic.
  \item\label{item-bounded} $\Cqk(q)$ is "bounded" and $\Cqk(q)=\certain(q)$.
  \item\label{item-fo} $\certain(q)$ is in \FO.
  \end{enumerate}
\end{prop}
\begin{proof}
  \proofcase{\eqref{item-bounded} $\Rightarrow$ \eqref{item-fo}}
  This is an immediate
  consequence of "boundedness".

  \proofcase{\eqref{item-fo} $\Rightarrow$ \eqref{item-attack}} This follows from
  \Cref{thm-fo-char-attack}.

  \proofcase{\eqref{item-attack} $\Rightarrow$ \eqref{item-bounded}}
  Assume that the "attack graph" of $q$ and $\Gamma$ is acyclic. Let
  $\tau=A_1,\dotsc, A_k$ be any topological ordering of the atoms of $q$,
  "ie", if there is an attack from $A_i$ to $A_j$ then $i<j$. As argued in the
  proof of \Cref{thm-sjf-lowerbound}, $\tau$ is a "$\Gamma$-sequence" and
  $q\models \HPtau$, therefore $\Cqk(q)=\certain(q)$. Notice that $\tau$ is a "$\Gamma$-sequence" where all "stable
  sets" have size one. We show that whenever we have such a "$\Gamma$-sequence"
  then $\Cqk(q)$ is "bounded". In order to show this, we revisit the proof of
  \Cref{theorem-HP-ptime}. The key property that we used to prove  \Cref{theorem-HP-ptime} was \Cref{induc-prop} showing that if
  $\Ind_{i+1}$ and $\HPtaui \tau i$ holds, then $\Ind_{i}$ also holds, where
  \begin{align*}
    \Ind_i = \text{For all "$i$-minimal" "repair" $s$ and "facts" $\bar u$ s.t.\ $s
    \models \ql i(\bar u)$, we have $\bar u \in \Deltak(q,D)$.}
  \end{align*}

  When all "stable sets" in the "$\Gamma$-sequence" have size one, we can have a stronger
  version of the induction step where $\Ind_i$ becomes
\AP
  \begin{align*}
    \intro*\Indpi = \text{For all "$i$-minimal" "repair" $s$ and "facts" $\bar u$ s.t.\ $s
    \models \ql i(\bar u)$, we have $\bar u \in \Deltaki(q,D,k-i)$}
  \end{align*}
 (recall that $k$ is the number of atoms in $q$), and \Cref{induc-prop} becomes:

  \begin{clm}
\label{induc-prop-fo}
Given $q$, $D$ and a "$\Gamma$-sequence" $\tau$ for $q$ such that all "stable sets"
of $\tau$ have size one. Then for every
$0 \leq i < k$, if $\Indpi[i+1]$ and $\HPtaui \tau i$, then $\Indpi$.
\end{clm}
\begin{proof}[Proof of claim]
  The proof is similar to the proof of \Cref{induc-prop}, but simpler. 

  By means of contradiction, assume that $\Indpi[i+1]$ and $\HPtaui \tau i$ hold but $\Indpi$
  fails.
  By definition, there is a "$i$-minimal" "repair" $s$ and a tuple $\bu$ such that
  $s\models \ql i(\bu)$ but $\bu \not\in\Deltaki(q,D,k-i)$. From 
  \Cref{claim-strong-minimal}, we can assume that $s$ is "strong $i$-minimal".
  As $s\models \ql i(\bu)$, there is a "fact" $a$ "matching" $S_{i+1}$ such that
  $s\models \ql {i+1}(\bu a)$.
  
  As $\bu \not\in\Deltaki(q,D,k-i)$ there is by definition a "fact" $a'\sim a$
  such that $\bu a' \not\in\Deltaki(q,D,k-i-1)$. By \Cref{claim-strong-delta}
  the "repair" $s'=s[a\to a']$ is "strong $i$-minimal" and
  $s'\models q(\bu a' \bbeta)$ for some tuple $\bbeta$. By assumption, the
  "stable set" $S_{i+1}$ of $\tau$ contains only one atom hence  $a'$
  "matches" $S_{i+1}$ hence $s'\models \ql {i+1}(\bu a')$. By $\Indpi[i+1]$ this implies that $\bu a'
  \in \Deltaki(q,D,k-i-1)$, a contradiction.\renewcommand{\qedsymbol}{$\lhd$}
\end{proof}
\renewcommand{\qedsymbol}{\usebox{\lmcsQEDSymbol}}

Therefore, by
\Cref{induc-prop-fo} we conclude that $\Indpi[0]$ holds. This immediately implies
that whenever the empty set is derived, it is derived within $k$-steps. Then $\Cqk(q)$ is equivalent to $\psii[k,k](q)$ and is therefore bounded; this proves Item~\eqref{item-bounded}.
\end{proof}

\subsection{Path queries}

For "path queries" we rely on the following result of~\cite{DBLP:conf/pods/KoutrisOW21}.

\begin{thmC}[{\cite[Theorem~3.2 \& Lemma 7.1]{DBLP:conf/pods/KoutrisOW21}}]\label{dichotomy-paths-fo}
 Let $q$ be a "path query". If $q$ is a prefix of all the words
 in the language $\Lq$, then $\certain(q)$ is in \FO. Otherwise, $\certain(q)$ is \NL-hard under \FO reductions.
\end{thmC}

We obtain the following:
\begin{prop}\label{thm-path-fo}
  Let $q$ be a "path query" of length $k$. The following are equivalent.
  \begin{enumerate}
  \item\label{item-prefix} $q$ is a prefix of all the words in the language $\Lq$
  \item\label{item-bounded-path} $\Cqk(q)$ is bounded and $\Cqk(q)=\certain(q)$.
  \item\label{item-fo-path} $\certain(q)$ is in \FO.
  \end{enumerate}
\end{prop}

To prove the proposition, we use the following result, which is a restatement of \cite[Corollary 5.9]{DBLP:conf/pods/KoutrisOW21} in the case where \Cref{item-prefix} holds.\footnote{In fact, the condition assumed in Corollary 5.9 in \cite{DBLP:conf/pods/KoutrisOW21} is implied by \Cref{item-prefix}.}
\begin{propC}[{\cite[Corollary 5.9]{DBLP:conf/pods/KoutrisOW21}}]\label{prop-path-fo}
  Let $q$ be a "path query" such that $q$ is a prefix of all the words in
  $\Lq$. The following are equivalent for all database $D$.
  \begin{itemize}
  \item $D\models \certain(q)$
  \item there exists a "block" $B$ such that for all repair $r$ of $D$, there
    exists a sequence of facts $\bv$ of $r$ such that $r\models q(a\bv)$, where $a$ is
    the fact of $r$ in $B$.
  \end{itemize}
\end{propC}

\begin{proof}[Proof of \Cref{thm-path-fo}]
  \proofcase{\Cref{item-bounded-path} $\Rightarrow$ \Cref{item-fo-path}} This is a
  consequence of "boundedness".

  \proofcase{\Cref{item-fo-path} $\Rightarrow$ \Cref{item-prefix}} Follows from \Cref{dichotomy-paths-fo}.

  \proofcase{\Cref{item-prefix} $\Rightarrow$ \Cref{item-bounded-path}}.
Assume \Cref{item-prefix} and assume that $D\models \certain(q)$,
we show that $\emptyset \in \Deltaki(q,D,k)$. 
 This is a consequence of the following lemma that we show by induction on
 $i$. Let $B$ be the "block" of $D$ given by \Cref{prop-path-fo}. For a sequence
 of fact $\bv$ we denote by $\bv_i$ the first $i$ facts of the sequence and $\bv[i]$ denotes the $i$-th fact of $\bv$:

 \begin{lem}\label{claim-fo}
 For all $i\in[0,k-1],$ for all repair $r$ and for all $\bv$ if $r\models q(a\bv)$ with $a\in B$ then
 $a\bv_i \in \Deltaki(q,D,k-i-1)$.
\end{lem}
\begin{proof}
 For $i=k-1$, this is clear (since $q(a\bv)$ holds and $\bv = \bv_{k-1}$).
 Assume now the property shown for $i+1$. We show it for $i$.

 Consider a repair $r$. Let $a=B\cap r$. Let $\bv$ be such that $r\models
 q(a\bv)$. Let $b=\bv[i+1]$. Let $b'\sim b$ and $r'=r[b\to b']$. By
 \Cref{prop-path-fo} there exists $\bv'$ such that $r'\models q(a\bv')$. As $q$
 is a "path query" we have $\bv'_i=\bv_i$. By induction we get that $a\bv_i b' \in
 \Deltaki(q,D,k-i-2)$. As $b'$ is arbitrary, the "block" of $b$ witnesses the
 fact that $a\bv_i \in \Deltaki(q,D,k-i-1)$.
\end{proof}

\bigskip

 Consider now an arbitrary "fact" $a\in B$ and any repair $r$ containing $a$. By
 \Cref{prop-path-fo}, since $D\models \certain(q)$, there exists $\bv$ such that $r\models q(a\bv)$. From
 \Cref{claim-fo}, applied with $i=0$, we get that $a\in\Deltaki(q,D,k-1)$.
 This implies that $\emptyset\in\Deltaki(q,D,k)$ as desired.

 We have proved that $D \models \certain(q)$ implies $\emptyset\in\Deltaki(q,D,k)$. On the other hand $\emptyset \in \Deltaki(q,D,k)$ implies $\emptyset \in \Deltaki(q,D)$ which in turn implies $D\models \certain(q)$.
Then we  have $D\models \certain(q)$ iff $\emptyset \in \Deltaki(q,D)$ iff $\emptyset \in \Deltaki(q,D,k)$; the first equivalence implies $\certain(q) =\Cqk(q)$, the second implies that $\Cqk(q)$ is equivalent to $\psii[k,k](q)$, and is therefore bounded.
\end{proof}


\section{Conclusion}
\label{section-Conclusion}

We have presented a simple polynomial time algorithm for certain query
answering over inconsistent "databases" under "primary key constraints".
The query is always certain when the algorithm outputs ``yes'', but it may
produce false negative answers. For "path queries" and "self-join-free" queries
we have characterized the cases when our algorithm computes all certain
answers. In particular, when "certainty" is in polynomial time, the algorithm
correctly computes the answer.
We have also shown that our fixpoint algorithm is bounded if, and only if, the "certainty" of the query
can be expressed in first-order logic.

Throughout the paper we have only considered boolean queries.  The
  algorithm can be extended to the non-boolean setting as follows: if $\bar{x}$
  are the free variables of $q$, we compute $\Delta(q,D,\bar{a})$ as for
  $\Delta(q,D)$ but always assigning $\bar{x}$ to $\bar{a}$. The algorithm
  returns $\bar{a}$ if $\Delta(q,D,\bar{a})$ eventually contains the empty set
  and if this happens $\bar{a}$ is a certain answer. This still takes
  polynomial time in data complexity.  In the presence of constants we can
  use a similar technique where the interpretation of the constants $\bar{c}$
  are fixed to $\bar{a}$ (and we only need one iteration if there are no free
  variables).  Thus, the algorithm can be adapted to the non-boolean setting as
  well.

Even though recent progress has been made, the "Dichotomy Conjecture" remains a challenging problem.
We add to the list of challenges a (decidable) characterization of when our fixpoint
algorithm correctly computes the certain answers. We believe this is an
interesting problem.  

It is interesting to note that a similar fixpoint algorithm can be obtained for other kinds of
constraints. For instance for "key constraints" or ``denial constraints'' as
defined in~\cite{DBLP:journals/iandc/ChomickiM05}  one can define a ``conflict
hypergraph'' where each hyperedge is a minimal set of facts making the
constraint false. A ``repair'' in this context is a maximal independent set of the conflict
hypergraph and certainty can be computed in \conp. Notice that in the case of
"primary key constraints", the conflict hypergraph is just a graph connecting any
two "facts" that share the same key. The connected components of
this graph are then cliques, which correspond to "blocks". The inductive rule of the
fixpoint then produces a new $k$-set $S$ if there is a connected component $C$ of the
conflict hypergraph such that for all facts $b$ of $C$, $S\cup\set{b}$ 
contains a previously produced $k$-set. This clearly generalizes the current definition. However, the properties of
this algorithm under non-primary key constraints are yet to be studied.

It would also be interesting to see if the simplicity of our algorithm can be
combined with an optimal computational cost. For instance, for "self-join-free"
queries satisfying \HP, it is known that the complexity of the certain
answering problem is in \logspace. However our fixpoint algorithm -- which works
for arbitrary queries, with or without self-joins -- cannot be evaluated in
\logspace. It is however plausible that, assuming self-join-freeness, simpler
rules can be used, providing a lower evaluation complexity. This is left for
future work.


\section*{Acknowledgment}
  \noindent This work is supported by ANR QUID, grant ANR-18-CE40-0031.

\bibliographystyle{alphaurl}
\bibliography{biblio-repair}
\end{document}